\documentclass{article}
\usepackage[a4paper,margin=1.3in]{geometry}
\pdfoutput=1
\usepackage[utf8]{inputenc}
\usepackage[english]{babel}
\usepackage[T1]{fontenc}
\usepackage{amsmath,amssymb,amsfonts,thmtools,thm-restate,mathtools,physics, mathrsfs, enumitem}
\usepackage{csquotes}

\usepackage[dvipsnames]{xcolor}
\usepackage[pdfusetitle,bookmarksnumbered,colorlinks,linkcolor=MidnightBlue,citecolor=olive]{hyperref}
\usepackage[capitalise,nameinlink,noabbrev]{cleveref}
\crefname{equation}{Eq.}{Eqs.}
\Crefname{equation}{Eq.}{Eqs.}

\usepackage[style=numeric-comp,
    sortcites=true,
    sorting=none,
    maxcitenames=3,maxbibnames=9,
    date=year, 
    doi=false,isbn=false,url=false,eprint=true]{biblatex}
\DeclareFieldFormat{postnote}{#1} 
\renewbibmacro{in:}{\ifentrytype{article}{}{\printtext{\bibstring{in}\intitlepunct}}} 

\usepackage{tikzit}
\usepackage{tikz, circuitikz}
\tikzstyle{env}=[copoint,regular polygon rotate=0,minimum width=0.2cm, fill=black]

\tikzstyle{probs}=[shape=semicircle,fill=white,draw=black,shape border rotate=180,minimum width=1.2cm]

\tikzstyle{nudge}=[yshift=0.6mm]

\tikzstyle{every picture}=[baseline=-0.25em,scale=0.5]
\tikzstyle{dotpic}=[] 
\tikzstyle{diredges}=[every to/.style={diredge}]
\tikzstyle{math matrix}=[matrix of math nodes,left delimiter=(,right delimiter=),inner sep=2pt,column sep=1em,row sep=0.5em,nodes={inner sep=0pt},text height=1.5ex, text depth=0.25ex]

\tikzstyle{inline text}=[text height=1.5ex, text depth=0.25ex,yshift=0.5mm]
\tikzstyle{label}=[font=\footnotesize,text height=1.5ex, text depth=0.25ex]
\tikzstyle{left label}=[label,anchor=east,xshift=2mm]
\tikzstyle{right label}=[label,anchor=west,xshift=-2mm]

\tikzstyle{braceedge}=[decorate,decoration={brace,amplitude=2mm,raise=-1mm}]
\tikzstyle{small braceedge}=[decorate,decoration={brace,amplitude=1mm,raise=-1mm}]

\tikzstyle{doubled}=[line width=1.6pt] 
\tikzstyle{boldedge}=[doubled,shorten <=-0.17mm,shorten >=-0.17mm]
\tikzstyle{boldedgegray}=[doubled,gray,shorten <=-0.17mm,shorten >=-0.17mm]
\tikzstyle{singleedgegray}=[gray]

\tikzstyle{semidoubled}=[line width=1.4pt] 
\tikzstyle{semiboldedgegray}=[semidoubled,gray,shorten <=-0.17mm,shorten >=-0.17mm]

\tikzstyle{boxedge}=[semiboldedgegray]

\tikzstyle{boldedgedashed}=[very thick,dashed,shorten <=-0.17mm,shorten >=-0.17mm]
\tikzstyle{vboldedgedashed}=[doubled,dashed,shorten <=-0.17mm,shorten >=-0.17mm]
\tikzstyle{left hook arrow}=[left hook-latex]
\tikzstyle{right hook arrow}=[right hook-latex]
\tikzstyle{sembracket}=[line width=0.5pt,shorten <=-0.07mm,shorten >=-0.07mm]

\tikzstyle{causal edge}=[->,thick,gray]
\tikzstyle{causal nondir}=[thick,gray]
\tikzstyle{timeline}=[thick,gray, dashed]

\tikzstyle{cedge}=[<->,thick,gray!70!white]

\tikzstyle{empty diagram}=[draw=gray!40!white,dashed,shape=rectangle,minimum width=1cm,minimum height=1cm]
\tikzstyle{empty diagram small}=[draw=gray!50!white,dashed,shape=rectangle,minimum width=0.6cm,minimum height=0.5cm]

\tikzstyle{dot}=[inner sep=0mm,minimum width=2mm,minimum height=2mm,draw,shape=circle]  
\tikzstyle{Wsquare}=[white dot, shape=regular polygon, rounded corners=0.8 mm, minimum size=3.3 mm, regular polygon sides=3, outer sep=-0.2mm]
\tikzstyle{Wsquareadj}=[white dot, shape=regular polygon, rounded corners=0.8 mm, minimum size=3.3 mm, regular polygon sides=3, outer sep=-0.2mm, regular polygon rotate=180]
\tikzstyle{ddot}=[inner sep=0mm, doubled, minimum width=2.5mm,minimum height=2.5mm,draw,shape=circle]

\tikzstyle{black dot}=[dot,fill=black]
\tikzstyle{white dot}=[dot,fill=white,,text depth=-0.2mm]
\tikzstyle{white Wsquare}=[Wsquare,fill=white,,text depth=-0.2mm]
\tikzstyle{white Wsquareadj}=[Wsquareadj,fill=white,,text depth=-0.2mm]
\tikzstyle{green dot}=[white dot] 
\tikzstyle{gray dot}=[dot,fill=gray!40!white,,text depth=-0.2mm]
\tikzstyle{red dot}=[gray dot] 

\tikzstyle{black ddot}=[ddot,fill=black]
\tikzstyle{white ddot}=[ddot,fill=white]
\tikzstyle{gray ddot}=[ddot,fill=gray!40!white]

\tikzstyle{gray edge}=[gray!60!white]

\tikzstyle{small dot}=[inner sep=0.5mm,minimum width=0pt,minimum height=0pt,draw,shape=circle]

\tikzstyle{small black dot}=[small dot,fill=black]
\tikzstyle{small white dot}=[small dot,fill=white]
\tikzstyle{small gray dot}=[small dot,fill=gray!40!white]

\tikzstyle{very small dot}=[inner sep=0.3mm,minimum width=0pt,minimum height=0pt,draw,shape=circle]

\tikzstyle{very small black dot}=[very small dot,fill=black]
\tikzstyle{very small white dot}=[small dot,fill=white]
\tikzstyle{very small gray dot}=[small dot,fill=gray!40!white]

\tikzstyle{causal dot}=[inner sep=0.4mm,minimum width=0pt,minimum height=0pt,draw=white,shape=circle,fill=gray!40!white]

\tikzstyle{phase dimensions}=[minimum size=5mm,font=\footnotesize,rectangle,rounded corners=2.5mm,inner sep=0.2mm,outer sep=-2mm]
\tikzstyle{dphase dimensions}=[minimum size=5mm,font=\footnotesize,rectangle,rounded corners=2.5mm,inner sep=0.2mm,outer sep=-2mm]

\tikzstyle{white phase dot}=[dot,fill=white,phase dimensions]
\tikzstyle{white phase ddot}=[ddot,fill=white,dphase dimensions]

\tikzstyle{white rect ddot}=[draw=black,fill=white,doubled,minimum size=5mm,font=\footnotesize,rectangle,rounded corners=2.5mm,inner sep=0.2mm]
\tikzstyle{gray rect ddot}=[draw=black,fill=gray!40!white,doubled,minimum size=6mm,font=\footnotesize,rectangle,rounded corners=3mm]

\tikzstyle{gray phase dot}=[dot,fill=gray!40!white,phase dimensions]
\tikzstyle{gray phase ddot}=[ddot,fill=gray!40!white,dphase dimensions]
\tikzstyle{grey phase dot}=[gray phase dot]
\tikzstyle{grey phase ddot}=[gray phase ddot]

\tikzstyle{small phase dimensions}=[minimum size=4mm,font=\tiny,rectangle,rounded corners=2mm,inner sep=0.2mm,outer sep=-2mm]
\tikzstyle{small dphase dimensions}=[minimum size=4mm,font=\tiny,rectangle,rounded corners=2mm,inner sep=0.2mm,outer sep=-2mm]

\tikzstyle{small gray phase dot}=[dot,fill=gray!40!white,small phase dimensions]
\tikzstyle{small gray phase ddot}=[ddot,fill=gray!40!white,small dphase dimensions]

\tikzstyle{small map}=[draw,shape=rectangle,minimum height=4mm,minimum width=4mm,fill=white]

\tikzstyle{cnot}=[fill=white,shape=circle,inner sep=-1.4pt]

\tikzstyle{asym hadamard}=[fill=white,draw,shape=NEbox,inner sep=0.6mm,font=\footnotesize,minimum height=4mm]
\tikzstyle{asym hadamard conj}=[fill=white,draw,shape=NWbox,inner sep=0.6mm,font=\footnotesize,minimum height=4mm]
\tikzstyle{asym hadamard dag}=[fill=white,draw,shape=SEbox,inner sep=0.6mm,font=\footnotesize,minimum height=4mm]

\tikzstyle{hadamard}=[fill=white,draw,inner sep=0.6mm,font=\footnotesize,minimum height=4mm,minimum width=4mm]
\tikzstyle{small hadamard}=[fill=white,draw,inner sep=0.6mm,minimum height=1.5mm,minimum width=1.5mm]
\tikzstyle{small hadamard rotate}=[small hadamard,rotate=45]
\tikzstyle{dhadamard}=[hadamard,doubled]
\tikzstyle{small dhadamard}=[small hadamard,doubled]
\tikzstyle{small dhadamard rotate}=[small hadamard rotate,doubled]
\tikzstyle{antipode}=[white dot,inner sep=0.3mm,font=\footnotesize]

\tikzstyle{scalar}=[diamond,draw,inner sep=0.5pt,font=\small]
\tikzstyle{dscalar}=[diamond,doubled, draw,inner sep=0.5pt,font=\small]

\tikzstyle{small box}=[rectangle,inline text,fill=white,draw,minimum height=5mm,yshift=-0.5mm,minimum width=5mm,font=\small]
\tikzstyle{small gray box}=[small box,fill=gray!30]
\tikzstyle{medium box}=[rectangle,inline text,fill=white,draw,minimum height=5mm,yshift=-0.5mm,minimum width=8mm,font=\small]
\tikzstyle{square box}=[small box] 
\tikzstyle{medium gray box}=[small box,fill=gray!30]
\tikzstyle{semilarge box}=[rectangle,inline text,fill=white,draw,minimum height=5mm,yshift=-0.5mm,minimum width=12.5mm,font=\small]
\tikzstyle{large box}=[rectangle,inline text,fill=white,draw,minimum height=5mm,yshift=-0.5mm,minimum width=15mm,font=\small]
\tikzstyle{large gray box}=[small box,fill=gray!30]

\tikzstyle{Bayes box}=[rectangle,fill=black,draw, minimum height=3mm, minimum width=3mm]

\tikzstyle{gray square point}=[small box,fill=gray!50]

\tikzstyle{dphase box white}=[dhadamard]
\tikzstyle{dphase box gray}=[dhadamard,fill=gray!50!white]
\tikzstyle{phase box white}=[hadamard]
\tikzstyle{phase box gray}=[hadamard,fill=gray!50!white]

\tikzstyle{point}=[regular polygon,regular polygon sides=3,draw,scale=0.75,inner sep=-0.5pt,minimum width=9mm,fill=white,regular polygon rotate=180]
\tikzstyle{copoint}=[regular polygon,regular polygon sides=3,draw,scale=0.75,inner sep=-0.5pt,minimum width=9mm,fill=white]
\tikzstyle{dpoint}=[point,doubled]
\tikzstyle{dcopoint}=[copoint,doubled]

\tikzstyle{wide copoint}=[fill=white,draw,shape=isosceles triangle,shape border rotate=90,isosceles triangle stretches=true,inner sep=0pt,minimum width=1.5cm,minimum height=6.12mm]
\tikzstyle{wide point}=[fill=white,draw,shape=isosceles triangle,shape border rotate=-90,isosceles triangle stretches=true,inner sep=0pt,minimum width=1.5cm,minimum height=6.12mm,yshift=-0.0mm]
\tikzstyle{wide point plus}=[fill=white,draw,shape=isosceles triangle,shape border rotate=-90,isosceles triangle stretches=true,inner sep=0pt,minimum width=1.74cm,minimum height=7mm,yshift=-0.0mm]

\tikzstyle{wide dpoint}=[fill=white,doubled,draw,shape=isosceles triangle,shape border rotate=-90,isosceles triangle stretches=true,inner sep=0pt,minimum width=1.5cm,minimum height=6.12mm,yshift=-0.0mm]

\tikzstyle{tinypoint}=[regular polygon,regular polygon sides=3,draw,scale=0.55,inner sep=-0.15pt,minimum width=6mm,fill=white,regular polygon rotate=180] 

\tikzstyle{white point}=[point]
\tikzstyle{white dpoint}=[dpoint]
\tikzstyle{green point}=[white point] 
\tikzstyle{white copoint}=[copoint]
\tikzstyle{gray point}=[point,fill=gray!40!white]
\tikzstyle{gray dpoint}=[gray point,doubled]
\tikzstyle{red point}=[gray point] 
\tikzstyle{gray copoint}=[copoint,fill=gray!40!white]
\tikzstyle{gray dcopoint}=[gray copoint,doubled]

\tikzstyle{white point guide}=[regular polygon,regular polygon sides=3,font=\scriptsize,draw,scale=0.65,inner sep=-0.5pt,minimum width=9mm,fill=white,regular polygon rotate=180]

\tikzstyle{black point}=[point,fill=black,font=\color{white}]
\tikzstyle{black copoint}=[copoint,fill=black,font=\color{white}]

\tikzstyle{tiny gray point}=[tinypoint,fill=gray!40!white]

\tikzstyle{diredge}=[->]
\tikzstyle{ddiredge}=[<->]
\tikzstyle{rdiredge}=[<-]
\tikzstyle{thickdiredge}=[->, very thick]
\tikzstyle{pointer edge}=[->,very thick,gray]
\tikzstyle{pointer edge part}=[very thick,gray]
\tikzstyle{dashed edge}=[dashed]
\tikzstyle{thick dashed edge}=[very thick,dashed]
\tikzstyle{thick gray dashed edge}=[thick dashed edge,gray!40]
\tikzstyle{thick map edge}=[very thick,|->]

\makeatletter
\newcommand{\boxshape}[3]{%
\pgfdeclareshape{#1}{
\inheritsavedanchors[from=rectangle] 
\inheritanchorborder[from=rectangle]
\inheritanchor[from=rectangle]{center}
\inheritanchor[from=rectangle]{north}
\inheritanchor[from=rectangle]{south}
\inheritanchor[from=rectangle]{west}
\inheritanchor[from=rectangle]{east}
\backgroundpath{
\southwest \pgf@xa=\pgf@x \pgf@ya=\pgf@y
\northeast \pgf@xb=\pgf@x \pgf@yb=\pgf@y

\@tempdima=#2
\@tempdimb=#3

\pgfpathmoveto{\pgfpoint{\pgf@xa - 5pt + \@tempdima}{\pgf@ya}}
\pgfpathlineto{\pgfpoint{\pgf@xa - 5pt - \@tempdima}{\pgf@yb}}
\pgfpathlineto{\pgfpoint{\pgf@xb + 5pt + \@tempdimb}{\pgf@yb}}
\pgfpathlineto{\pgfpoint{\pgf@xb + 5pt - \@tempdimb}{\pgf@ya}}
\pgfpathlineto{\pgfpoint{\pgf@xa - 5pt + \@tempdima}{\pgf@ya}}
\pgfpathclose
}
}}

\boxshape{NEbox}{0pt}{0pt}
\boxshape{SEbox}{0pt}{-5pt}
\boxshape{NWbox}{5pt}{0pt}
\boxshape{SWbox}{-5pt}{0pt}
\boxshape{EBox}{-3pt}{3pt}
\boxshape{WBox}{3pt}{-3pt}
\makeatother

\tikzstyle{cloud}=[shape=cloud,draw,minimum width=1.5cm,minimum height=1.5cm]

\tikzstyle{map}=[draw,shape=NEbox,inner sep=2pt,minimum height=6mm,fill=white]
\tikzstyle{dashedmap}=[draw,dashed,gray,shape=NEbox,inner sep=2pt,minimum height=6mm,fill=white]
\tikzstyle{medium dashedmap}=[draw,dashed,gray,shape=NEbox,inner sep=2pt,minimum height=6mm,fill=white,minimum width=7mm]
\tikzstyle{semilarge dashedmap}=[draw,dashed,gray,shape=NEbox,inner sep=2pt,minimum height=6mm,fill=white,minimum width=9.5mm]
\tikzstyle{large dashedmap}=[draw,dashed,gray,shape=NEbox,inner sep=2pt,minimum height=6mm,fill=white,minimum width=12mm]
\tikzstyle{very large dashedmap}=[draw,dashed,gray,shape=NEbox,inner sep=2pt,minimum height=6mm,fill=white,minimum width=17mm]

\tikzstyle{dashed map}=[fill=white, draw=gray, shape=rectangle, style=map, dashed]

\tikzstyle{mapdag}=[draw,shape=SEbox,inner sep=2pt,minimum height=6mm,fill=white]
\tikzstyle{mapadj}=[draw,shape=SEbox,inner sep=2pt,minimum height=6mm,fill=white]
\tikzstyle{maptrans}=[draw,shape=SWbox,inner sep=2pt,minimum height=6mm,fill=white]
\tikzstyle{mapconj}=[draw,shape=NWbox,inner sep=2pt,minimum height=6mm,fill=white]

\tikzstyle{medium map}=[draw,shape=NEbox,inner sep=2pt,minimum height=6mm,fill=white,minimum width=7mm]
\tikzstyle{medium map dag}=[draw,shape=SEbox,inner sep=2pt,minimum height=6mm,fill=white,minimum width=7mm]
\tikzstyle{medium map adj}=[draw,shape=SEbox,inner sep=2pt,minimum height=6mm,fill=white,minimum width=7mm]
\tikzstyle{medium map trans}=[draw,shape=SWbox,inner sep=2pt,minimum height=6mm,fill=white,minimum width=7mm]
\tikzstyle{medium map conj}=[draw,shape=NWbox,inner sep=2pt,minimum height=6mm,fill=white,minimum width=7mm]
\tikzstyle{semilarge map}=[draw,shape=NEbox,inner sep=2pt,minimum height=6mm,fill=white,minimum width=9.5mm]
\tikzstyle{semilarge map trans}=[draw,shape=SWbox,inner sep=2pt,minimum height=6mm,fill=white,minimum width=9.5mm]
\tikzstyle{semilarge map adj}=[draw,shape=SEbox,inner sep=2pt,minimum height=6mm,fill=white,minimum width=9.5mm]
\tikzstyle{semilarge map dag}=[draw,shape=SEbox,inner sep=2pt,minimum height=6mm,fill=white,minimum width=9.5mm]
\tikzstyle{semilarge map conj}=[draw,shape=NWbox,inner sep=2pt,minimum height=6mm,fill=white,minimum width=9.5mm]
\tikzstyle{large map}=[draw,shape=NEbox,inner sep=2pt,minimum height=6mm,fill=white,minimum width=12mm]
\tikzstyle{large map conj}=[draw,shape=NWbox,inner sep=2pt,minimum height=6mm,fill=white,minimum width=12mm]
\tikzstyle{very large map}=[draw,shape=NEbox,inner sep=2pt,minimum height=6mm,fill=white,minimum width=17mm]
\tikzstyle{very very large map}=[draw,shape=NEbox,inner sep=2pt,minimum height=6mm,fill=white,minimum width=50mm]
\tikzstyle{large map dag}=[draw,shape=SEbox,inner sep=2pt,minimum height=6mm,fill=white,minimum width=12mm]

\tikzstyle{medium dmap}=[draw,doubled,shape=NEbox,inner sep=2pt,minimum height=6mm,fill=white,minimum width=7mm]
\tikzstyle{medium dmap dag}=[draw,doubled,shape=SEbox,inner sep=2pt,minimum height=6mm,fill=white,minimum width=7mm]
\tikzstyle{medium dmap adj}=[draw,doubled,shape=SEbox,inner sep=2pt,minimum height=6mm,fill=white,minimum width=7mm]
\tikzstyle{medium dmap trans}=[draw,doubled,shape=SWbox,inner sep=2pt,minimum height=6mm,fill=white,minimum width=7mm]
\tikzstyle{medium dmap conj}=[draw,doubled,shape=NWbox,inner sep=2pt,minimum height=6mm,fill=white,minimum width=7mm]
\tikzstyle{semilarge dmap}=[draw,doubled,shape=NEbox,inner sep=2pt,minimum height=6mm,fill=white,minimum width=9.5mm]
\tikzstyle{semilarge dmap trans}=[draw,doubled,shape=SWbox,inner sep=2pt,minimum height=6mm,fill=white,minimum width=9.5mm]
\tikzstyle{semilarge dmap adj}=[draw,doubled,shape=SEbox,inner sep=2pt,minimum height=6mm,fill=white,minimum width=9.5mm]
\tikzstyle{semilarge dmap dag}=[draw,doubled,shape=SEbox,inner sep=2pt,minimum height=6mm,fill=white,minimum width=9.5mm]
\tikzstyle{semilarge dmap conj}=[draw,doubled,shape=NWbox,inner sep=2pt,minimum height=6mm,fill=white,minimum width=9.5mm]
\tikzstyle{large dmap}=[draw,doubled,shape=NEbox,inner sep=2pt,minimum height=6mm,fill=white,minimum width=12mm]
\tikzstyle{large dmap conj}=[draw,doubled,shape=NWbox,inner sep=2pt,minimum height=6mm,fill=white,minimum width=12mm]
\tikzstyle{large dmap trans}=[draw,doubled,shape=SWbox,inner sep=2pt,minimum height=6mm,fill=white,minimum width=12mm]
\tikzstyle{large dmap adj}=[draw,doubled,shape=SEbox,inner sep=2pt,minimum height=6mm,fill=white,minimum width=12mm]
\tikzstyle{large dmap dag}=[draw,doubled,shape=SEbox,inner sep=2pt,minimum height=6mm,fill=white,minimum width=12mm]
\tikzstyle{very large dmap}=[draw,doubled,shape=NEbox,inner sep=2pt,minimum height=6mm,fill=white,minimum width=19.5mm]

\tikzstyle{muxbox}=[draw,shape=rectangle,minimum height=3mm,minimum width=3mm,fill=white]
\tikzstyle{dmuxbox}=[muxbox,doubled]

\tikzstyle{box}=[draw,shape=rectangle,inner sep=2pt,minimum height=6mm,minimum width=6mm,fill=white]
\tikzstyle{dbox}=[draw,doubled,shape=rectangle,inner sep=2pt,minimum height=6mm,minimum width=6mm,fill=white]
\tikzstyle{dmap}=[draw,doubled,shape=NEbox,inner sep=2pt,minimum height=6mm,fill=white]
\tikzstyle{dmapdag}=[draw,doubled,shape=SEbox,inner sep=2pt,minimum height=6mm,fill=white]
\tikzstyle{dmapadj}=[draw,doubled,shape=SEbox,inner sep=2pt,minimum height=6mm,fill=white]
\tikzstyle{dmaptrans}=[draw,doubled,shape=SWbox,inner sep=2pt,minimum height=6mm,fill=white]
\tikzstyle{dmapconj}=[draw,doubled,shape=NWbox,inner sep=2pt,minimum height=6mm,fill=white]

\tikzstyle{ddmap}=[draw,doubled,dashed,shape=NEbox,inner sep=2pt,minimum height=6mm,fill=white]
\tikzstyle{ddmapdag}=[draw,doubled,dashed,shape=SEbox,inner sep=2pt,minimum height=6mm,fill=white]
\tikzstyle{ddmapadj}=[draw,doubled,dashed,shape=SEbox,inner sep=2pt,minimum height=6mm,fill=white]
\tikzstyle{ddmaptrans}=[draw,doubled,dashed,shape=SWbox,inner sep=2pt,minimum height=6mm,fill=white]
\tikzstyle{ddmapconj}=[draw,doubled,dashed,shape=NWbox,inner sep=2pt,minimum height=6mm,fill=white]

\boxshape{sNEbox}{0pt}{3pt}
\boxshape{sSEbox}{0pt}{-3pt}
\boxshape{sNWbox}{3pt}{0pt}
\boxshape{sSWbox}{-3pt}{0pt}
\tikzstyle{smap}=[draw,shape=sNEbox,fill=white]
\tikzstyle{smapdag}=[draw,shape=sSEbox,fill=white]
\tikzstyle{smapadj}=[draw,shape=sSEbox,fill=white]
\tikzstyle{smaptrans}=[draw,shape=sSWbox,fill=white]
\tikzstyle{smapconj}=[draw,shape=sNWbox,fill=white]

\tikzstyle{dsmap}=[draw,dashed,shape=sNEbox,fill=white]
\tikzstyle{dsmapdag}=[draw,dashed,shape=sSEbox,fill=white]
\tikzstyle{dsmaptrans}=[draw,dashed,shape=sSWbox,fill=white]
\tikzstyle{dsmapconj}=[draw,dashed,shape=sNWbox,fill=white]

\boxshape{mNEbox}{0pt}{10pt}
\boxshape{mSEbox}{0pt}{-10pt}
\boxshape{mNWbox}{10pt}{0pt}
\boxshape{mSWbox}{-10pt}{0pt}
\tikzstyle{mmap}=[draw,shape=mNEbox]
\tikzstyle{mmapdag}=[draw,shape=mSEbox]
\tikzstyle{mmaptrans}=[draw,shape=mSWbox]
\tikzstyle{mmapconj}=[draw,shape=mNWbox]

\tikzstyle{mmapgray}=[draw,fill=gray!40!white,shape=mNEbox]
\tikzstyle{smapgray}=[draw,fill=gray!40!white,shape=sNEbox]

\makeatletter

\pgfdeclareshape{cornerpoint}{
\inheritsavedanchors[from=rectangle] 
\inheritanchorborder[from=rectangle]
\inheritanchor[from=rectangle]{center}
\inheritanchor[from=rectangle]{north}
\inheritanchor[from=rectangle]{south}
\inheritanchor[from=rectangle]{west}
\inheritanchor[from=rectangle]{east}
\backgroundpath{
\southwest \pgf@xa=\pgf@x \pgf@ya=\pgf@y
\northeast \pgf@xb=\pgf@x \pgf@yb=\pgf@y

\pgfmathsetmacro{\pgf@shorten@left}{\pgfkeysvalueof{/tikz/shorten left}}
\pgfmathsetmacro{\pgf@shorten@right}{\pgfkeysvalueof{/tikz/shorten right}}

\pgfpathmoveto{\pgfpoint{0.5 * (\pgf@xa + \pgf@xb)}{\pgf@ya - 5pt}}
\pgfpathlineto{\pgfpoint{\pgf@xa - 8pt + \pgf@shorten@left}{\pgf@yb - 1.5 * \pgf@shorten@left}}
\pgfpathlineto{\pgfpoint{\pgf@xa - 8pt + \pgf@shorten@left}{\pgf@yb}}
\pgfpathlineto{\pgfpoint{\pgf@xb + 8pt - \pgf@shorten@right}{\pgf@yb}}
\pgfpathlineto{\pgfpoint{\pgf@xb + 8pt - \pgf@shorten@right}{\pgf@yb - 1.5 * \pgf@shorten@right}}
\pgfpathclose
}
}

\pgfdeclareshape{cornercopoint}{
\inheritsavedanchors[from=rectangle] 
\inheritanchorborder[from=rectangle]
\inheritanchor[from=rectangle]{center}
\inheritanchor[from=rectangle]{north}
\inheritanchor[from=rectangle]{south}
\inheritanchor[from=rectangle]{west}
\inheritanchor[from=rectangle]{east}
\backgroundpath{
\southwest \pgf@xa=\pgf@x \pgf@ya=\pgf@y
\northeast \pgf@xb=\pgf@x \pgf@yb=\pgf@y

\pgfmathsetmacro{\pgf@shorten@left}{\pgfkeysvalueof{/tikz/shorten left}}
\pgfmathsetmacro{\pgf@shorten@right}{\pgfkeysvalueof{/tikz/shorten right}}

\pgfpathmoveto{\pgfpoint{0.5 * (\pgf@xa + \pgf@xb)}{\pgf@yb + 5pt}}
\pgfpathlineto{\pgfpoint{\pgf@xa - 8pt + \pgf@shorten@left}{\pgf@ya + 1.5 * \pgf@shorten@left}}
\pgfpathlineto{\pgfpoint{\pgf@xa - 8pt + \pgf@shorten@left}{\pgf@ya}}
\pgfpathlineto{\pgfpoint{\pgf@xb + 8pt - \pgf@shorten@right}{\pgf@ya}}
\pgfpathlineto{\pgfpoint{\pgf@xb + 8pt - \pgf@shorten@right}{\pgf@ya + 1.5 * \pgf@shorten@right}}
\pgfpathclose
}
}

\makeatother

\pgfkeyssetvalue{/tikz/shorten left}{0pt}
\pgfkeyssetvalue{/tikz/shorten right}{0pt}

\tikzstyle{kpoint common}=[draw,fill=white,inner sep=1pt,minimum height=4mm]
\tikzstyle{kpoint sc}=[shape=cornerpoint,kpoint common]
\tikzstyle{kpoint adjoint sc}=[shape=cornercopoint,kpoint common]
\tikzstyle{kpoint}=[shape=cornerpoint,shorten left=5pt,kpoint common]
\tikzstyle{kpoint adjoint}=[shape=cornercopoint,shorten left=5pt,kpoint common]
\tikzstyle{kpoint conjugate}=[shape=cornerpoint,shorten right=5pt,kpoint common]
\tikzstyle{kpoint transpose}=[shape=cornercopoint,shorten right=5pt,kpoint common]
\tikzstyle{kpoint symm}=[shape=cornerpoint,shorten left=5pt,shorten right=5pt,kpoint common]

\tikzstyle{black kpoint}=[shape=cornerpoint,shorten left=5pt,kpoint common,fill=black,font=\color{white}]
\tikzstyle{black kpoint adjoint}=[shape=cornercopoint,shorten left=5pt,kpoint common,fill=black,font=\color{white}]
\tikzstyle{black kpointadj}=[shape=cornercopoint,shorten left=5pt,kpoint common,fill=black,font=\color{white}]

\tikzstyle{black dkpoint}=[shape=cornerpoint,shorten left=5pt,kpoint common,fill=black, doubled,font=\color{white}]
\tikzstyle{black dkpoint adjoint}=[shape=cornercopoint,shorten left=5pt,kpoint common,fill=black, doubled,font=\color{white}]
\tikzstyle{black dkpointadj}=[shape=cornercopoint,shorten left=5pt,kpoint common,fill=black, doubled,font=\color{white}] 

\tikzstyle{kpointdag}=[kpoint adjoint]
\tikzstyle{kpointadj}=[kpoint adjoint]
\tikzstyle{kpointconj}=[kpoint conjugate]
\tikzstyle{kpointtrans}=[kpoint transpose]

\tikzstyle{big kpoint}=[kpoint, minimum width=1.2 cm, minimum height=8mm, inner sep=4pt, text depth=3mm]

\tikzstyle{wide kpoint}=[kpoint, minimum width=1 cm, inner sep=2pt]
\tikzstyle{wide kpointdag}=[kpointdag, minimum width=1 cm, inner sep=2pt]
\tikzstyle{wide kpointconj}=[kpointconj, minimum width=1 cm, inner sep=2pt]
\tikzstyle{wide kpointtrans}=[kpointtrans, minimum width=1 cm, inner sep=2pt]

\tikzstyle{gray kpoint}=[kpoint,fill=gray!50!white]
\tikzstyle{gray kpointdag}=[kpointdag,fill=gray!50!white]
\tikzstyle{gray kpointadj}=[kpointadj,fill=gray!50!white]
\tikzstyle{gray kpointconj}=[kpointconj,fill=gray!50!white]
\tikzstyle{gray kpointtrans}=[kpointtrans,fill=gray!50!white]

\tikzstyle{gray dkpoint}=[kpoint,fill=gray!50!white,doubled]
\tikzstyle{gray dkpointdag}=[kpointdag,fill=gray!50!white,doubled]
\tikzstyle{gray dkpointadj}=[kpointadj,fill=gray!50!white,doubled]
\tikzstyle{gray dkpointconj}=[kpointconj,fill=gray!50!white,doubled]
\tikzstyle{gray dkpointtrans}=[kpointtrans,fill=gray!50!white,doubled]

\tikzstyle{white label}=[draw,fill=white,rectangle,inner sep=0.7 mm]
\tikzstyle{gray label}=[draw,fill=gray!50!white,rectangle,inner sep=0.7 mm]
\tikzstyle{black label}=[draw,fill=black,rectangle,inner sep=0.7 mm]

\tikzstyle{dkpoint}=[kpoint,doubled]
\tikzstyle{wide dkpoint}=[wide kpoint,doubled]
\tikzstyle{dkpointdag}=[kpoint adjoint,doubled]
\tikzstyle{wide dkpointdag}=[wide kpointdag,doubled]
\tikzstyle{dkcopoint}=[kpoint adjoint,doubled]
\tikzstyle{dkpointadj}=[kpoint adjoint,doubled]
\tikzstyle{dkpointconj}=[kpoint conjugate,doubled]
\tikzstyle{dkpointtrans}=[kpoint transpose,doubled]

\tikzstyle{kscalar}=[kpoint common, shape=EBox, inner xsep=-1pt, inner ysep=3pt,font=\small]
\tikzstyle{kscalarconj}=[kpoint common, shape=WBox, inner xsep=-1pt, inner ysep=3pt,font=\small]

\tikzstyle{spekpoint}=[kpoint sc,minimum height=5mm,inner sep=3pt]
\tikzstyle{spekcopoint}=[kpoint adjoint sc,minimum height=5mm,inner sep=3pt]

\tikzstyle{dspekpoint}=[spekpoint,doubled]
\tikzstyle{dspekcopoint}=[spekcopoint,doubled]

 \tikzstyle{discard}=[ground,rotate=180,scale=1.5,inner sep=-2mm]
 \tikzstyle{downground}=[circuit ee IEC,thick,ground,rotate=-90,scale=1.5,inner sep=-2mm]

\tikzstyle{maxmix}=[regular polygon,regular polygon sides=3,draw=black,xscale=0.4,yscale=0.3,inner sep=-0.5pt,minimum width=10mm,fill=gray,regular polygon rotate=180]

 \tikzstyle{bigground}=[regular polygon,regular polygon sides=3,draw=gray,scale=0.50,inner sep=-0.5pt,minimum width=10mm,fill=gray]

\tikzstyle{arrs}=[-latex,font=\small,auto]
\tikzstyle{arrow plain}=[arrs]
\tikzstyle{arrow dashed}=[dashed,arrs]
\tikzstyle{arrow bold}=[very thick,arrs]
\tikzstyle{arrow hide}=[draw=white!0,-]
\tikzstyle{arrow reverse}=[latex-]
\tikzstyle{cdnode}=[]

\tikzstyle{green dashed arrow}=[green, arrow dashed]
\tikzstyle{dashed blue}=[blue, dashed]
\tikzstyle{red dashed arrow}=[red, arrow dashed]
\tikzstyle{orange arrow}=[orange, arrs]
\tikzstyle{blue arrow}=[blue, arrs]
\tikzstyle{magenta arrow}=[magenta, arrs]

\tikzstyle{dotted line}=[-, style=dotted, tikzit draw=brown]
\tikzstyle{dashed line}=[-, style=dashed, tikzit draw=cyan]
\tikzstyle{green fill line}=[-, fill={green!90!black}, tikzit draw=green]
\tikzstyle{blue fill}=[-, fill=blue, tikzit fill=blue, tikzit draw={rgb,255: red,102; green,117; blue,255}]
\tikzstyle{red}=[-, draw=red, tikzit draw=red]
\tikzstyle{blue}=[-, draw=blue, tikzit draw=blue]
\tikzstyle{thick black}=[-, draw=black, tikzit draw=black, line width=1pt]
\tikzstyle{dotted red}=[-, draw=red, style=dotted, tikzit draw=red]
\tikzstyle{dotted blue}=[-, draw=blue, tikzit draw=blue, style=dotted]
\tikzstyle{dashed thick blue}=[-, draw={rgb,255: red,28; green,176; blue,255}, tikzit draw={rgb,255: red,83; green,19; blue,156}, line width=1pt, style=dashed]
\tikzstyle{dashed thick red}=[-, draw=red, tikzit draw={rgb,255: red,255; green,100; blue,10}, line width=1pt, style=dashed]
\tikzstyle{green}=[-, draw=green, tikzit draw=green]
\tikzstyle{dotted green}=[-, draw=green, tikzit draw=green, style=dotted]
\tikzstyle{arrow}=[->]
\tikzstyle{arrow green dashed}=[draw=green, ->, tikzit draw=green, style=dashed]
\tikzstyle{arrow dashed red}=[draw=red, ->, style=dashed, tikzit draw=red]
\tikzstyle{dashed green}=[-, tikzit draw=green, draw=green, style=dashed]

\DeclareFontSeriesDefault[rm]{bf}{b}
\newcommand{\defn}[1]{\textbf{#1}}

\usepackage{abstract}
\AtBeginDocument{} 

\usepackage[font=small,labelfont=it]{caption}

\usepackage{tikz}
\usepackage{lipsum}

\usepackage{amsthm, bm}
\usepackage{mdframed}

\theoremstyle{definition}
\newtheorem{example}{Example}[section]
\newtheorem{definition}{Definition}[section]
\newtheorem{remark}{Remark}[section]
\theoremstyle{plain}
\newtheorem{theorem}{Theorem}[section]

\newtheorem{proposition}{Proposition}[section]

\newenvironment{examplebox}
  {\begin{mdframed}[
     linewidth=0.8pt,
     roundcorner=5pt,
     innertopmargin=1pt,
     innerbottommargin=12pt,
     innerleftmargin=8pt,
     innerrightmargin=8pt
   ]
   \begin{example}}
  {\end{example}
   \end{mdframed}}

\newenvironment{remarkbox}
  {\begin{mdframed}[
     linewidth=0.8pt,
     roundcorner=5pt,
     innertopmargin=1pt,
     innerbottommargin=12pt,
     innerleftmargin=8pt,
     innerrightmargin=8pt
   ]
   \begin{remark}}
  {\end{remark}
   \end{mdframed}}

\newcommand{\gen}[1]{\left\langle #1 \right\rangle}

\newcommand{\ca}{\mathcal A}
\newcommand{\cb}{\mathcal B}
\newcommand{\cc}{\mathcal C}
\newcommand{\cd}{\mathcal D}

\newcommand{\ch}{\mathcal H}
\newcommand{\ci}{\mathcal I}

\newcommand{\cl}{\mathcal L}

\newcommand{\cs}{\mathcal S}

\newcommand{\cu}{\mathcal U}
\newcommand{\cv}{\mathcal V}
\newcommand{\cw}{\mathcal W}
\newcommand{\cx}{\mathcal X}

\newcommand{\bA}{{\bm A}}
\newcommand{\bB}{{\bm B}}
\newcommand{\bC}{{\bm C}}

\newcommand{\bF}{{\bm F}}
\newcommand{\bG}{{\bm G}}
\newcommand{\bH}{{\bm H}}
\newcommand{\bI}{{\bm I}}
\newcommand{\bJ}{{\bm J}}
\newcommand{\bK}{{\bm K}}
\newcommand{\bM}{{\bm M}}
\newcommand{\bN}{{\bm N}}
\newcommand{\bO}{{\bm O}}
\newcommand{\bP}{{\bm P}}

\newcommand{\bR}{{\bm R}}
\newcommand{\bS}{{\bm S}}
\newcommand{\bT}{{\bm T}}

\newcommand{\bX}{{\bm X}}
\newcommand{\bY}{{\bm Y}}
\newcommand{\bZ}{{\bm Z}}

\DeclareMathOperator{\spn}{span}
\newcommand{\quadand}{\quad\text{and}\quad}
\newcommand{\qquadand}{\qquad\text{and}\qquad}

\usepackage{bettercenternot}
\newcommand{\ninf}[1]{\centernotsmall{\xrightarrow{#1}}}
\DeclareMathOperator{\CNOT}{CNOT}

\usepackage[normalem]{ulem}

\usepackage{authblk}

\title{\huge \bf Decoherence without the state \\ [0.5em] \Large A  causal quantum Darwinist approach}
\author[1]{Nick Ormrod\footnote{normrod@perimeterinstitute.ca}}
\author[1]{Tein van der Lugt}
\author[1,2]{Yìlè Yīng}
\author[3]{Jarosław K. Korbicz}
\affil[1]{\small Perimeter Institute for Theoretical Physics,  31 Caroline Street North, Waterloo, Ontario, N2L 2Y5, Canada}
\affil[2]{\small Department of Physics and Astronomy, University of Waterloo, Waterloo, Ontario, N2L 3G1, Canada}
\affil[3]{\small Center for Theoretical Physics, Polish Academy of Sciences, Lotników 32/46, 02-668 Warszawa, Poland}
\date{\today}

\begin{document}
    \maketitle

\begin{abstract}
    The consistent histories formalism can be used to describe histories comprised of many events associated with a variety of systems, times, and places, plausibly rich enough to describe our experiences of the classical world; however, many consistent history sets are nonclassical and bear no obvious relevance to our experiences. 
    On the other hand, the program of environmentally induced decoherence identifies classical degrees of freedom that are dynamically privileged, but does not itself provide a general account of when or how many such degrees of freedom can be consistently combined to form histories.         
    This work shows that the strengths of these two approaches can be combined by adopting a \textit{dynamics-first} perspective on decoherence. 
    Drawing inspiration from both quantum causal models and quantum Darwinism, we define the \textit{process} of decoherence in terms of the causal influences through unitary dynamics required for the redundant proliferation of information about observables.
    This approach characterises decoherence as a property of the unitary dynamics alone; it does not presuppose the existence of any quantum state.
    Instead, we show that the quantum state emerges from \emph{dual decoherence}: the process related to decoherence by time-reversal of the unitary dynamics.
    Indeed, for any set of systems in an arbitrary unitary circuit, we show that decoherence and its dual single out a unique privileged consistent history set---and we demonstrate through examples that in the privileged history set, states emerge from dual decoherence while outcomes emerge from decoherence.
    The interplay between the emergent states and outcomes makes the history set rich enough to describe complex experimental scenarios.
    Hence the idea that quantum states themselves emerge from the process of decoherence turns out to be the key missing ingredient for unifying environmentally induced decoherence and consistent histories. 
    Moreover, we show that taking this idea ontologically seriously leads to a recently proposed causal interpretation of quantum theory, or else to a dynamics-first version of the Everett interpretation.
    Finally, the causal approach developed here sheds light on the relation between decoherence and certain concepts commonly associated with it, such as the suppression of off-diagonal terms, time asymmetry, robustness of the pointer basis, the predictability sieve, and noiseless subsystems.
\end{abstract}

\clearpage
\tableofcontents

\clearpage
\section{Introduction} \label{sec:intro}

Since the pioneering works of Zeh \cite{zeh1970interpretation, zeh1973toward} and Zurek \cite{zurek1981pointer, zurek1982environment} half a century ago, decoherence has been investigated intensively, and is now widely believed to play an essential role in the emergence of classicality and the appearance of definite measurement outcomes. Despite this, a precise and widely adopted definition of decoherence remains lacking.

One common approach defines decoherence in terms of the suppression of off-diagonal terms in a density matrix. But as emphasized by Zurek,
\begin{quote}
    diagonality alone is only a symptom—and not a cause—of the effective classicality of the preferred states\ldots\ Causes of diagonality may, on occasion, differ from those relevant for the dynamics of the process of decoherence\ldots \cite{Zurek:1994zq}
\end{quote}
A different tradition defines decoherence in terms of the suppression of interference between alternative \textit{histories}, or, more formally, the \textit{consistency} of a history set \cite{griffiths1984consistent, griffiths2003consistent, gell1990complexity}. But as argued by Dowker and Kent \cite{dowker1996consistent}, consistency can be satisfied by history sets that are nonclassical. Like diagonality, consistency is a symptom of decoherence---and does not define the underlying dynamical \textit{process of decoherence}.

This paper will give a precise formulation of the process of decoherence independently of its kinematical symptoms, and explore its implications for the emergence of classicality.
Combining insights from two previously unrelated fields---quantum Darwinism~\cite{zurek2003decoherence, Ollivier_2004, zurek2004quantum, Ollivier_2005, Zurek_2009, zurek2025decoherence, korbicz2014objectivity, horodecki2015quantum, Korbicz2021roadstoobjectivity} and quantum causal modelling~\cite{Allen_2017, barrett2020quantum, barrett2021cyclic, ormrod2023causal, ormrod2025causal}---we develop an approach to decoherence that is grounded \emph{purely} in the unitary dynamics, and does not presuppose the existence of any quantum state.
This leads to a striking conclusion: the quantum state itself emerges from the unitary dynamics. Not only the diagonality of the state, but its very existence is a symptom of the process of decoherence.
The idea that not only outcomes but also states are emergent then proves its worth by allowing us to derive a unique dynamically privileged consistent history set in a highly general setting, thus unifying the consistent histories approach to decoherence with the program of environmentally induced decoherence associated with Zeh and Zurek.

Our starting point is the notion of a \emph{causal influence} between quantum degrees of freedom, drawn from the literature on quantum causal modelling, which is defined directly in terms of the unitary dynamics via noncommutation in the Heisenberg picture.
(Note that this notion of causal influence is not an emergent thermodynamic phenomenon, but a property of the unitary dynamics that turns out to be time-symmetric.)
Importantly, the presence or absence of a causal influence does not depend on the quantum state.

Inspired also by quantum Darwinism, we define decoherence in terms of the causal influences through a unitary interaction between a system and one fragment of its environment required for an observable to both ``reproduce'' and ``survive''---more precisely, for information about the observable to both be transmitted into the fragment and remain available for further transmission into other fragments.
Decohered observables are precisely those that both survive and reproduce, and thus have the potential to redundantly proliferate throughout the environment. They turn out to form a commutative algebra, singling out a preferred classical degree of freedom.

The absence of quantum states from our fundamental analysis leads us to the notion of \emph{dual decoherence}, related to decoherence by time reversal of the unitary dynamics.
Like decoherence proper, dual decoherence selects a preferred commutative algebra. While decoherence can be understood as the potential for the proliferation of information about observables, dual decoherence can be understood as the potential for the \textit{redundant implementation of generators} of reversible transformations. As we will show, decoherence leads to the emergence of outcome-like events, while dual decoherence leads to the emergence of state-like events.

We use the causal definition of decoherence and dual decoherence to combine the strengths of the two major approaches to emergent classicality: the environmentally induced approach and consistent histories. The latter is capable of describing, in a highly general setting, how many classical degrees of freedom associated with various systems, times, and places collectively emerge from the underlying quantum reality, comprising full classical \textit{histories} that admit a natural probability distribution, and explaining our experience of a complex but integrated classical world. However, as Dowker and Kent~\cite{dowker1996consistent} point out, not all consistent history sets are classical in any strong sense or bear any obvious physical significance or relation to our experiences. 
On the other hand, environmentally induced decoherence identifies individual dynamically privileged degrees of freedom that are plausibly relevant to our experience, but does not itself derive, for a general unitary process, privileged consistent history sets rich enough to account for our experience of the classical world.

Using the causal approach, we show that for an arbitrary set of systems of interest in an arbitrary unitary circuit, a unique consistent history set is privileged by the dynamics, which may be rich enough to describe an arbitrarily complicated experimental scenario. The privileged history set consists of events half of which are outcome-like events emerging from decoherence, and half of which are state-like events from dual decoherence. It is the interplay between these two types of event that allows us to reproduce the quantum predictions. Hence the idea that quantum states are emergent turns out to be the essential missing ingredient required for a derivation of privileged consistent history sets from decoherence in a highly general setting. 

This brings us to the fundamental message of this work, expressed in its title albeit in a slightly exaggerated form. We do not deny that the quantum state plays an important role in the emergence of the classical world; indeed, the classical world partly consists of events corresponding to preparations of quantum states. However, on our account, the classical world does not emerge from a decohering quantum state, nor a state of any kind. Instead, it emerges directly from the unitary dynamics, and more precisely from their causal structure.

On this approach, there is simply no need to posit the existence of any quantum state on which the unitary dynamics act. Indeed, in our derivation of consistent history sets, no states appear except for those that emerge out of the dynamics; there is no initial state on which the dynamics act. This suggests that the dynamics should thus not be thought of as a rule for transforming states, but as a \enquote{thing in itself} out of which states emerge. That this is a consistent way of thinking about the dynamics is confirmed by the causal interpretation of quantum theory recently introduced in \cite{ormrod2024quantum}, which we show in this paper is naturally expressed as a theory by which histories emerge from (causal) decoherence. 

Other results in this paper include derivations of familiar kinematical symptoms of decoherence such as off-diagonal suppression from the dynamical process; a discussion of the relationship between decoherence and time (a)symmetry; and a critical discussion of the use of robustness (i.e.\ dynamical invariance) as a criterion for identifying preferred states. In the appendices, we connect our approach with another criterion for preferred states known as the \textit{predictability sieve}~\cite{Zurek:1994zq} and extend our analysis of decoherence to Hamiltonians (rather than unitary channels).
Key limitations include the assumptions throughout that (1) all systems are finite-dimensional and that (2) decoherence is \enquote{ideal} in the sense that decohered observables survive completely unscathed by their interaction with the environment, rather than being allowed to leak partially into the environment. 
Relaxing these simplifying assumptions remains an important direction for future work.

The rest of this paper is structured as follows.
After \cref{sec:notation} sets up notations, \cref{sec:channel} provides a causal analysis of the process of decoherence in the simple context of a two-input, two-output unitary interaction between a system and a fragment of its environment. \Cref{sec:symptoms} uses this causal characterization to clarify the relationship between the process of decoherence and its symptoms. This includes a discussion of the suppression of off-diagonal entries and time asymmetry. It also discusses how robustness is related to our analysis.
\Cref{sec:dual} introduces dual decoherence and gives a heuristic account of how it gives rise to states.
This account is then made more precise in \cref{sec:circuit}, in which the causal approach is extended beyond bipartite unitary interactions to the larger context of unitary circuits.
It is in this section that we give a general derivation of consistent history sets from decoherence, aided by three concrete examples: a single measurement, two noncommuting measurements performed in sequence, and the Wigner's friend scenario.
The resulting formalism is applied in \cref{sec:when} to answer some common questions about when exactly a measurement takes place, and in \cref{sec:crqt} to reformulate the interpretation from \cite{ormrod2024quantum} and motivate a ``dynamics-first'' version of the Everett interpretation. Finally, \cref{sec:limitations} discusses limitations and directions for future work before \cref{sec:conc} concludes the paper.

\section{Notation} \label{sec:notation}

We denote a quantum system by a boldface letter, e.g. $\bA$, using the same letter to denote the algebra of linear operators $\bA  =  \cl(\ch_\bA)$ on the Hilbert space $\ch_\bA$ of the system. Throughout this work, we only consider finite-dimensional systems.

An operator in $\bA$ is denoted by $M_\bA$ and is sometimes referred to as an \enquote{$\bA$-operator}. When taking a product of operators that act locally on different systems, e.g.\
\begin{equation} \label{eq:prod}
    (M_\bA \otimes I_\bB)(I_\bA \otimes N_\bB),
\end{equation}
we often adopt the shorthand of leaving tensor products with identity operators implicit, so that \cref{eq:prod} becomes
\begin{equation}
    M_\bA N_\bB.
\end{equation}
With this convention, $\bA\bB$ denotes the tensor product of the operator algebras $\bA$ and $\bB$, i.e.\ 
\begin{equation}
    \begin{split}
        \bA \bB &= (\bA \otimes I_\bB)(I_\bA \otimes \bB) \\
        &= \bA \otimes \bB,
    \end{split}
\end{equation}
and thus represents the composite system formed by the subsystems $\bA$ and $\bB$.

Closed-system dynamics are described by unitary channels $\cu: \bA \to \bB$, which are of the form $\cu = U(\cdot) U^\dagger$ where $U: \ch_\bA \to \ch_\bB$ is a unitary map.
We adopt the convention that the input systems to a channel are labelled differently from the output systems, even if they have matching dimensions and could therefore be regarded as the \enquote{same system}.

Throughout this work, an \textit{algebra} is a von Neumann algebra. In our finite-dimensional setting, this is any set of operators on a Hilbert space that contains the identity and is closed under complex linear combinations, products, and adjoints.
The commutant of a set of operators $S\subseteq \cl(\ch)$ is denoted by $S' \coloneqq \{ M \in \cl(\ch) \mid \forall N \in S, [M,N]=0 \}$. The \textit{centre} $Z(S)$ of $S$ is the subset of $S$ that commute with all elements of $S$, i.e.\ $Z(S) = S \cap S'$.

We denote by $\gen{M}$ the smallest von Neumann algebra containing the operator $M$; this is the algebra generated by taking the closure of $\{M, I\}$ under adjoints, complex sums and products, and if $M$ is Hermitian it coincides with the double commutant $M''$.
Finally, given any two algebras $\bA$ and $\bB$, $\bA \lor \bB$ denotes the algebra they generate, i.e.\ $\bA \lor \bB = \gen{\bA \cup \bB} = (\bA\cup\bB)''$.

\section{Decoherence: the causal structure of proliferation} \label{sec:channel}

Why, in a world that classical physics ultimately fails to describe, do the objects that we perceive appear classical? Quantum Darwinism~\cite{zurek2003decoherence, Ollivier_2004, zurek2004quantum, Ollivier_2005, Zurek_2009, zurek2025decoherence, korbicz2014objectivity, horodecki2015quantum, Korbicz2021roadstoobjectivity} claims that the key to answering this question lies in the simple fact that we are only able to perceive an object when information about it is capable of proliferating throughout a macroscopic environment.\footnote{Quantum Darwinists often cast the ability to proliferate as a prerequisite for \enquote{objectivity} or “objective existence”. Here we will instead focus on its status as a prerequisite for \enquote{perceivability}. The reason for this shift in terminology is to avoid interpretative or philosophical commitments that are not strictly necessary for the project of explaining the \textit{appearance} of the classical world: if the goal is simply to save the appearances, then it is only necessary to establish that all perceivable, rather than all objectively existing, objects are effectively classical. Our arguments still apply if one makes the stronger assumption that proliferation is required for objectivity.} The fact that I perceive the table implies that there is an imprint of the table in the electromagnetic field just in front of my eyes, but the fact that you see it too implies that a similar imprint is in front of your eyes; imprints of the table are all over the room, in every location from which it can be seen. It can be argued that all perceivable objects share the same property: the ability to interact with their surrounding environment in such a way that information about them becomes redundantly encoded in many different parts of the environment, ready for access by many independent macroscopic observers.

In fact, this ability is arguably required even for a single macroscopic observer to perceive an object. Macroscopic observers have large sense organs that are inevitably not strongly sensitive to individual microscopic degrees of freedom. The imprint of the table through which I alone perceive it is not formed by a single photon, or even a small family, but by a vast array of photons.

It is therefore \textit{not} safe to assume that every object postulated by a theory is something that can be perceived. We do not perceive whatever exists, but only that which has the physical capabilities necessary to be perceived. This suggests that in order to explain why our quantum world appears so classical, one should seek to demonstrate that, according to quantum theory, the only objects capable of informationally proliferating throughout the environment are effectively classical, and that their evolution is well-approximated by classical equations of motion. 

This Darwinist perspective on emergent classicality will guide our characterization of decoherence. But as discussed in the introduction, we will make sure to characterise the \textit{process} of decoherence independently of its state-level symptoms. To this end, we will also draw from recent work on quantum causal models \cite{Allen_2017, barrett2020quantum, ormrod2023causal,ormrod2025causal}. This section will begin to develop a causal quantum Darwinist approach to decoherence by asking:
\begin{quote}
    \textit{What \  quantum \ causal \ influences \ facilitate \ the \ redundant \ proliferation \ of \ information?}
\end{quote}

To refine this question, we consider an arbitrary unitary interaction between a system of interest and a fragment of its environment, represented by a two-input, two-output unitary channel $\cu: \bS\bF \rightarrow \bT \bG$. We think of $\bS$ and $\bT$ as respectively representing the system before and after the interaction and of $\bF$ and $\bG$ as respectively representing the environment fragment before and after the interaction:
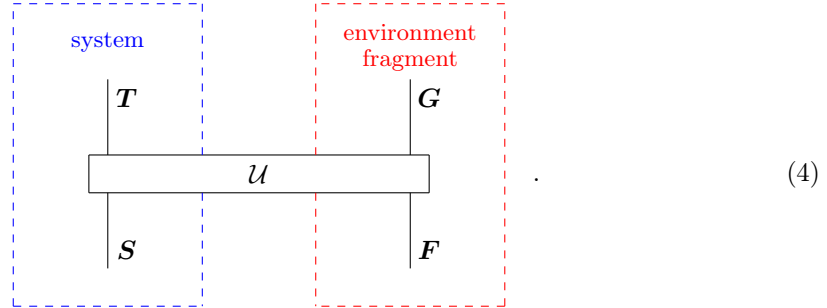
\begin{equation}
    \begin{tikzpicture}
	\begin{pgfonlayer}{nodelayer}
		\node [style=none] (0) at (-4.5, -0.5) {};
		\node [style=none] (1) at (4.5, -0.5) {};
		\node [style=none] (2) at (-4.5, 0.5) {};
		\node [style=none] (3) at (4.5, 0.5) {};
		\node [style=none] (4) at (0, 0) {$\cu$};
		\node [style=none] (5) at (-4, -2.5) {};
		\node [style=none] (6) at (-4, -0.5) {};
		\node [style=none] (7) at (-4, 0.5) {};
		\node [style=none] (8) at (-4, 2.5) {};
		\node [style=none] (9) at (4, -2.5) {};
		\node [style=none] (10) at (4, -0.5) {};
		\node [style=none] (11) at (4, 0.5) {};
		\node [style=none] (12) at (4, 2.5) {};
		\node [style=none] (13) at (-3.5, -2) {$\bS$};
		\node [style=none] (14) at (-3.5, 2) {$\bT$};
		\node [style=none] (15) at (4.5, -2) {$\bF$};
		\node [style=none] (16) at (4.5, 2) {$\bG$};
		\node [style=none] (17) at (-1.5, 0.5) {};
		\node [style=none] (18) at (-1.5, 4.5) {};
		\node [style=none] (19) at (-4, 3.5) {\small \color{blue} system};
		\node [style=none] (20) at (4, 3.75) {\small \color{red} environment};
		\node [style=none] (21) at (-1.5, -3.5) {};
		\node [style=none] (22) at (-1.5, -0.5) {};
		\node [style=none] (23) at (-6.5, 4.5) {};
		\node [style=none] (24) at (-6.5, -3.5) {};
		\node [style=none] (33) at (1.5, 0.5) {};
		\node [style=none] (34) at (1.5, 4.5) {};
		\node [style=none] (36) at (1.5, -3.5) {};
		\node [style=none] (37) at (1.5, -0.5) {};
		\node [style=none] (38) at (6.5, 4.5) {};
		\node [style=none] (39) at (6.5, -3.5) {};
		\node [style=none] (40) at (4, 3) {\small \color{red} fragment};
	\end{pgfonlayer}
	\begin{pgfonlayer}{edgelayer}
		\draw (3.center) to (1.center);
		\draw (1.center) to (0.center);
		\draw (0.center) to (2.center);
		\draw (2.center) to (3.center);
		\draw (6.center) to (5.center);
		\draw (8.center) to (7.center);
		\draw (10.center) to (9.center);
		\draw (12.center) to (11.center);
		\draw [dashed, blue] (18.center) to (17.center);
		\draw [dashed, blue] (22.center) to (21.center);
		\draw [dashed, blue] (24.center) to (21.center);
		\draw [dashed, blue] (24.center) to (23.center);
		\draw [dashed, blue] (23.center) to (18.center);
		\draw [dashed, red] (34.center) to (33.center);
		\draw [dashed, red] (37.center) to (36.center);
		\draw [dashed, red] (39.center) to (36.center);
		\draw [dashed, red] (39.center) to (38.center);
		\draw [dashed, red] (38.center) to (34.center);
	\end{pgfonlayer}
\end{tikzpicture}\quad .
\end{equation}
(For the sake of generality, however, we do not assume that ${\rm dim}(\bS) = {\rm dim}(\bT)$ or ${\rm dim}(\bF) = {\rm dim}(\bG)$ except where we say so explicitly.)
We will ask what causal influences $\cu$ must exhibit to help ensure that information about a system observable $M_\bS$ proliferates throughout the environment. 

At the intuitive level, our answer to this question will be: to help $M_\bS$ proliferate, the causal structure of $\cu$ must not only make $M_\bS$ \enquote{reproduce}, i.e.\ transmit information about itself into the environment fragment $\bG$, but also \enquote{survive}, i.e.\ remain available for further transmission into other parts of the environment, via possible future interactions. The rest of this section is devoted to making this answer precise and exploring some of its consequences for the emergence of classicality.

\subsection{Causal influence}
To more precisely answer the question of what causal influences permit proliferation, the first step is to give a formal definition of causal influence. We follow works on quantum causation~\cite{Allen_2017,barrett2020quantum,ormrod2023causal,ormrod2025causal} by defining it as noncommutation in the Heisenberg picture:

\begin{definition}\label{def:influence}
    Consider operators $M \in \bS\bF$ and $N\in \bT\bG$.
    $M$ \defn{influences $N$ through} $\cu$, written $M \xrightarrow{\cu} N$, if $[\cu^{-1}(N), M] \neq 0$. 
\end{definition}

To obtain an operational interpretation of this notion, let us assume that $M$ and $N$ are Hermitian, and interpret $M$ as a \emph{generator} of time evolution and $N$ as an \emph{observable}.
Also assume for the moment that $\cu=\ci$, the identity channel.
Heisenberg's equation reads ${\dot N = -i  [N, M]}$; thus, $M \xrightarrow \ci N$ if and only if implementing the generator $M$ affects the outcome probabilities for a measurement of the observable $N$, given at least one initial state of the composite system.
For a general $\cu$, $M \xrightarrow \cu N$ if and only if implementing $M$ \textit{before} $\cu$ changes the outcome probabilities for a measurement of $N$ carried out \textit{after} $\cu$.%
\footnote{\label{ft:hermitian}One might be concerned that we have defined influence as a relation between any pair of operators, even though we have only stated an operational interpretation for influences between Hermitian operators. After this section's analysis of decoherence is complete, we will refer the reader to \cref{app:hermitian}, which explains that the use of non-Hermitian operators in this section is a harmless mathematical convenience.}

Note that according to \cref{def:influence}, whether or not one operator influences another is a property of the unitary dynamics, and does not depend on the initial quantum state of $\bS\bF$. 

\begin{examplebox}[{Influences through a CNOT}] \label{ex:inf_cnot_op} Letting all systems be qubits, we consider the controlled NOT (CNOT) unitary interaction. It is defined at the Hilbert space level by $\ket k_\bS \ket l_\bF \mapsto  \ket k_\bT \ket{l \oplus k}_\bG$, where $\oplus$ denotes addition modulo 2, and is represented by the circuit diagram
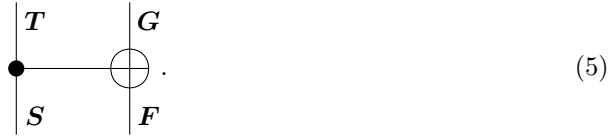
\begin{equation}\label{eq:cnotcirc}
   \begin{tikzpicture}
	\begin{pgfonlayer}{nodelayer}
		\node [style=none] (66) at (1.5, -1.75) {};
		\node [style=none] (67) at (1.5, 1.75) {};
		\node [style=none] (68) at (2, -1.25) {$\bF$};
		\node [style=none] (69) at (1.5, 1.75) {};
		\node [style=none] (71) at (-1.5, -1.75) {};
		\node [style=none] (72) at (-1, -1.25) {$\bS$};
		\node [style=none] (73) at (2, 1.25) {$\bG$};
		\node [style=none] (74) at (-1, 1.25) {$\bT$};
		\node [style=none] (75) at (-1.5, 1.75) {};
		\node [style=none] (76) at (-1.5, 1.75) {};
		\node [style=black dot] (83) at (-1.5, 0) {};
		\node [style=none] (84) at (2, 0) {};
		\node [style=none] (85) at (1, 0) {};
		\node [style=none] (86) at (2, 0) {};
		\node [style=none] (87) at (1, 0) {};
	\end{pgfonlayer}
	\begin{pgfonlayer}{edgelayer}
		\draw (83) to (84.center);
		\draw [bend left=90, looseness=1.75] (85.center) to (84.center);
		\draw [bend right=90, looseness=1.75] (87.center) to (86.center);
		\draw (66.center) to (69.center);
		\draw (71.center) to (76.center);
	\end{pgfonlayer}
\end{tikzpicture}.
\end{equation}
We define the usual Pauli operators,
\begin{equation}
X = \ket{0}\!\bra{1} + \ket{1}\!\bra{0}, \qquad
Y = -i\ket{0}\!\bra{1} + i\ket{1}\!\bra{0}, \qquad
Z = \ket{0}\!\bra{0} - \ket{1}\!\bra{1}.
\end{equation}
To find out which Pauli operators on $\bS$ influence $X_\bG$ and $Z_\bG$ through $\cu$, we pull back the latter two operators to obtain 
\begin{equation}\label{eq:cnot_ze_pullback}
    {\rm CNOT}^{-1}(X_\bG) = X_\bF \qquadand 
    {\rm CNOT}^{-1}(Z_\bG) = Z_\bS Z_\bF. 
\end{equation}
$X_\bF$ commutes with all operators on $\bS$, whereas the only Pauli on $\bS$ that $Z_\bS Z_\bF$ commutes with is $Z_\bS$. Hence
\begin{subequations} \label{eq:inf_cnot} 
\begin{align}
X_\bS &\ninf{\CNOT} X_\bG, \qquad 
Y_\bS \ninf{\CNOT} X_\bG, 
\qquad Z_\bS \ninf{\CNOT} X_\bG, \label{eq:inf_cnot_x} \\
X_\bS &\xrightarrow{\CNOT} Z_\bG,
\qquad Y_\bS \xrightarrow{\CNOT} Z_\bG, 
\qquad Z_\bS \ninf{\CNOT} Z_\bG . \label{eq:inf_cnot_z}
\end{align}
\end{subequations}

Operationally, $X_\bS \xrightarrow{\ \mathrm{CNOT}\ } Z_\bG$ means that an agent who chooses whether to implement the generator $X_\bS$ before the CNOT can signal to an agent who measures the observable $Z_\bG$ after the CNOT (for at least one initial state of $\bS\bF$). This is easily confirmed, since e.g.
\begin{equation}
    \label{eq:cnot_signalling}
    e^{-\frac{iX_\bS\pi}{2}} \ket 0_\bS \ket 0_\bF = \ket  1_\bS \ket 0_\bF \ \mapsto \  \ket 1_\bT \ket 1_\bG,
    \quad \text{whereas} \quad
    \ket 0_\bS \ket 0_\bF \ \mapsto \ \ket 0_\bT \ket 0_\bG.
\end{equation}
On the other hand, $Z_\bS \ninf{\CNOT} Z_\bG $ means that no similar protocol exists using the $Z_\bS$ generator instead of $X_\bS$.

One might have expected that $Z_\bS$ \textit{would} influence $Z_\bG$, since $Z_\bG$ comes to encode information about the observable $Z_\bS$. Recall however that causal influence, as defined by \cref{def:influence}, is a relation by which a \emph{generator} signals to an \emph{observable}, and not a relation by which one observable encodes information about another. (The latter is in fact closer to the notion of \emph{accessibility} that we will soon define.) \end{examplebox}

It is both useful and straightforward to extend \cref{def:influence} to define causal influences from an operator to a \textit{system}, or from one system to another. Recall that the \textit{commutant} $\bA'$ of a subalgebra $\bA \subseteq \bS\bF$ is the set of operators in $\bS\bF$ that commute with all operators in $\bA$.

\begin{definition}
    \label{def:influence2}
    We write
    \begin{itemize}
        \item $M \xrightarrow{\cu} \bG$ if $M\notin\cu^{-1}(\bG)'$, i.e.\ if some $N_\bG\in\bG$ is influenced by $M$;
        \item $\bS \xrightarrow \cu N$ if $\cu^{-1}(N) \notin \bS'$, i.e.\  if there exists an $M_\bS \in \bS$ such that $M_\bS \xrightarrow \cu N$;
        \item $\bS \xrightarrow \cu \bG$ if $\bS\nsubseteq\cu^{-1}(\bG)'$, i.e.\ if there exist an $M_\bS \in \bS$ and an $N_\bG \in \bG$ such that $M_\bS \xrightarrow \cu N_\bG$.
    \end{itemize}
\end{definition}

\begin{examplebox}[Influences on a system] \label{ex:inf_cnot_system}
   In the case of the CNOT of \cref{ex:inf_cnot_op}, the system observables $M_\bS$ that influence the environment fragment $\bG$ are precisely those that fail to commute with $Z_\bS$:
    \begin{equation}  \label{eq:cnot_inf_on_env}
        M_\bS \xrightarrow {\rm CNOT} \bG \qquad \Longleftrightarrow \qquad [M_\bS, Z_\bS] \neq 0.
    \end{equation}
    To verify this, note that the qubit algebra $\bG$ is obtained from the Pauli operators $\{X_\bG, Z_\bG\}$ by taking closure under complex sums and products. Therefore, $\bG' = \{X_\bG, Z_\bG\}'$, and
    \begin{equation}
    \begin{split}
            {\rm CNOT}^{-1}(\bG)' &=\{{\rm CNOT}^{-1}(X_\bG), {\rm CNOT}^{-1} (Z_\bG)\}' \\
            &  \stackrel{\eqref{eq:cnot_ze_pullback}}{=} \{X_\bF, Z_\bS Z_\bF\}'.
    \end{split}
    \end{equation}
    Hence the operators in $\bS$ that do not commute with ${\rm CNOT}^{-1}(\bG)$ are precisely those that do not commute with $Z_\bS$. 
\end{examplebox}

\subsection{Accessibility as \enquote{reproduction}}
\label{sec:reprod}

Having formalized the notion of causal influence, we can now ask: what causal influences are necessary for $\cu$ to allow all information about $M_\bS$ to flow into $\bG$, i.e.\ for $M_\bS$ to \enquote{reproduce} in this fragment? Our answer is given by the notion of \textit{accessibility}, which was introduced and motivated at length in \cite{ormrod2025causal}.

\begin{definition}
    \label{def:accessibility}
    Let $M_\bS\in\bS$. $M_\bS$ is \defn{made accessible to $\bG$ by $\cu$} (in short, $M_\bS$ is \defn{accessible to $\bG$}) if for all $G_\bS\in\bS$ such that $[M_\bS,G_\bS]\neq 0$, we have $G_\bS \xrightarrow{\cu} \bG$.
    We denote the set of all accessible operators by $\bS_{{\rm acc}\bG}^\cu$.
\end{definition}

In other words, an observable $M_\bS$ is made accessible to $\bG$ by $\cu$ if and only if every local generator on $\bS$ that nontrivially transforms $M_\bS$ also influences $\bG$ through $\cu$.
Intuitively, this means that $\bG$ \enquote{notices} any local change made to $M_\bS$.
(For non-Hermitian operators $M_\bS$, see \cref{ft:hermitian}.)
The fact that we consider only generators $G_\bS$ local to $\bS$ is motivated as follows.
$\bF$ and $\bG$ together represent the environment fragment; thus any generator $G\in\bS\bF$ that fails to commute with $\bF$ immediately affects the fragment, even if it does not influence $\bG$.
We can therefore restrict our attention to generators that commute with $\bF$, which are precisely the local generators $G_\bS \in \bF' = \bS$.
($\bF$ thus plays an important role in the notion of accessibility, in addition to $\bG$; however, it does not feature in our notation as it will suffice to think of $\bF$ as \enquote{all input systems besides $\bS$}.)

We propose that accessibility is a necessary condition for all information about $M_\bS$ to reproduce in this fragment.
This is reinforced by the following example.

\begin{examplebox}[Accessibility through a CNOT] \label{ex:accessibility}
    Again, let
    \begin{equation}
       \cu \qquad = \qquad   \begin{tikzpicture}
	\begin{pgfonlayer}{nodelayer}
		\node [style=none] (66) at (1.5, -1.75) {};
		\node [style=none] (67) at (1.5, 1.75) {};
		\node [style=none] (68) at (2, -1.25) {$\bF$};
		\node [style=none] (69) at (1.5, 1.75) {};
		\node [style=none] (71) at (-1.5, -1.75) {};
		\node [style=none] (72) at (-1, -1.25) {$\bS$};
		\node [style=none] (73) at (2, 1.25) {$\bG$};
		\node [style=none] (74) at (-1, 1.25) {$\bT$};
		\node [style=none] (75) at (-1.5, 1.75) {};
		\node [style=none] (76) at (-1.5, 1.75) {};
		\node [style=black dot] (83) at (-1.5, 0) {};
		\node [style=none] (84) at (2, 0) {};
		\node [style=none] (85) at (1, 0) {};
		\node [style=none] (86) at (2, 0) {};
		\node [style=none] (87) at (1, 0) {};
	\end{pgfonlayer}
	\begin{pgfonlayer}{edgelayer}
		\draw (83) to (84.center);
		\draw [bend left=90, looseness=1.75] (85.center) to (84.center);
		\draw [bend right=90, looseness=1.75] (87.center) to (86.center);
		\draw (66.center) to (69.center);
		\draw (71.center) to (76.center);
	\end{pgfonlayer}
\end{tikzpicture}.
    \end{equation}
    \cref{eq:cnot_inf_on_env} implies that among the Pauli operators on $\bS$, only $Z_\bS$ is accessible:
    \begin{equation} \label{eq:acc_example}
            X_\bS \notin \bS_{{\rm acc}\bG}^{\rm CNOT}, 
            \qquad Y_\bS \notin \bS_{{\rm acc}\bG}^{\rm CNOT},  
            \qquad Z_\bS \in \bS_{{\rm acc}\bG}^{\rm CNOT}. 
    \end{equation}
    Intuitively, this is the correct verdict: as \cref{eq:cnot_signalling} shows, one can signal to $\bG$ through the CNOT by choosing a $Z_\bS$ eigenstate, but no similar signalling protocols exist for $X_\bS$ or $Y_\bS$ eigenstates.
\end{examplebox}

However, we do not claim that $M_\bS \in \bS_{{\rm acc}\bG}^\cu$ is also a \textit{sufficient} condition for reproduction. Indeed, whether or not any information about $M_\bS$ actually flows into $\bG$ may depend on the initial state of $\bF$. For instance, in \cref{ex:accessibility}, one can prevent information about $Z_\bS$ flowing into $\bG$ by preparing $\bF$ in the state $\ket +_\bF$. (In \cref{sec:circuit}, however, we will see that such initial states can themselves be seen as emerging from prior unitary interactions, meaning that whether information actually flows ultimately depends only on the unitary dynamics and that which emerges from them.)

\subsection{Potential accessibility as \enquote{survival}}\label{subsec:pacc}

Recall that the question of this section is what quantum causal influences facilitate \textit{proliferation} of information about an observable. And proliferation requires not just a single act of reproduction, but many. This means that for the interaction $\cu$ to help $M_\bS$ proliferate, it must not only make $M_\bS$ reproduce in $\bG$, but also \textit{survive} this act of reproduction. By \enquote{survival}, we mean $M_\bS$ living on after $\cu$ as an observable capable of further reproduction via possible future interactions between the system and \textit{other} fragments of the environment.

We therefore consider a subsequent unitary interaction $\cv: \bT\bH \rightarrow \bR \bI$ between the system $(\bT, \bR)$ and another fragment of the environment $(\bH, \bI)$ (but \textit{not} involving the original fragment). The net result is a three-input, three-output unitary interaction
\begin{equation} \label{eq:potential_acc}.
    \begin{tikzpicture}
	\begin{pgfonlayer}{nodelayer}
		\node [style=none] (0) at (-13.75, -0.75) {};
		\node [style=none] (1) at (-13.75, 0.75) {};
		\node [style=none] (2) at (-2, 0.75) {};
		\node [style=none] (3) at (-2, -0.75) {};
		\node [style=none] (4) at (-8, 0) {$\cw(\cu, \cv)$};
		\node [style=none] (5) at (-13, -4.25) {};
		\node [style=none] (6) at (-13, 4.25) {};
		\node [style=none] (7) at (-13, 0.75) {};
		\node [style=none] (8) at (-13, -0.75) {};
		\node [style=none] (9) at (-2.75, -4.25) {};
		\node [style=none] (10) at (-2.75, -0.75) {};
		\node [style=none] (11) at (-2.75, 0.75) {};
		\node [style=none] (12) at (-2.75, 4.25) {};
		\node [style=none] (13) at (-12.5, -3.75) {$\bH$};
		\node [style=none] (14) at (-2.25, -3.75) {$\bF$};
		\node [style=none] (15) at (-13, 4.25) {};
		\node [style=none] (16) at (-2.75, 4.25) {};
		\node [style=none] (17) at (-8, 0.75) {};
		\node [style=none] (18) at (-8, 4.25) {};
		\node [style=none] (19) at (-8, -0.75) {};
		\node [style=none] (20) at (-8, -4.25) {};
		\node [style=none] (21) at (-7.5, -3.75) {$\bS$};
		\node [style=none] (22) at (-12.5, 3.75) {$\bI$};
		\node [style=none] (23) at (-2.25, 3.75) {$\bG$};
		\node [style=none] (24) at (-7.5, 3.75) {$\bR$};
		\node [style=none] (25) at (5.5, -2.75) {};
		\node [style=none] (26) at (5.5, -1.25) {};
		\node [style=none] (27) at (11, -1.25) {};
		\node [style=none] (28) at (11, -2.75) {};
		\node [style=none] (29) at (8.25, -2) {$\cu$};
		\node [style=none] (30) at (10.25, -4.25) {};
		\node [style=none] (31) at (10.25, -2.75) {};
		\node [style=none] (32) at (10.25, -1.25) {};
		\node [style=none] (33) at (10.25, 4.25) {};
		\node [style=none] (34) at (10.75, -3.75) {$\bF$};
		\node [style=none] (35) at (10.25, 4.25) {};
		\node [style=none] (36) at (6.25, 0) {};
		\node [style=none] (37) at (6.25, -2.75) {};
		\node [style=none] (38) at (6.25, -4.25) {};
		\node [style=none] (39) at (6.75, -3.75) {$\bS$};
		\node [style=none] (40) at (10.75, 3.75) {$\bG$};
		\node [style=none] (41) at (6.75, 0) {$\bT$};
		\node [style=none] (42) at (1.5, 1.25) {};
		\node [style=none] (43) at (1.5, 2.75) {};
		\node [style=none] (44) at (7, 2.75) {};
		\node [style=none] (45) at (7, 1.25) {};
		\node [style=none] (46) at (4.25, 2) {$\cv$};
		\node [style=none] (47) at (6.25, -1.25) {};
		\node [style=none] (48) at (6.25, 1.25) {};
		\node [style=none] (49) at (6.25, 2.75) {};
		\node [style=none] (50) at (6.25, 4.25) {};
		\node [style=none] (51) at (6.25, 4.25) {};
		\node [style=none] (52) at (2.25, 2.75) {};
		\node [style=none] (53) at (2.25, 4.25) {};
		\node [style=none] (54) at (2.25, 1.25) {};
		\node [style=none] (55) at (2.25, -4.25) {};
		\node [style=none] (56) at (2.75, -3.75) {$\bH$};
		\node [style=none] (57) at (6.75, 3.75) {$\bR$};
		\node [style=none] (58) at (2.75, 3.75) {$\bI$};
		\node [style=none] (59) at (0, 0) {$\coloneqq$};
		\node [style=none] (78) at (-6, 6.25) {};
		\node [style=none] (79) at (-8, 5.25) {\small \color{blue} sys.};
		\node [style=none] (80) at (-2.75, 5.25) {\small \color{red} frag.\ 1};
		\node [style=none] (81) at (-6, -5) {};
		\node [style=none] (83) at (-10, 6.25) {};
		\node [style=none] (84) at (-10, -5) {};
		\node [style=none] (86) at (-5, 6.25) {};
		\node [style=none] (87) at (-5, -5) {};
		\node [style=none] (89) at (-1, 6.25) {};
		\node [style=none] (90) at (-1, -5) {};
		\node [style=none] (92) at (-6, 6.25) {};
		\node [style=none] (93) at (-6, -5) {};
		\node [style=none] (94) at (-5, 6.25) {};
		\node [style=none] (95) at (-5, -5) {};
		\node [style=none] (96) at (-13, 5.25) {\small \color{red} frag.\ 2};
		\node [style=none] (97) at (-15, 6.25) {};
		\node [style=none] (98) at (-15, -5) {};
		\node [style=none] (99) at (-11, 6.25) {};
		\node [style=none] (100) at (-11, -5) {};
		\node [style=none] (101) at (-15, 6.25) {};
		\node [style=none] (102) at (-15, -5) {};
	\end{pgfonlayer}
	\begin{pgfonlayer}{edgelayer}
		\draw (1.center) to (2.center);
		\draw (2.center) to (3.center);
		\draw (3.center) to (0.center);
		\draw (0.center) to (1.center);
		\draw (8.center) to (5.center);
		\draw (7.center) to (6.center);
		\draw (12.center) to (11.center);
		\draw (10.center) to (9.center);
		\draw (18.center) to (17.center);
		\draw (19.center) to (20.center);
		\draw (26.center) to (27.center);
		\draw (27.center) to (28.center);
		\draw (28.center) to (25.center);
		\draw (25.center) to (26.center);
		\draw (33.center) to (32.center);
		\draw (31.center) to (30.center);
		\draw (37.center) to (38.center);
		\draw (43.center) to (44.center);
		\draw (44.center) to (45.center);
		\draw (45.center) to (42.center);
		\draw (42.center) to (43.center);
		\draw (50.center) to (49.center);
		\draw (48.center) to (47.center);
		\draw (53.center) to (52.center);
		\draw (54.center) to (55.center);
		\draw [dashed, blue] (84.center) to (81.center);
		\draw [dashed, blue] (84.center) to (83.center);
		\draw [dashed, blue] (83.center) to (78.center);
		\draw [dashed, red] (90.center) to (87.center);
		\draw [dashed, red] (90.center) to (89.center);
		\draw [dashed, red] (89.center) to (86.center);
		\draw [dashed, blue] (93.center) to (92.center);
		\draw [dashed, red] (95.center) to (94.center);
		\draw [dashed, red] (100.center) to (98.center);
		\draw [dashed, red] (100.center) to (99.center);
		\draw [dashed, red] (99.center) to (97.center);
		\draw [dashed, red] (102.center) to (101.center);
	\end{pgfonlayer}
\end{tikzpicture},
\end{equation}
where, as this equation indicates, we are thinking of $\bR$ as the system of interest after $\cv$.
We now want to ask what system observables are accessible to the second fragment.

We cannot directly apply \cref{def:accessibility}, since it presupposes a \textit{bipartite} interaction.
To address this, we will group some of our systems together.
Recall from \cref{sec:reprod} that to determine whether an observable $M_\bS$ is accessible to a fragment $(\bH,\bI)$, we need to consider generators that nontrivially transform $M_\bS$ and that do not immediately affect $\bH$, i.e.\ that are in $\bH' = \bS\bF$.
These generators are then required to influence $\bI$.
Therefore the appropriate choice is to treat $\bS\bF$ and $\bR\bG$ as composite systems:
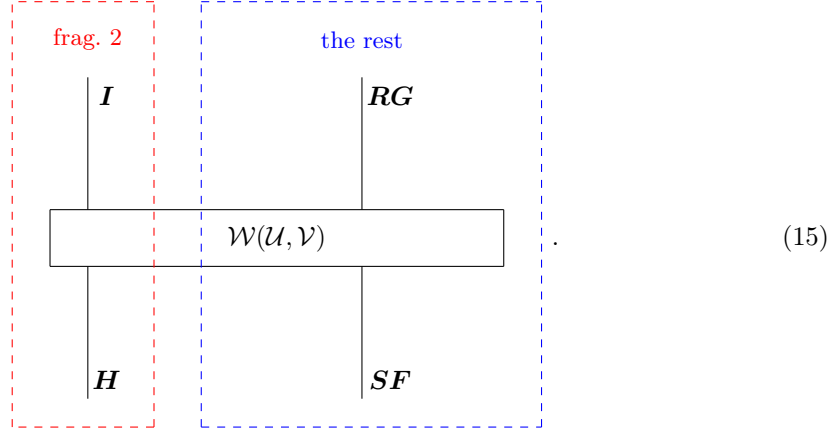
\begin{equation} \label{eq:uni_regrouping}
    \begin{tikzpicture}
	\begin{pgfonlayer}{nodelayer}
		\node [style=none] (0) at (-6, -0.75) {};
		\node [style=none] (1) at (-6, 0.75) {};
		\node [style=none] (2) at (6, 0.75) {};
		\node [style=none] (3) at (6, -0.75) {};
		\node [style=none] (4) at (0, 0) {$\cw(\cu, \cv)$};
		\node [style=none] (5) at (-5, -4.25) {};
		\node [style=none] (6) at (-5, 4.25) {};
		\node [style=none] (7) at (-5, 0.75) {};
		\node [style=none] (8) at (-5, -0.75) {};
		\node [style=none] (13) at (-4.5, -3.75) {$\bH$};
		\node [style=none] (15) at (-5, 4.25) {};
		\node [style=none] (17) at (2.25, 0.75) {};
		\node [style=none] (18) at (2.25, 4.25) {};
		\node [style=none] (19) at (2.25, -0.75) {};
		\node [style=none] (20) at (2.25, -4.25) {};
		\node [style=none] (21) at (3, -3.75) {$\bS \bF$};
		\node [style=none] (22) at (-4.5, 3.75) {$\bI$};
		\node [style=none] (24) at (3, 3.75) {$\bR \bG$};
		\node [style=none] (78) at (7, 6.25) {};
		\node [style=none] (79) at (2.25, 5.25) {\small \color{blue} the rest};
		\node [style=none] (81) at (7, -5) {};
		\node [style=none] (83) at (-2, 6.25) {};
		\node [style=none] (84) at (-2, -5) {};
		\node [style=none] (92) at (7, 6.25) {};
		\node [style=none] (93) at (7, -5) {};
		\node [style=none] (96) at (-5, 5.25) {\small \color{red} frag.\ 2};
		\node [style=none] (97) at (-7, 6.25) {};
		\node [style=none] (98) at (-7, -5) {};
		\node [style=none] (99) at (-3.25, 6.25) {};
		\node [style=none] (100) at (-3.25, -5) {};
		\node [style=none] (101) at (-7, 6.25) {};
		\node [style=none] (102) at (-7, -5) {};
	\end{pgfonlayer}
	\begin{pgfonlayer}{edgelayer}
		\draw (1.center) to (2.center);
		\draw (2.center) to (3.center);
		\draw (3.center) to (0.center);
		\draw (0.center) to (1.center);
		\draw (8.center) to (5.center);
		\draw (7.center) to (6.center);
		\draw (18.center) to (17.center);
		\draw (19.center) to (20.center);
		\draw [dashed, blue] (84.center) to (81.center);
		\draw [dashed, blue] (84.center) to (83.center);
		\draw [dashed, blue] (83.center) to (78.center);
		\draw [dashed, blue] (93.center) to (92.center);
		\draw [dashed, red] (100.center) to (98.center);
		\draw [dashed, red] (100.center) to (99.center);
		\draw [dashed, red] (99.center) to (97.center);
		\draw [dashed, red] (102.center) to (101.center);
	\end{pgfonlayer}
\end{tikzpicture}\ .
\end{equation}
We may then apply \cref{def:accessibility} to obtain a set of accessible observables $\bS\bF_{{\rm acc}\bI}^{\cw(\cu, \cv)} \subseteq \bS\bF$.

For a given $M_\bS$ and $\cu$, we may now define what it means for $M_\bS$ to remain \emph{potentially accessible} to other environment fragments after $\cu$---i.e.\ for it to \enquote{survive} $\cu$: we ask that there exists any subsequent interaction $\cv$ with any fragment $(\bH,\bI)$ so that $M_\bS$ would reproduce into that fragment.

\begin{definition} \label{def:potential_accessibility}
    Let $M_\bS\in\bS$.
    $\bG$ leaves $M_\bS$ \defn{potentially accessible} after $\cu$ if there exist systems $\bH,\bI$ and a unitary interaction $\cv: \bT\bH \rightarrow \bR \bI$ such that $M_\bS \in \bS\bF_{{\rm acc}\bI}^{\cw(\cu, \cv)}$. 
    We denote the set of all potentially accessible operators by $\bS_{{\rm pacc}\bG}^\cu$. 
\end{definition}

Interestingly, potential accessibility turns out to be equivalent to non-influence.

\begin{theorem}\label{thm:pacc_ninf}
    For any unitary channel $\cu: \bS\bF \to \bT\bG$ and observable $M_\bS \in \bS$, the following are equivalent:
    \begin{enumerate}
        \item\label{itm:pacc} $M_\bS$ is left potentially accessible by $\bG$ after $\cu$, i.e.\ $M_\bS \in \bS_{{\rm pacc}\bG}^\cu$;
        \item\label{itm:ninf} $M_\bS$ does not influence $\bG$, i.e.\ $M_\bS \ninf\cu \bG$;
        \item\label{itm:autonomy} $M_\bS$ does not leak into the environment, i.e.\ $\cu(M_\bS) \in \bT$.
    \end{enumerate}
    In other words, $\bS_{{\rm pacc}\bG}^\cu = \cu^{-1}(\bG)' \cap \bS = \cu^{-1}(\bT) \cap \bS$.
\end{theorem}
\begin{proof}
        To see the equivalence of $(2)$ and $(3)$, first note that $\cu^{-1}(\bG)'  = \cu^{-1}(\bG') = \cu^{-1}(\bT)$. 
        Therefore,
        \begin{equation}
        \begin{split}
            M_\bS \ninf \cu \bG \qquad &\Longleftrightarrow \qquad M_\bS \in \cu^{-1}(\bG)'  \\
            &\Longleftrightarrow  \qquad  M_\bS \in \cu^{-1}(\bT) \\
            &\Longleftrightarrow  \qquad \cu(M_\bS) \in \bT.
        \end{split}
    \end{equation}

    For the implication $(3) \Rightarrow (1)$, note that if $M_\bS$ evolves to a local system observable $\cu(M_\bS) \in \bT$, then it can be accessed by a second environment fragment $(\bH,\bI)$ as in \cref{eq:potential_acc} by simply letting $\cv$ swap the system and environment fragment.
    That is, $M_\bS \in \bS \bF_{{\rm acc}\bI}^{\cw(\cu, {\rm SWAP})}$, from which it follows that $M_\bS \in \bS_{{\rm pacc}\bG}^\cu$.

    Finally, we prove that $\lnot(2) \Rightarrow \lnot(1)$.
    From $\neg(2)$, we know that there exists $N_\bG\in\bG$ such that $[M_\bS, \cu^{-1}(N_\bG)]\neq 0$.
    Let $\cv : \bH\bT \to \bI\bR$ be any future interaction with a second environment fragment $(\bH,\bI)$ as in \cref{eq:potential_acc}.
    Then $M_\bS$ cannot be made accessible to $\bI$ by $\cw(\cu,\cv)$: indeed, it is nontrivially transformed by $\cu^{-1}(N_\bG)$, yet by the design of $\cw(\cu,\cv)$, $\cu^{-1}(N_\bG)$ does not influence $\bI$.
\end{proof}

The equivalence of (1) and (2) reveals an interesting connection between the \textit{observable} and the \textit{generator} that are both represented by the same Hermitian operator $M_\bS$: for the observable to remain potentially accessible to other fragments of the environment, the generator must not influence $\bG$, and vice versa.
(3) then forms an intuitive bridge between the first two conditions: it can be read as saying that an observable $M_\bS$ evolves to a local system observable, or as saying that applying the generator $M_\bS$ before $\cu$ is equivalent to acting locally on $\bT$ after $\cu$.

\begin{examplebox}[Potential accessibility through a CNOT]
    In \cref{ex:inf_cnot_system}, \cref{eq:cnot_inf_on_env}, we have already characterized the system operators that influence $\bG$ through the CNOT as precisely those that do not commute with $Z_\bS$.
    Hence, by \cref{thm:pacc_ninf},
    \begin{equation}  \label{eq:cnot_pacc_by_env}
       \bS_{{\rm pacc}\bG}^{\rm CNOT} = \{Z_\bS\}' = \gen{Z_\bS}, 
    \end{equation}
   the algebra of operators generated by $Z_\bS$.
   
   To verify this, suppose first that $M_\bS$ does \emph{not} commute with $Z_\bS$. Then it is nontrivially transformed by the generator $Z_\bS Z_\bF$. But since ${\rm CNOT}(Z_\bS Z_\bF) = Z_\bG$, this generator does not influence $\bI$ through $\cw({\rm CNOT},\cv)$---no matter the choice of $\cv$. This shows that $M_\bS$ is not left potentially accessible by the CNOT.

    Now suppose that $M_\bS$ does commute with $Z_\bS$. It can then be made accessible to $\bI$ by letting $\cv$ be, for example, another CNOT, resulting in the interaction
    
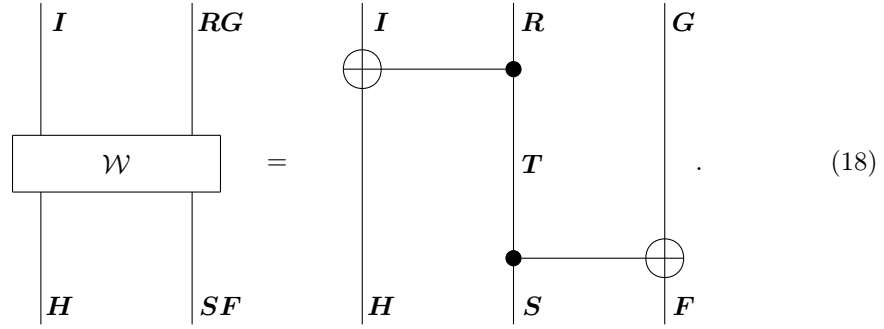
\begin{equation}
        \begin{tikzpicture}
	\begin{pgfonlayer}{nodelayer}
		\node [style=none] (0) at (4.5, -0.75) {};
		\node [style=none] (1) at (4.5, 0.75) {};
		\node [style=none] (2) at (10, 0.75) {};
		\node [style=none] (3) at (10, -0.75) {};
		\node [style=none] (4) at (7.25, 0) {$\cw$};
		\node [style=none] (5) at (5.25, -4.25) {};
		\node [style=none] (6) at (5.25, 4.25) {};
		\node [style=none] (7) at (5.25, 0.75) {};
		\node [style=none] (8) at (5.25, -0.75) {};
		\node [style=none] (13) at (5.75, -3.75) {$\bH$};
		\node [style=none] (15) at (5.25, 4.25) {};
		\node [style=none] (17) at (9.25, 0.75) {};
		\node [style=none] (18) at (9.25, 4.25) {};
		\node [style=none] (19) at (9.25, -0.75) {};
		\node [style=none] (20) at (9.25, -4.25) {};
		\node [style=none] (21) at (10, -3.75) {$\bS \bF$};
		\node [style=none] (22) at (5.75, 3.75) {$\bI$};
		\node [style=none] (24) at (10, 3.75) {$\bR \bG$};
		\node [style=none] (25) at (21.75, -4.25) {};
		\node [style=none] (26) at (21.75, 4.25) {};
		\node [style=none] (27) at (22.25, -3.75) {$\bF$};
		\node [style=none] (28) at (21.75, 4.25) {};
		\node [style=none] (29) at (17.75, 0) {};
		\node [style=none] (30) at (17.75, -4.25) {};
		\node [style=none] (31) at (18.25, -3.75) {$\bS$};
		\node [style=none] (32) at (22.25, 3.75) {$\bG$};
		\node [style=none] (33) at (18.25, 0) {$\bT$};
		\node [style=none] (34) at (17.75, 4.25) {};
		\node [style=none] (35) at (17.75, 4.25) {};
		\node [style=none] (36) at (13.75, 4.25) {};
		\node [style=none] (37) at (13.75, -4.25) {};
		\node [style=none] (38) at (14.25, -3.75) {$\bH$};
		\node [style=none] (39) at (18.25, 3.75) {$\bR$};
		\node [style=none] (40) at (14.25, 3.75) {$\bI$};
		\node [style=none] (41) at (11.5, 0) {$=$};
		\node [style=black dot] (42) at (17.75, -2.5) {};
		\node [style=none] (43) at (22.25, -2.5) {};
		\node [style=none] (44) at (21.25, -2.5) {};
		\node [style=none] (45) at (22.25, -2.5) {};
		\node [style=none] (46) at (21.25, -2.5) {};
		\node [style=black dot] (47) at (17.75, 2.5) {};
		\node [style=none] (48) at (13.25, 2.5) {};
		\node [style=none] (49) at (14.25, 2.5) {};
		\node [style=none] (50) at (13.25, 2.5) {};
		\node [style=none] (51) at (14.25, 2.5) {};
	\end{pgfonlayer}
	\begin{pgfonlayer}{edgelayer}
		\draw (1.center) to (2.center);
		\draw (2.center) to (3.center);
		\draw (3.center) to (0.center);
		\draw (0.center) to (1.center);
		\draw (8.center) to (5.center);
		\draw (7.center) to (6.center);
		\draw (18.center) to (17.center);
		\draw (19.center) to (20.center);
		\draw (42) to (43.center);
		\draw [bend left=90, looseness=1.75] (44.center) to (43.center);
		\draw [bend right=90, looseness=1.75] (46.center) to (45.center);
		\draw (25.center) to (28.center);
		\draw (47) to (48.center);
		\draw [bend right=90, looseness=1.75] (49.center) to (48.center);
		\draw [bend left=90, looseness=1.75] (51.center) to (50.center);
		\draw (30.center) to (35.center);
		\draw (36.center) to (37.center);
	\end{pgfonlayer}
\end{tikzpicture}. 
    \end{equation}
\end{examplebox}

\subsection{Decoherence}\label{subsec:decoherence}

Since it is an interaction between the system and just one fragment of the environment, $\cu$ itself cannot make information about an observable proliferate throughout the environment. But what it can do is \textit{help} it proliferate, by both making the observable reproduce in the fragment $\bG$ and letting it survive so that it can potentially go on to reproduce in another fragments. When $\cu$ does both of these things, we say that it \textit{decoheres} the observable. More formally:

\begin{definition} \label{def:decoherent}
    Let $M_\bS \in \bS$. $M_\bS$ is \defn{decohered} by $\cu$ relative to $\bG$ if it is both accessible to $\bG$ and left potentially accessible by $\bG$.
    We denote the set of all decohered operators by $\bS_{{\rm dec}\bG}^\cu$; thus,
    \begin{equation}
        \bS_{{\rm dec}\bG}^\cu \coloneqq \bS_{{\rm acc}\bG}^\cu \cap \bS_{{\rm pacc}\bG}^\cu.
    \end{equation}
    Moreover, we call $\cu$ \defn{decohering} if at least one nontrivial observable is decohered; that is, if $\bS_{{\rm dec}\bG}^\cu \neq \mathbb{C}I_\bS$.
\end{definition}

Although historically the notion of decoherence predates quantum Darwinism's emphasis on proliferation, here we are suggesting that decoherence is best defined in terms of proliferation: not as proliferation itself, but for the \textit{potential} for proliferation. Proliferation is when information about an observable spreads through many fragments of the environment; as such, proliferation cannot take place in a single bipartite interaction. On the other hand, decoherence takes place when an observable takes a meaningful step towards proliferation, and thus \textit{can} take place in a single bipartite interaction.

The theorem below explains how decoherence selects a preferred effectively classical (i.e.\ commutative) degree of freedom. Before stating this theorem, we recall from \cref{sec:notation} that, in finite dimensions, a (von Neumann) subalgebra of $\bA$ is any set of operators that contains the identity and is closed under adjoints, complex linear combinations, and products of its members, and that the \emph{centre} of an algebra $\bX$ is the subalgebra $Z(\bX) = \bX \cap \bX'$.

\begin{theorem}\label{thm:algebras}
    For any unitary channel $\cu: \bS\bF \to \bT\bG$, the sets $\bS_{{\rm acc}\bG}^\cu$ and $\bS_{{\rm pacc}\bG}^\cu$ are (von Neumann) subalgebras of $\bS$, each equal to the other's commutant within $\bS$:
    \begin{subequations} \label{eq:pacc_acc_comm}
        \begin{align}
            \bS_{{\rm pacc}\bG}^\cu &= (\bS_{{\rm acc}\bG}^\cu)' \cap \bS \label{eq:pacc_comm_acc} \\
            \text{and} \quad \quad
            \bS_{{\rm acc}\bG}^\cu &= (\bS_{{\rm pacc}\bG}^\cu)' \cap \bS.\label{eq:acc_comm_pacc} 
        \end{align}
    \end{subequations}
    The set of decohered operators $\bS_{{\rm dec}\bG}^\cu = \bS_{{\rm acc}\bG}^\cu \cap \bS_{{\rm pacc}\bG}^\cu$ thus forms a \emph{commutative} algebra, and is the centre of the accessible algebra and of the potentially accessible algebra.
\end{theorem}

\begin{proof} 
    By \cref{thm:pacc_ninf}, the $\bS$-operators left potentially accessible by $\bG$ are precisely those that do not influence $\bG$, and thus they form the algebra
    \begin{equation} \label{eq:pacc_form}
        \bS_{{\rm pacc}\bG}^\cu = \cu^{-1}(\bG)' \cap \bS.
    \end{equation}
    In its contrapositive form, \cref{def:accessibility} tells us that $M_\bS$ is accessible to $\bG$ precisely if it commutes with every operator in $\bS$ that does not influence $\bG$. Hence the accessible operators form the algebra
    \begin{equation}\label{eq:accdoublecom}
        \bS_{{\rm acc}\bG}^\cu = (\cu^{-1}(\bG)' \cap \bS)' \cap \bS.
    \end{equation}
    This implies \cref{eq:acc_comm_pacc}. By the double commutant theorem for von Neumann algebras (cf.\ \cref{sec:notation}), \cref{eq:pacc_comm_acc} follows. Thus $\bS_{{\rm dec}\bG}^\cu$ is the intersection of two commuting algebras, and is therefore a commutative algebra itself.
\end{proof}

\begin{examplebox}[Maximal decoherence]\label{re:maxdec}
    An important special case is encountered when the decohered algebra $\bS_{{\rm dec}\bG}^\cu$ is not only nontrivial, but in fact contains a \emph{nondegenerate} observable with eigenbasis $\{\ket{\psi_k}\}_k$.
    $\bS_{{\rm dec}\bG}^\cu$ must then contain the algebra generated by this observable, and is therefore identical to the algebra of operators diagonal in the eigenbasis: $\bS_{{\rm dec}\bG}^\cu = \spn_{\mathbb C}\{\ketbra{\psi_k}\}_k$.
    In this case one might say that the \emph{basis} $\{\ket{\psi_k}\}_k$ itself is decohered by $\cu$.
    We discuss this case in more detail in \cref{sec:symptoms}, where we call it \emph{maximal decoherence}.
\end{examplebox}

From the commutativity of the decohered algebra, another aspect of classicality arises: the emergence of a preferred subspace decomposition of $\ch_\bS$. In this decomposition, the preferred subspaces are the joint eigenspaces of the decohered observables.
More precisely, $\ch_\bS$ admits an orthogonal subspace decomposition
\begin{equation}\label{eq:preferred-decomposition}
    \ch_\bS = \bigoplus_{i=1}^n \ch_\bS^i
\end{equation}
such that
\begin{equation} \label{eq:dec_form}
    \bS_{{\rm dec}\bG}^\cu = \spn_{\mathbb C} \{\pi^i|i \in \{1, \ldots, n\}\},
\end{equation}
where each $\pi^i : \ch_\bS \to \ch_\bS$ is the projector onto $\ch_\bS^i$.
This subspace decomposition also provides a useful representation of the accessible and potentially accessible algebras: by the Artin--Wedderburn structure theorem~\cite{farenick2000algebras} and \cref{eq:pacc_acc_comm}, each of the preferred subspaces $\ch_\bS^i$ factorises into a tensor product $\ch_\bS^i = \ch_{\bS_L^i} \otimes \ch_{\bS_R^i}$ such that 
\begin{equation} \label{eq:algebra_forms}
    \bS_{{\rm acc}\bG}^\cu = \bigoplus_{i=1}^n \cl(\ch_{\bS_L^i}) \otimes I_{\bS_R^i}
    \qquad\text{and}\qquad
    \bS_{{\rm pacc}\bG}^\cu = \bigoplus_{i=1}^n I_{\bS_L^i} \otimes \cl(\ch_{\bS_R^i}).
\end{equation}

\Cref{thm:algebras} makes precise the quantum Darwinist idea that classicality emerges through a competition between observables to survive and reproduce. If all observables survive, i.e.\ if $\bS_{{\rm pacc}\bG}^\cu = \bS$, then \cref{thm:algebras} forces that $\bS_{{\rm acc}\bG}^\cu = \mathbb{C} I_\bS$, meaning that only trivial observables can reproduce. Conversely, if all observables reproduce, then only trivial observables can survive. It is therefore not possible for all observables to both survive and reproduce---hence the \enquote{competition}. The winners of this competition necessarily form a subalgebra of $\bS$ and thus pick out a preferred classical degree of freedom.

\begin{figure}[b!]
\centering
\begin{center} \small
\begin{minipage}[t]{0.31\textwidth}
\begin{examplebox}\label{ex:balance_a}\vspace{.8em} Let $\cu$ be the identity:
\begin{equation} \nonumber
\begin{tikzpicture}
	\begin{pgfonlayer}{nodelayer}
		\node [style=none] (30) at (20.25, -1.75) {};
		\node [style=none] (33) at (20.25, 1.5) {};
		\node [style=none] (34) at (20.75, -1.25) {$\bF$};
		\node [style=none] (35) at (20.25, 1.5) {};
		\node [style=none] (38) at (18, -1.75) {};
		\node [style=none] (39) at (18.5, -1.25) {$\bS$};
		\node [style=none] (40) at (20.75, 1) {$\bG$};
		\node [style=none] (41) at (18.5, 1) {$\bT$};
		\node [style=none] (50) at (18, 1.5) {};
		\node [style=none] (51) at (18, 1.5) {};
		\node [style=none] (65) at (16, 0) {$\cu$};
		\node [style=none] (66) at (17, 0) {$=$};
	\end{pgfonlayer}
	\begin{pgfonlayer}{edgelayer}
		\draw (30.center) to (35.center);
		\draw (38.center) to (51.center);
	\end{pgfonlayer}
\end{tikzpicture}
\end{equation}

We obtain the algebras
\begin{equation}  \nonumber
    \begin{split}
        \bS_{{\rm acc}\bG}^\ci &= \mathbb C I_\bS\\
        \bS_{{\rm pacc}\bG}^\ci &= \bS \\
        \bS_{{\rm dec}\bG}^\ci &= \mathbb C I_\bS, 
    \end{split}
\end{equation}
and the (trivial) preferred decomposition:
\begin{equation} \nonumber
    \ch_\bS = \ch_\bS.
\end{equation}
\end{examplebox}
\end{minipage}
\hfill
\begin{minipage}[t]{0.31\textwidth}

\begin{examplebox}\label{ex:balance_b}\vspace{.8em} Let $\cu$ be the SWAP:
\begin{equation} \nonumber
    \begin{tikzpicture}
	\begin{pgfonlayer}{nodelayer}
		\node [style=none] (30) at (20.25, -1.75) {};
		\node [style=none] (33) at (20.25, -1.25) {};
		\node [style=none] (34) at (20.75, -1.25) {$\bF$};
		\node [style=none] (35) at (20.25, -1.25) {};
		\node [style=none] (38) at (18, -1.75) {};
		\node [style=none] (39) at (18.5, -1.25) {$\bS$};
		\node [style=none] (40) at (20.75, 1) {$\bG$};
		\node [style=none] (41) at (18.5, 1) {$\bT$};
		\node [style=none] (50) at (18, -1.25) {};
		\node [style=none] (51) at (18, -1.25) {};
		\node [style=none] (65) at (16, 0) {$\cu$};
		\node [style=none] (66) at (17, 0) {$=$};
		\node [style=none] (67) at (20.25, 1) {};
		\node [style=none] (68) at (20.25, 1.5) {};
		\node [style=none] (69) at (20.25, 1.5) {};
		\node [style=none] (70) at (18, 1) {};
		\node [style=none] (72) at (18, 1.5) {};
		\node [style=none] (73) at (18, 1.5) {};
		\node [style=none] (80) at (20.25, -1.25) {};
		\node [style=none] (81) at (20.25, -1.25) {};
		\node [style=none] (82) at (18, 1) {};
		\node [style=none] (89) at (20.25, -1.25) {};
		\node [style=none] (90) at (20.25, -1.25) {};
		\node [style=none] (91) at (18, 1) {};
	\end{pgfonlayer}
	\begin{pgfonlayer}{edgelayer}
		\draw (30.center) to (35.center);
		\draw (38.center) to (51.center);
		\draw (67.center) to (69.center);
		\draw (70.center) to (73.center);
		\draw [in=-90, out=90, looseness=1.50] (51.center) to (67.center);
		\draw [in=-90, out=90, looseness=1.50] (90.center) to (91.center);
	\end{pgfonlayer}
\end{tikzpicture}
\end{equation}

We obtain the algebras
\begin{equation}  \nonumber
    \begin{split}
        \bS_{{\rm acc}\bG}^{\rm SWAP} &= \bS\\
        \bS_{{\rm pacc}\bG}^{\rm SWAP} &= \mathbb C I_\bS \\
        \bS_{{\rm dec}\bG}^{\rm SWAP} &= \mathbb C I_\bS, 
    \end{split}
\end{equation}
and the (trivial) preferred decomposition:
\begin{equation} \nonumber
    \ch_\bS = \ch_\bS.
\end{equation}
\end{examplebox}
\end{minipage}
\hfill
\begin{minipage}[t]{0.31\textwidth}

\begin{examplebox}\label{ex:balance_c}\vspace{.8em} Let $\cu$ be the CNOT:
\begin{equation} \nonumber
    \begin{tikzpicture}
	\begin{pgfonlayer}{nodelayer}
		\node [style=none] (66) at (1.5, -1.75) {};
		\node [style=none] (67) at (1.5, 1.75) {};
		\node [style=none] (68) at (2, -1.25) {$\bF$};
		\node [style=none] (69) at (1.5, 1.75) {};
		\node [style=none] (71) at (-1.5, -1.75) {};
		\node [style=none] (72) at (-1, -1.25) {$\bS$};
		\node [style=none] (73) at (2, 1.25) {$\bG$};
		\node [style=none] (74) at (-1, 1.25) {$\bT$};
		\node [style=none] (75) at (-1.5, 1.75) {};
		\node [style=none] (76) at (-1.5, 1.75) {};
		\node [style=black dot] (83) at (-1.5, 0) {};
		\node [style=none] (84) at (2, 0) {};
		\node [style=none] (85) at (1, 0) {};
		\node [style=none] (86) at (2, 0) {};
		\node [style=none] (87) at (1, 0) {};
	\end{pgfonlayer}
	\begin{pgfonlayer}{edgelayer}
		\draw (83) to (84.center);
		\draw [bend left=90, looseness=1.75] (85.center) to (84.center);
		\draw [bend right=90, looseness=1.75] (87.center) to (86.center);
		\draw (66.center) to (69.center);
		\draw (71.center) to (76.center);
	\end{pgfonlayer}
\end{tikzpicture}
\end{equation}

We obtain the algebras
\begin{equation}  \nonumber
    \begin{split}
        \bS_{{\rm acc}\bG}^{\rm CNOT} &= \gen{Z_\bS} \\
        \bS_{{\rm pacc}\bG}^{\rm CNOT} &= \gen{Z_\bS}  \\
        \bS_{{\rm dec}\bG}^{\rm CNOT} &= \gen{Z_\bS} , 
    \end{split}
\end{equation}
and the preferred decomposition:
\begin{equation} \nonumber
    \ch_\bS = \mathbb{C}\ket{0}_\bS \oplus \mathbb{C}\ket{1}_\bS
\end{equation} 
\end{examplebox}
\end{minipage}
\end{center}
    \caption{Examples illustrating the competitive aspect of decoherence: if all nontrivial observables survive, then none can reproduce, and vice versa.       (Recall from \cref{sec:notation} that $\gen{Z_\bS}$ denotes the von Neumann algebra generated by $Z_\bS$.)}
    \label{fig:balance}
\end{figure}
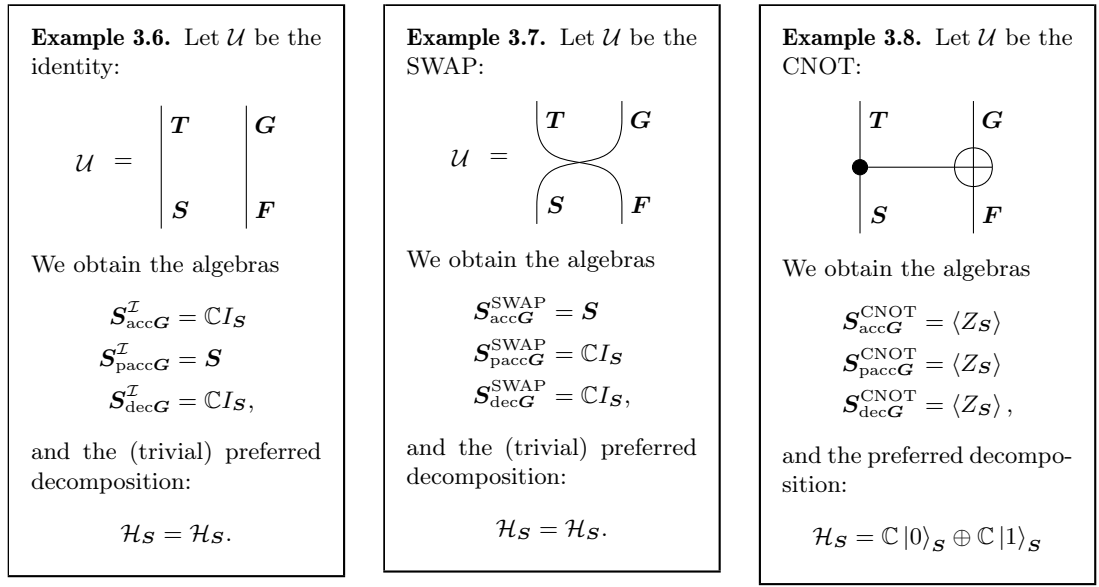

\Cref{fig:balance} illustrates the competition to achieve decoherence through some simple examples. In \cref{ex:balance_a}, all observables survive, and thus no nontrivial observables reproduce. Hence no nontrivial observables succeed in doing both; the decoherent algebra is trivial. In \cref{ex:balance_b}, all observables reproduce, and thus no nontrivial observables survive. Again, the decoherent algebra is trivial. But in \cref{ex:balance_c}, a nontrivial subalgebra of observables survive, and the same subalgebra of observables reproduce. This is the only one of the three examples in which there are nontrivial observables that win the competition to decohere.

These examples also point to a causal intuition for decoherence. The reason that there are no nontrivial decohered observables in \cref{ex:balance_a} is that there is \textit{not enough} influence from $\bS$ to $\bG$: since $G_\bS \ninf \cu \bG$ for all $G_\bS \in \bS$, no nontrivial observables are accessible to $\bG$. On the other hand, in \cref{ex:balance_b}, there is \textit{too much} influence from $\bS$ to $\bG$: since $G_\bS \xrightarrow \cu \bG$ for all nontrivial $G_\bS \in \bS$, no nontrivial observables are left potentially accessible (\cref{thm:pacc_ninf}). But in \cref{ex:balance_c}, a \textit{causal balance} is struck: some but not all nontrivial $G_\bS$ influence $\bG$. This makes it possible for nontrivial observables to both be accessible and remain potentially accessible. 

More generally, \cref{thm:algebras} shows that the very causal influences necessary for a given subalgebra $\bA \subseteq \bS$ to reproduce prevent all operators outside $\bA'$ from surviving. Conversely, the very no-influence relations required for $\bA$ to survive prevent all operators outside $\bA'$ from reproducing. Hence if one is choosing an interaction $\cu$ to ensure that nontrivial $\bS$-operators are decohered, one must be careful to pick one with neither too many nor too few influences from $\bS$-operators to $\bG$-operators. 

At the start of this section, we recalled the quantum Darwinist insight that one should not expect that all quantum degrees of freedom can be perceived, but rather only those capable of informationally proliferating throughout the environment. We now see that for a unitary $\cu$ to help observables achieve proliferation, a balance of causal influences must be struck, and that even when the best possible balance is struck, no more than a commutative subalgebra of observables can be so helped. This suggests that in a quantum world we should expect to perceive only effectively classical degrees of freedom.

Two final remarks are in order. First, as flagged in \cref{ft:hermitian}, even though our operational interpretations of the definitions we have introduced have always presupposed that the relevant operators were Hermitian, we defined them as relations between arbitrary operators. In \cref{app:hermitian}, we demonstrate that the consideration of non-Hermitian elements in our formal definitions is an inconsequential mathematical convenience. (The key fact is that any algebra is spanned by its Hermitian elements.)
Finally, here and throughout the main body of this work, the focus is on the discrete-time setting, in which interactions are represented by unitary channels rather than Hamiltonians. However, the causal approach to decoherence naturally extends to the Hamiltonian setting too, as we show in \cref{app:continuous}.

\subsection{Relation to other work} 

We now comment on the relation between the ideas developed in this section and existing literature, beginning with the literature on quantum causation.
\Cref{def:influence,def:influence2} first explicitly appeared in \cite{ormrod2025causal}. However, as \cite{ormrod2025causal} explains further, they are natural fine-grainings of the notion of causal influence that had already appeared in many papers on quantum causal modelling \cite{Allen_2017, barrett2020quantum, barrett2021cyclic, ormrod2023causal}. \cite{ormrod2023causal} argues for this approach to causal influence in quantum theory by showing that it is equivalent to many other a priori reasonable definitions. \Cref{def:accessibility} first appeared in \cite{ormrod2025causal}, in which it was used to derive the algebraic representation of quantum systems by assuming that a system is a complete set of observables that a probe can access through some unitary interaction. 

In the literature on quantum error correction, a \textit{noiseless subsystem} \cite{knill1997theory, lidar2003decoherence, lidar2014review} of $\bS$ is associated with a tensor factor $\ch_{\bS_R^i}$ of a subspace of the Hilbert space of the system that is protected from decoherence. \Cref{app:noiseless} discusses the connection between the present causal approach and the noiseless subsystems associated with quantum error correction. 

While developing a causal interpretation of quantum theory, \cite{ormrod2024quantum} defined \enquote{preferred projectors} in terms of causal influences between projectors. \cref{app:bubble} shows that these are precisely the projectors onto the preferred subspaces that are singled out by decoherence, as part of a larger proof of the claim from  \cref{sec:crqt} that the interpretation can be reformulated using our formalism for decoherence.

Finally, throughout this section we have been guided by the quantum Darwinist perspective on emergent classicality, which emphasises the importance of the proliferation of information. However, in the Darwinism literature, whether or not proliferation has taken place is commonly taken to be a property of the state, often formalized using information-theoretic quantities (e.g.\ \cite{zurek2003decoherence, Ollivier_2004, Ollivier_2005, zurek2009quantum, zurek2025decoherence}) or, more recently, in the proximity of a state to a \enquote{spectral broadcast structure} \cite{korbicz2014objectivity, horodecki2015quantum, Korbicz2021roadstoobjectivity, le2019strong, chisholm2023meaning}. Here, our object of study has been slightly different: rather than studying the kinematical properties that determine whether proliferation has occurred, we have instead characterized the dynamical structures that are prerequisite for proliferation to occur. Our motivation for focusing on the dynamics comes from both within the decoherence literature, specifically arguments that the process of decoherence should be distinguished from its symptoms~\cite{Zurek:1994zq}, and from outside, specifically recent advances in quantum causation that, as we have seen, turn out to provide a natural language for pinning down the process of decoherence.

While we have seen that this process can be defined without reference to the quantum state, we are yet to see how the quantum state itself emerges from decoherence. Before we get there, however, there is more to be said about the relation between the process of decoherence and some of its familiar symptoms.

\section{The process of decoherence and its symptoms}\label{sec:symptoms}

This section will clarify the relationship between the process of decoherence and some of decoherence's most familiar aspects: diagonality, time asymmetry, and robustness.

For simplicity,  we will focus on a special (but commonly studied) case of a decohering interaction---namely a \emph{maximally decohering} interaction. As in the previous section, we consider a unitary interaction $\bS\bF \to \bT\bG$. 

\begin{definition}
    \label{def:maximally-decohering}
    $\cu$ is called \defn{maximally decohering} if it satisfies any of the following equivalent conditions:
    \begin{enumerate}
            \item $\bS_{{\rm dec}\bG}^\cu = {\rm span}_{\mathbb{C}}\,\{\ketbra{\psi_k}{\psi_k}_\bS\}_k$
        for some orthonormal basis $\{\ket{\psi_k}\}_k$.
        \item Every subspace $\ch_\bS^i$ in the preferred decomposition is one-dimensional.
         \item $\bS_{{\rm dec}\bG}^\cu$ contains at least one \emph{nondegenerate} observable (not just nontrivial, as in \cref{def:decoherent}).
    \end{enumerate}
\end{definition}

The equivalence of these conditions follows directly from the form of the decoherent algebra as described in \cref{eq:dec_form}. The last condition is often useful: if one can identify just \textit{one} nondegenerate observable $M_\bS$ that $\cu$ decoheres, then one can immediately infer that $\cu$ is maximally decohering. If such an $M_\bS$ can be found, then it immediately follows that not just $M_\bS$ but also every operator diagonal in its eigenbasis is also decohered. 

The following theorem characterizes the maximally decohering interactions for the case where $\dim(\bS) = \dim(\bT)$.

\begin{theorem}\label{thm:uni_control}
    If $\dim(\bS) = \dim(\bT)$ then $\cu:\bS\bF \to \bT\bG$ is maximally decohering if and only if $\cu = U(\cdot)U^\dagger$ for some $U$ of the \enquote{coherent control} form
    \begin{equation}\label{eq:uni_control}
        U = \sum_k \ket*{\tilde \psi_k}_\bT \bra{\psi_k}_\bS \otimes V_{\bF \rightarrow \bG}^{(k)},
    \end{equation}
    where $\{\ket{ \psi_k}_\bS \}_{k} \subseteq \ch_\bS$ and $\{ \ket*{\tilde \psi_k}_\bT\}_{k} \subseteq \ch_\bT$ are orthonormal bases and $V_{\bF\to\bG}^{(k)} : \bF \to \bG$ are unitaries that are neither equal nor equivalent up to phase, i.e.\ $V_{\bF\to\bG}^{(k)} \not\propto V_{\bF\to\bG}^{(l)}$ for all $k \neq l$. Moreover, when $\cu$ is maximally decohering the control basis $\{\ket{ \psi_k}_\bS \}_{k}$ is unique up to phase.
\end{theorem}

To understand what this theorem says, bear in mind the equivalence of the conditions in \cref{def:maximally-decohering}. On the one hand, if \textit{any} given nondegenerate observable $M_\bS$ is decohered, then $\cu$ has the coherent control form above (for phase-inequivalent $V_{\bF\to\bG}^{(k)}$), in which case the eigenbasis of $M_\bS$ coincides with the control basis $\{\ket{\psi_k}_\bS\}_k$. Conversely, if $\cu$ has this form then it decoheres \textit{all} operators that are diagonal in the control basis, so that $\bS_{{\rm dec}\bG}^\cu = {\rm span}_{\mathbb{C}}\,\{\ketbra{\psi_k}{\psi_k}_\bS\}_k$.

\begin{proof}[Proof of \cref{thm:uni_control}]
    Suppose that $M_\bS \in \bS_{{\rm dec} \bG}^\cu$ is a nondegenerate decohered observable, and let $\{\ket{\psi_k}_\bS\}_k$ be its eigenbasis.
    Since $\bS_{{\rm dec} \bG}^\cu$ is an algebra, it must contain the algebra generated by $M_\bS$, which due to nondegeneracy is $\spn_\mathbb{C}\{\ketbra{\psi_k}_\bS\}_k$.
    Thus, each projector $\ketbra{\psi_k}_\bS$ is decohered, which means by \cref{thm:pacc_ninf} that it is transformed into a $\bT$-operator, i.e.\ $\cu(\ketbra{\psi_k}_\bS \otimes I_\bF) \in {\bT \otimes I_\bG}$.
    Since $\ketbra{\psi_k}_\bS\otimes I_\bF$ has rank $\dim(\bF)$ and $\dim(\bF)=\dim(\bG)$, $\cu(\ketbra{\psi_k}_\bS \otimes I_\bF)$ must correspond to a rank-one projector on $\ch_\bT$: $\cu(\ketbra{\psi_k}_\bS\otimes I_\bF) = \ketbra*{\tilde\psi_k}_\bT\otimes I_\bG$ for some $\ket*{\tilde\psi_k}_\bT \in \ch_\bT$, unique up to phase.
    It follows from unitarity that $\{\ket*{\tilde\psi_k}_\bT\}_k$ is an orthonormal basis for $\ch_\bT$ and that $\cu = U(\cdot)U^\dagger$ for some $U$ of the form $U = \sum_k \ket*{\tilde \psi_k}_\bT \bra{\psi_k}_\bS \otimes V_{\bF \rightarrow \bG}^{(k)}$.
    
    Assume for contradiction that there exist distinct $k$ and $l$ such that $V_{\bF \rightarrow \bG}^{(k)} \propto V_{\bF \rightarrow \bG}^{(l)}$.
    Then none of the system operators that are supported only on the subspace $\mathbb{C}\ket{\psi_k}_\bS + \mathbb{C} \ket{\psi_l}_\bS$ influence $\bG$ through $\cu$. But since some of those operators do not commute with $M_\bS$, it follows that $M_\bS \notin \bS_{{\rm acc} \bG}^\cu$, hence $M_\bS \notin \bS_{{\rm dec} \bG}^\cu$, which is a contradiction.
    Thus all unitaries $V_{\bF \rightarrow \bG}^{(k)}$ are phase-inequivalent.
    
    Conversely, assume \cref{eq:uni_control} for phase-inequivalent $V_{\bF \rightarrow \bG}^{(k)}$.
    We show that an operator $M_\bS\in\bS$ is left potentially accessible after $\cu$ iff it is diagonal in $\{\ket{\psi_k}_\bS\}_k$.
    Indeed, by \cref{thm:pacc_ninf} we have
    \begin{equation}
        \begin{split}
            M_\bS \in \bS_{{\rm pacc}\bG}^\cu \quad
            &\iff \quad \cu(M_\bS)\in\bT \\
            &\iff \quad 
              \sum_{k,l} \bra{\psi_k}_\bS M_\bS \ket{\psi_l}_\bS \ketbra*{\tilde \psi_k}{\tilde \psi_l}_\bT  \otimes V_{\bF \rightarrow \bG}^{(k)}(V_{\bF \rightarrow \bG}^{(l)})^\dagger  \in  \bT  \\
             &\iff \quad \forall k, l: \quad \bra{\psi_k}_\bS M_\bS \ket{\psi_l}_\bS = 0 \quad \text{or} \quad  V_{\bF \rightarrow \bG}^{(k)}(V_{\bF \rightarrow \bG}^{(l)})^\dagger \propto   I_{\bF \rightarrow \bG}  \\
             &\iff \quad \forall k \neq l : \quad  \bra{\psi_k}_\bS M_\bS \ket{\psi_l}_\bS = 0 \\
             &\iff \quad M_\bS \in \spn_{\mathbb{C}}\{\ketbra{\psi_k}_\bS\}_k.
        \end{split}
    \end{equation}
    The decohered algebra $\bS_{{\rm dec}\bG}^\cu$ is the centre of $\bS_{{\rm pacc}\bG}^\cu$ and therefore also equals $\spn_{\mathbb{C}}\{\ketbra{\psi_k}_\bS\}_k$, showing that $\cu$ is maximally decohering.
\end{proof}

\subsection{Suppression of off-diagonal elements}

The suppression of off-diagonal terms in the reduced density matrix of the system is one of the best-known consequences of decoherence. To see how it comes about, let us assume that $\cu$ acts on an initial product state
\begin{equation}\label{eq:initial-product-state}
    \rho_\bS \otimes \rho_\bF.
\end{equation}
\Cref{thm:uni_control} shows that a maximally decohering interaction $\cu$ transforms the preferred basis $\{\ket{\psi_k}_\bS\}_k$ on $\bS$ to a preferred basis  $\{\ket*{\tilde \psi_k}_\bT\}_k$ on $\bT$. If the marginal states of the system before and after $\cu$ are written in terms of these preferred bases, then the matrix elements evolve from
\begin{equation} \label{eq:suppression}
  \bra{\psi_k}_\bS \rho_\bS \ket{\psi_l}_\bS \quad \text{to}
   \quad \bra*{\tilde \psi_k}_\bT \rho_\bT \ket*{\tilde \psi_l}_\bT  = \Tr\bigl((V_{\bF \rightarrow \bG}^{(l)})^\dagger V_{\bF \rightarrow \bG}^{(k)} \rho_\bF \bigr) \,
      \bra{\psi_k}_\bS \rho_\bS \ket{\psi_l}_\bS.
\end{equation}
Hence the magnitude of the $kl$th matrix element is suppressed by a factor of
\begin{equation}
    \left|\Tr\bigl((V_{\bF\to\bG}^{(l)})^\dagger V_{\bF\to\bG}^{(k)} \rho_\bF \bigr)\right|.
\end{equation}
This equals 1 for $k=l$, meaning that on-diagonal terms are preserved. For $k \neq l$, the unitarity of $V_{\bF\to\bG}^{(l)}$ and $V_{\bF\to\bG}^{(k)}$ implies that it is no greater than $1$.

There may exist certain initial environment states $\rho_\bF$ that lead to a suppression factor of 1 for some $k \neq l$, meaning that there is an off-diagonal term that is not \emph{strictly} suppressed. But these states are fine-tuned.
Indeed, one way to see this is to realize that all full-rank $\rho_\bF$ lead to strict off-diagonal suppression, and that in the space of density operators of a fixed dimension, the subset of rank-deficient density operators has measure zero. For if $k\neq l$ then $(V_{\bF\to\bG}^{(l)})^\dagger V_{\bF\to\bG}^{(k)} \not\propto I_\bF$ by \cref{thm:uni_control}, which together with full-rankness of $\rho_\bF$ implies that the suppression factor is strictly less than 1.

Hence off-diagonal suppression is very likely, though not logically guaranteed, to occur when a maximally decohering $\cu$ is applied to a product state. This supports the view that diagonality should be regarded as a symptom rather than a defining feature of the process of decoherence~\cite{Zurek:1994zq}. 
It also turns out that regardless of how much or how little off-diagonal elements of $\rho_\bS$ are suppressed, they make no difference to the environment fragment $\bG$.
Indeed, \cref{eq:uni_control} shows that while $\bG$ obtains information encoded in the basis $\{\ket{\psi_k}_\bS\}_k$, it obtains no information about the phases \emph{between} these basis elements.
Formally, if we define the dephasing channel
\begin{equation}
    \label{eq:dephase}
    \cd_\bS(\rho_\bS) := \sum_k \ketbra{\psi_k}_\bS \,\rho_\bS\, \ketbra{\psi_k}_\bS,
\end{equation}
which eliminates the off-diagonal elements of $\rho_\bS$, then
\begin{equation}\label{eq:env-doesnt-notice}
    {\rm Tr}_\bS  \circ \cu 
    = {\rm Tr}_\bS  \circ \cu \circ \cd_\bS;
\end{equation}
that is, the environment cannot distinguish between $\rho_\bS$ and its diagonalized counterpart $\cd_\bS(\rho_\bS)$. It follows that
an observer who learns about $\bS$ only by interactions with the environment fragment $(\bF,\bG)$ cannot learn anything about these off-diagonal terms. 

Finally, we note that if $\cu$ is not maximally decohering but just \emph{decohering} (as in \cref{def:decoherent}, i.e.\ if $\bS_{{\rm dec}\bG}^\cu \neq \mathbb C I_\bS$) then, analogously to \cref{eq:suppression}, the off-diagonal \emph{blocks} of $\rho_\bS$ in the direct sum decomposition defined by $\bS_{{\rm dec}\bG}^\cu$ get suppressed.
Moreover, analogously to \cref{eq:env-doesnt-notice}, $\rho_\bS$ is equivalent to its \emph{block}-diagonalized counterpart from the point of view of the environment fragment. This reflects the fact, already evident from \cref{thm:algebras}, that the fragment has no access to observables that do not commute with the decohered observables.

\subsection{Time asymmetry}

Decoherence is often associated with time asymmetry. This is not surprising, since some of the symptoms of decoherence do indeed pick out a preferred direction of time: if off-diagonal elements are suppressed over time, then, when the same process is seen in reverse, they are \textit{amplified}. 

One might wonder if the ultimate reason for this time asymmetry is that the underlying dynamical structures responsible for decoherence somehow pick out a preferred direction of time.
But we will now argue that this is not the case.
First of all, $\cu$ itself is of course invertible and thus does not pick out a preferred direction of time.
Neither does causal influence, the basic conceptual ingredient in our analysis of decoherence, pick out a preferred direction: the existence of a causal influence between operators is time-reversal invariant, i.e.\ 
\begin{equation}\label{eq:reversibility_inf}
    M \xrightarrow \cu N \qquad \Longleftrightarrow \qquad N \xrightarrow{\cu^{-1}} M.
\end{equation}
(This follows from the fact that $[N, \cu^{-1}(M)] = \cu^{-1}([\cu(N), M]) \neq 0$ if and only if $[M, \cu(N)] \neq 0$; c.f.~\cite[Prop.~3.1]{ormrod2023causal}.)
Finally, decoherence itself prefers no particular direction of time: as the following shows, the effectively classical degrees of freedom picked out by decoherence are time-reversal invariant.

\begin{theorem} \label{thm:reversibility}
        The decoherent algebra is invariant under inversion of $\cu$, in the sense that
        \begin{equation}
            \label{eq:reversibility_dec}
            \bS_{{\rm dec}\bG}^\cu = \cu^{-1}(\bT_{{\rm dec}\bF}^{\cu^{-1}}).
        \end{equation}
        As a consequence, $\cu$ bijectively maps between the preferred subspaces $\ch_\bS^k$ of $\bS$ in \cref{eq:preferred-decomposition} and the preferred subspaces $\ch_\bT^l$ of $\bT$.
\end{theorem}

\begin{proof}
    First note that
    \begin{equation} \label{eq:reversibility_pacc}
        \begin{split}
            \bS_{{\rm pacc}\bG}^\cu &= \cu^{-1}(\bG)' \cap \bS
            = \cu^{-1}(\bT) \cap \bF'
            = \cu^{-1} \big(\bT \cap \cu(\bF')\big)
            = \cu^{-1}( \bT_{{\rm pacc}\bF}^{\cu^{-1}}).
        \end{split}
    \end{equation}
    By \cref{thm:algebras}, taking the centre of the algebras on each side of this equation gives \cref{eq:reversibility_dec}.
\end{proof}

These results raise a question. If neither the unitary interaction itself, nor the influences through it, nor the classical degree of freedom that they pick out a preferred direction of time, then what is the origin of the time asymmetry associated with the symptoms of decoherence? The answer is that the time asymmetry arises from background assumptions about quantum states that already pick out a preferred direction of time.

For example, to derive the suppression of off-diagonals in \cref{eq:suppression}, we assumed that the initial system-environment state factorized. If we had instead assumed that the \emph{final} state factorized, then by \cref{thm:reversibility} the very same dynamical structures would have entailed an amplification of off-diagonals over time. Thus depending on one's prior beliefs about the relevant quantum states, either the suppression or the amplification of off-diagonals can be regarded as a symptom of the process of decoherence.

\begin{remarkbox} \label{remark:reversibility}
    A similar story can be told about causation itself. Although causation is often regarded as necessarily time-asymmetric, we would argue that this is true only of its \textit{symptoms}, such as signalling. Whether agents can exploit a causal influence between systems to send a signal depends on their prior knowledge of those systems. Since agents typically know more about the past than the future, this provides a possible explanation of why agents are able to send signals into the future (but not vice versa) even if the underlying process through which they signal is reversible. 
\end{remarkbox}

\subsection{Robustness}

Unitaries of the coherent control form of \cref{eq:uni_control} feature prominently in the literature on decoherence. However, the unitaries in the literature often satisfy the additional constraint that the preferred basis $\{\ket{\psi_k}\}_k$ is \textit{robust}, meaning that 
\begin{equation}
    \forall k, \quad \ket{\psi_k}_\bS = \ket*{\tilde \psi_k}_\bT,
\end{equation}
where we are presupposing an identification between states on $\bS$ and states on $\bT$. More generally, for continuous-time evolution generated by a Hamiltonian, a basis is said to be robust if each projector $\ketbra{\psi_k}{\psi_k}$ commutes with the Hamiltonian, so that the basis states are invariant under the evolution.

In fact, robustness is often used as a defining \textit{criterion} of the preferred basis that is selected by decoherence, calling into question why our analysis allows the preferred basis to be nonrobust. To answer this question, it is helpful to begin by giving an argument for the robustness criterion adapted from \cite{zurek1981pointer}, which tackles the preferred basis problem in the context of a measurement modelled unitarily. 

The key assumption of the argument is that the preferred \enquote{pointer} basis of the measurement device should be one that is capable not only of displaying a record of the outcome of the measurement, but of \textit{storing} that record even while the device is being perturbed by its interaction with the environment.
With this in mind, we assume that $\cu$ represents not a measurement itself, but a \textit{post-measurement interaction} between $(\bS, \bT)$, representing the pointer on the measurement device, and $(\bF, \bG)$, representing the environment. Suppose that before $\cu$ the measurement outcome is recorded in the position basis of the pointer. If $\cu$ makes the pointer move, then this record may end up being lost in the correlations between the pointer and its environment. By contrast, if we assume that the position of the record is robust under $\cu$, then the pointer does not move at all, and thus is certain to continue recording the measurement outcome. 

In response to this argument, we note that not \textit{all} possible motions of the pointer lead to the loss of a record. Indeed, suppose $\cu$ makes the pointer shift by some fixed distance $d$. To recover the record after $\cu$, an observer could apply a local unitary that shifts the pointer back. Alternatively, the observer could simply change the way that they interpret the position of the pointer, bearing in mind that the position $x$ now denotes the outcome formerly denoted by $x-d$. This suggests that robustness, while sufficient, is not necessary for the preservation of records. 

By contrast, we would suggest that the condition of \emph{potential accessibility} can be regarded as both necessary and sufficient for the device to continue to store the record. Indeed, by \cref{thm:pacc_ninf}, $M_\bS$ is left potentially accessible if and only if it does not leak into the surrounding environment. If a potentially accessible $M_\bS$ is not robust, then any outcome encoded in $M_\bS$ before $\cu$ will be encoded in a different observable after $\cu$---but the new encoding observable is also a local system observable, and can be restored to the original observable by means of a local unitary applied to the system. Hence although this observable is nontrivially transformed by $\cu$, the transformation in question does not lead to the loss of a record.

Another indication that potential accessibility is a suitable criterion for identifying preferred degrees of freedom is its close connection with the \textit{predictability sieve}. Indeed, \cref{app:sieve} shows that in the present setting of a discrete-time interaction $\cu$, the causal approach and the predictability sieve are in agreement about when a basis is preferred. Hence the predictability sieve itself does not require robustness in this setting.
This strengthens the argument that in a world where time were fundamentally discrete, there would be no reason to believe that the preferred basis should be robust.  

As \cref{app:continuous} shows, robustness \textit{does} become a consequence of decoherence in the continuous-time setting of a Hamiltonian interaction.
However, the arguments above suggest that even in that setting, robustness should be regarded as precisely that: a consequence of decoherence rather than one of its conceptual ingredients. In this sense, robustness can be considered as a symptom of decoherence as well (though it is not a \enquote{kinematical} symptom, as it can be defined purely in terms of the Hamiltonian). 

Furthermore, \cref{app:noiseless} shows that when robustness is redefined in a rotating frame of reference that absorbs the system's local evolution, the resulting notion is equivalent to potential accessibility. This result also establishes a link between potential accessibility and the \textit{noiseless subsystems} from quantum error correction \cite{knill1997theory, lidar2003decoherence,lidar2014review}.

\section{Dual decoherence: the causal structure of redundant implementation} \label{sec:dual}

The next section will describe how outcome-like events emerge from decoherence. But it will also describe how states (or more precisely, state-like events) emerge from a \textit{dual} version of decoherence. The role of this section is to introduce dual decoherence and give a sense of why states should emerge from it.

Dual decoherence is naturally suggested by the previous section’s insight that the existence of a causal influence does not depend on the direction of time. Indeed, for a unitary $\cu:\bS\bF \to \bT\bG$, let us consider the algebra $\bT^{\cu^{-1}}_{{\rm dec}\bF}$. Obviously, one way of interpreting this algebra is by appealing to the causal influences through the inverse unitary $\cu^{-1}$. Explicitly, we know from \cref{sec:channel} that an operator $M_\bT$ lies in $\bT^{\cu^{-1}}_{{\rm dec}\bF}$ if and only if it satisfies both accessibility,
\begin{equation} \label{eq:acc} \forall N_\bT, \qquad\qquad [M_\bT, N_\bT] \neq 0 \qquad \implies \qquad N_\bT \xrightarrow {\cu^{-1}} \bF, \end{equation}
and potential accessibility, which by \cref{thm:pacc_ninf} is equivalent to the no-influence condition
\begin{equation} \label{eq:ninf}
    M_\bT \ninf{\cu^{-1}} \bF.
\end{equation}

There is, however, another interpretation of this algebra, which refers directly to the influences through $\cu$ rather than through its inverse. By the time symmetry of causal influence (as expressed by \cref{eq:reversibility_inf}), the causal condition of \cref{eq:acc} on $\cu^{-1}$ is equivalent to a causal condition on $\cu$, namely,
\begin{equation} \label{eq:imp}
    \forall N_\bT, \qquad\qquad [N_\bT, M_\bT] \neq 0 \qquad \implies \qquad \bF \xrightarrow {\cu} N_\bT.
\end{equation} 
Similarly, the causal condition of \cref{eq:ninf} is equivalent to
\begin{equation} \label{eq:pimp}
    \bF \ninf{\cu} M_\bT.
\end{equation}

The question now is how these last two conditions are to be interpreted. To that end, let us assume that $M_\bT$ is Hermitian. As we saw in \cref{sec:channel}, the accessibility condition of \cref{eq:acc} is then naturally interpreted by thinking of $M_\bT$ as an observable. But to interpret \cref{eq:imp} we will instead think of $M_\bT$ as a \textit{generator}.
Then \cref{eq:imp} says that any observable nontrivially transformed by the generator $M_\bT$ is influenced by $\bF$ through $\cu$.
That is, the environment fragment $\bF$ is able to transform any observable that the generator $M_\bT$ would have transformed. In this sense, the fragment is able to \textit{implement} the generator.

Similarly, the causal condition of \cref{eq:pimp} can be interpreted in terms of \textit{potential implementability}.
In a line of argument that closely resembles our motivation of potential accessibility around \cref{eq:potential_acc}, one can ask whether there exists any \emph{prior} (rather than posterior, as in \cref{eq:potential_acc}) interaction $\cv$ with another environment fragment $(\bH,\bI)$ that would allow $\bH$ to implement the generator $M_\bT$ via the resulting tripartite interaction
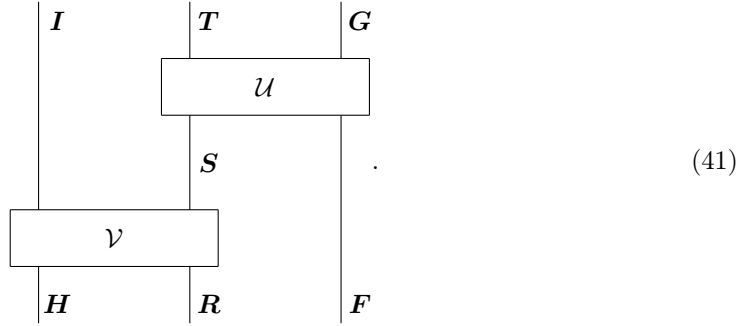
\begin{equation}  \label{eq:potential_imp}
    \begin{tikzpicture}
	\begin{pgfonlayer}{nodelayer}
		\node [style=none] (25) at (21, 1.25) {};
		\node [style=none] (26) at (21, 2.75) {};
		\node [style=none] (27) at (26.5, 2.75) {};
		\node [style=none] (28) at (26.5, 1.25) {};
		\node [style=none] (29) at (23.75, 2) {$\cu$};
		\node [style=none] (30) at (25.75, -4.25) {};
		\node [style=none] (31) at (25.75, 1.25) {};
		\node [style=none] (32) at (25.75, 2.75) {};
		\node [style=none] (33) at (25.75, 4.25) {};
		\node [style=none] (34) at (26.25, -3.75) {$\bF$};
		\node [style=none] (35) at (25.75, 4.25) {};
		\node [style=none] (36) at (21.75, 0) {};
		\node [style=none] (37) at (21.75, -2.75) {};
		\node [style=none] (38) at (21.75, -4.25) {};
		\node [style=none] (39) at (22.25, -3.75) {$\bR$};
		\node [style=none] (40) at (26.25, 3.75) {$\bG$};
		\node [style=none] (41) at (22.25, 0) {$\bS$};
		\node [style=none] (42) at (17, -2.75) {};
		\node [style=none] (43) at (17, -1.25) {};
		\node [style=none] (44) at (22.5, -1.25) {};
		\node [style=none] (45) at (22.5, -2.75) {};
		\node [style=none] (46) at (19.75, -2) {$\cv$};
		\node [style=none] (47) at (21.75, -1.25) {};
		\node [style=none] (48) at (21.75, 1.25) {};
		\node [style=none] (49) at (21.75, 2.75) {};
		\node [style=none] (50) at (21.75, 4.25) {};
		\node [style=none] (51) at (21.75, 4.25) {};
		\node [style=none] (52) at (17.75, -1.25) {};
		\node [style=none] (53) at (17.75, 4.25) {};
		\node [style=none] (54) at (17.75, -2.75) {};
		\node [style=none] (55) at (17.75, -4.25) {};
		\node [style=none] (56) at (18.25, -3.75) {$\bH$};
		\node [style=none] (57) at (22.25, 3.75) {$\bT$};
		\node [style=none] (58) at (18.25, 3.75) {$\bI$};
	\end{pgfonlayer}
	\begin{pgfonlayer}{edgelayer}
		\draw (26.center) to (27.center);
		\draw (27.center) to (28.center);
		\draw (28.center) to (25.center);
		\draw (25.center) to (26.center);
		\draw (33.center) to (32.center);
		\draw (31.center) to (30.center);
		\draw (37.center) to (38.center);
		\draw (43.center) to (44.center);
		\draw (44.center) to (45.center);
		\draw (45.center) to (42.center);
		\draw (42.center) to (43.center);
		\draw (50.center) to (49.center);
		\draw (48.center) to (47.center);
		\draw (53.center) to (52.center);
		\draw (54.center) to (55.center);
	\end{pgfonlayer}
\end{tikzpicture}.
\end{equation}
If so, then $M_\bT$ may be called \emph{potentially implementable}, and by an argument dual to \cref{thm:pacc_ninf}, this notion is equivalent to the causal condition in \cref{eq:pimp}. If, on the other hand, $M_\bT$ is not potentially implementable, then this means that $\cu$ intercepts the causal influences required for $\bH$ to implement the generator.

Following an argument analogous to that of \cref{thm:algebras}, we know that the operators on $\bT$ that are implementable via $\cu$ form an algebra.
And due to the equivalence of \cref{eq:acc,eq:imp}, this is precisely the algebra $\bT_{{\rm acc}\bF}^{\cu^{-1}}$ of operators that are \textit{accessible} through $\cu^{-1}$.
Similarly, the set of potentially implementable operators is precisely the algebra $\bT_{{\rm pacc}\bF}^{\cu^{-1}}$. Hence the set of operators on $\bT$ that are both implementable and potentially implementable is precisely the algebra $\bT_{{\rm dec}\bF}^{\cu^{-1}}$. We refer to the operators in $\bT_{{\rm dec}\bF}^{\cu^{-1}}$ as \emph{dual decohered}. Just as decoherence can be understood as the potential for the proliferation of information about observables, dual decoherence can be understood as the potential for the \textit{redundant implementation} of generators by many fragments of the environment.

Decoherence and dual decoherence each pick out a preferred commutative degree of freedom in an analogous way, namely by partitioning a system algebra into two subalgebras that are each other's commutant. This is no coincidence, but immediately follows from the fact that each version of decoherence is equivalent to the other up to the reversal of the unitary dynamics $\cu \mapsto \cu^{-1}$. If this reversal is seen as a change in how the physics is described rather than a change in the physics itself, then the difference between decoherence and dual decoherence must also be seen as mere convention.

Given this symmetry, and the widespread belief that decoherence plays a crucial role in the emergence of classicality, one must consider the possibility that dual decoherence plays a similarly important role. The question, then, is just what this role might be. Our answer will be: in much the same way that the outcomes of quantum measurements emerge from decoherence, quantum states emerge from dual decoherence. 

While a full understanding of how states arise from dual decoherence will not be possible until the next section, we can at this stage sketch a rough intuition for why one might expect dual decoherence to be associated with state preparation. 
Just as macroscopic observers can only perceive microscopic objects by means of the imprints the objects leave on their environment, such observers can only \textit{control} microscopic objects by controlling the objects' surrounding environment. 
To prepare a microscopic object in a particular state, the observer must act on that object, but the observer can only act on the object using tools that are comparatively vast, and composed of a great number of smaller fragments. Macroscopic observers have not only large eyes, but fat fingers.

Therefore, for a macroscopic observer to prepare a given state, it is not enough for a single small fragment to influence the system in whatever way is appropriate for preparing that state. Many independent fragments of the environment must influence the system the same way. Only then will the fragments serve as a large enough \enquote{handle} for the observer to turn, by which we mean that the ability of the fragments to prepare a state will translate into the ability of the observer to prepare that state.

This leaves the question of what sort of influence from a single fragment to the system is required to assist with the preparation of a state in a given basis $\{\ket*{\tilde \psi_k}_\bT\}_k$. On the account that we will develop, the relevant sort of influence is one that enables the fragment to implement the generators that are diagonal in this basis. This is a natural requirement because these generators may be used to randomize precisely the degrees of freedom that are conjugate to this basis.

Taken together, the remarks in the last two paragraphs suggest that whereas the proliferation of information is required for the perception of outcomes by a macroscopic observer, the redundant implementation of generators is required for the preparation of states by a macroscopic observer.
More precisely, for a given unitary interaction $\cu: \bS \bF \to \bT\bG$ to assist with the preparation of a state in $\{\ket*{\tilde \psi_k}_\bT\}_k$, it should make the generators diagonal in this basis implementable by $\bF$ while also leaving them potentially implementable via other fragments, such as $\bH$ in \cref{eq:potential_imp}. That is to say, $\cu$ should dually decohere these generators.

We close this section by briefly remarking on the need for two versions of decoherence rather than just one. We aim to explain the classicality of our experiences, but those experiences consist not only of perceptions, but also of actions; we are not merely observers, but also agents. Decoherence relates to perception, and dual decoherence to action. Decoherence is defined in terms of influences from system to environment required for proliferation, and will thus help us explain how we come to observe outcomes. Dual decoherence is defined in terms of influences from environment to system required for redundant implementation, and will thus help us explain how we prepare states.

Quantum Darwinism emphasises the importance of the environment’s role as a communication channel between system and observer \cite{Ollivier_2005}; here, we are further emphasising its status as a \textit{two-way} channel. As the next section will demonstrate, this observation is key to unifying the two major formalisms for decoherence: environmentally induced decoherence and consistent histories.

\section{Outcomes from decoherence, states from dual decoherence} \label{sec:circuit}

The program of environmentally induced decoherence developed in e.g.\ \cite{zeh1970interpretation, zeh1973toward, zurek1981pointer, zurek1982environment} tends to focus on how a certain classical degree of freedom associated with a particular system emerges from a particular interaction with its environment.  So far, this has also been true of the causal approach that we have been developing, which has focused on a single bipartite interaction $\cu$ between a system and a fragment of the environment. 
As a result, if we were to stop here, then it would not generally be  clear how, or in what circumstances, a single consistent story can be told about the emergence of \textit{many} classical degrees of freedom, where each one might be associated with a different system and a different interaction.

The consistent histories formalism \cite{griffiths1984consistent, griffiths2003consistent, gell1990complexity}, on the other hand, aptly describes histories involving many degrees of freedom associated with different systems, times, and places. However, there are vastly many consistent history sets, most of which are neither related to environmentally induced decoherence nor endowed with any obvious physical significance \cite{dowker1996consistent}.

The two approaches thus have complementary strengths: environmentally induced decoherence specializes in identifying physically significant emergent degrees of freedom, but is yet to provide a general way of “gluing them together” into a single consistent description, whereas the consistent histories formalism provides a general framework for describing many degrees of freedom across different systems, times, and places, but does not itself single out physically distinguished history sets.

This section will use the causal approach to show how a consistent history set is singled out by environmentally induced decoherence in a highly general setting. These histories are built out of many degrees of freedom associated with different systems, arising due to different decohering interactions. This result synthesizes the two approaches to decoherence within a causal framework, resulting in a powerful, unified formalism for decoherence.

Performing this synthesis requires us to move beyond the simple setting of a single bipartite unitary interaction. Our new setting will be a quantum \textit{circuit}, defined by a composition of a finite number of unitary interactions. In precise terms, what we will show is that for any subset $\mathfrak{B}$ of systems in an arbitrary circuit and every $\bX \in \mathfrak{B}$, two preferred subspace decompositions of $\ch_\bX$ are selected---one by decoherence and another by dual decoherence. Each of these subspace decompositions can be thought of as a collection of possible events, where each event corresponds to a preferred subspace of $\ch_\bX$ (or, equivalently, to the projector onto that subspace). A complete list of events, two for each system, is a \textit{history}. Unlike general history sets in quantum theory, the history set so derived from (dual) decoherence is guaranteed to be \textit{consistent} \cite{griffiths1984consistent, griffiths2003consistent, gell1990complexity}, meaning that the classical sum rule is respected, i.e.\ probabilities for mutually exclusive events are additive.

As anticipated in the previous section, we will see through examples that the events associated with decoherence are often naturally thought of as \textit{outcomes} (or records of outcomes) while those associated with dual decoherence are often naturally thought of as \textit{states}. Since dual decoherence is defined purely in terms of the unitary dynamics, this provides a precise sense in which quantum states emerge from quantum dynamics. And since dual decoherence supplies half of the events that make up a history, the idea that quantum states are not fundamental, but emergent, that they are symptoms of the process of decoherence rather than part of the definition, plays an essential role in the unification of environmentally induced decoherence and consistent histories.

The main result of this section serves as a simplification of a procedure introduced in \cite{ormrod2024quantum} for deriving consistent histories from a unitary circuit (in which the connection with decoherence was not made explicit). \Cref{sec:crqt} and \cref{app:bubble} will elaborate further on the relation between this work and \cite{ormrod2024quantum}.

This section begins with an example of a single measurement and the general derivation of consistent history sets from decoherence. We then go on to study two more examples, namely a sequence of two incompatible measurements and the Wigner's friend scenario \cite{wigner1995remarks}. Finally, we discuss what sort of relationalism is suggested by decoherence, and close the section with a summary of the account of emergent classicality suggested by our analysis.

\subsection{A single measurement: circuit and preferred subspaces}
\label{sec:singlemmt}

Since a huge part of the motivation for decoherence comes from measurement, it is natural to start with a model of a measurement.
As will become clear from our analysis, a minimal model of a quantum measurement is given by the circuit
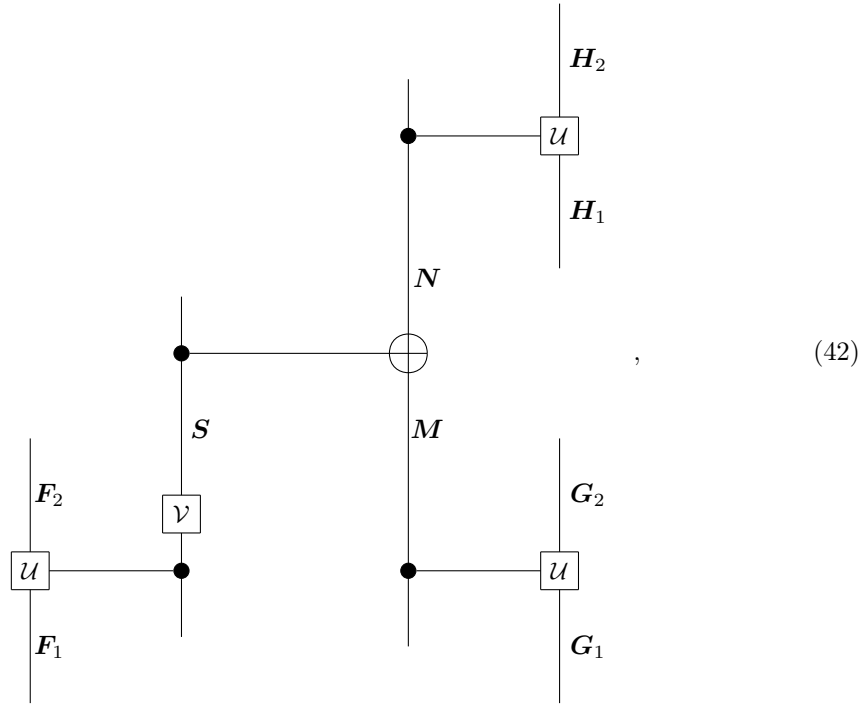
\begin{equation} \label{eq:measurement}
    \begin{tikzpicture}
	\begin{pgfonlayer}{nodelayer}
		\node [style=none] (0) at (17.75, -7.75) {};
		\node [style=none] (1) at (8.25, -7.75) {$\bF_1$};
		\node [style=none] (2) at (17.75, 7.25) {};
		\node [style=none] (3) at (11.75, -3.75) {};
		\node [style=none] (4) at (12.25, -2) {$\bS$};
		\node [style=none] (5) at (18.25, 2) {$\bN$};
		\node [style=none] (6) at (11.75, 1.5) {};
		\node [style=none] (7) at (7.75, -6.25) {};
		\node [style=none] (8) at (7.75, -9.25) {};
		\node [style=none] (9) at (18.25, -2) {$\bM$};
		\node [style=black dot] (10) at (11.75, 0) {};
		\node [style=none] (11) at (18.25, 0) {};
		\node [style=none] (12) at (17.25, 0) {};
		\node [style=none] (13) at (18.25, 0) {};
		\node [style=none] (14) at (17.25, 0) {};
		\node [style=black dot] (15) at (11.75, -5.75) {};
		\node [style=none] (16) at (8.25, -5.75) {};
		\node [style=none] (17) at (8.25, -5.25) {};
		\node [style=none] (18) at (8.25, -6.25) {};
		\node [style=none] (19) at (7.25, -6.25) {};
		\node [style=none] (20) at (7.25, -5.25) {};
		\node [style=none] (21) at (7.75, -5.25) {};
		\node [style=none] (22) at (7.75, -2.25) {};
		\node [style=none] (23) at (7.75, -5.75) {\small $\cu$};
		\node [style=none] (24) at (11.25, -3.75) {};
		\node [style=none] (25) at (12.25, -3.75) {};
		\node [style=none] (26) at (12.25, -4.75) {};
		\node [style=none] (27) at (11.25, -4.75) {};
		\node [style=none] (28) at (11.75, -4.25) {\small $\cv$};
		\node [style=none] (29) at (11.75, -4.75) {};
		\node [style=none] (30) at (11.75, -7.5) {};
		\node [style=none] (31) at (21.75, -6.25) {};
		\node [style=none] (32) at (21.75, -9.25) {};
		\node [style=none] (33) at (22.5, -7.75) {$\bG_1$};
		\node [style=black dot] (34) at (17.75, -5.75) {};
		\node [style=none] (35) at (21.25, -5.75) {};
		\node [style=none] (36) at (21.25, -5.25) {};
		\node [style=none] (37) at (21.25, -6.25) {};
		\node [style=none] (38) at (22.25, -6.25) {};
		\node [style=none] (39) at (22.25, -5.25) {};
		\node [style=none] (40) at (21.75, -5.25) {};
		\node [style=none] (41) at (21.75, -2.25) {};
		\node [style=none] (42) at (21.75, -5.75) {\small $\cu$};
		\node [style=none] (43) at (21.75, 5.25) {};
		\node [style=none] (44) at (21.75, 2.25) {};
		\node [style=none] (45) at (22.5, 7.75) {$\bH_2$};
		\node [style=black dot] (46) at (17.75, 5.75) {};
		\node [style=none] (47) at (21.25, 5.75) {};
		\node [style=none] (48) at (21.25, 6.25) {};
		\node [style=none] (49) at (21.25, 5.25) {};
		\node [style=none] (50) at (22.25, 5.25) {};
		\node [style=none] (51) at (22.25, 6.25) {};
		\node [style=none] (52) at (21.75, 6.25) {};
		\node [style=none] (53) at (21.75, 9.25) {};
		\node [style=none] (54) at (21.75, 5.75) {\small $\cu$};
		\node [style=none] (55) at (6.75, -4.5) {};
		\node [style=none] (56) at (17.75, 2.5) {};
		\node [style=none] (57) at (17.75, -2.5) {};
		\node [style=none] (58) at (8.25, -3.75) {$\bF_2$};
		\node [style=none] (59) at (22.5, -3.75) {$\bG_2$};
		\node [style=none] (60) at (22.5, 3.75) {$\bH_1$};
	\end{pgfonlayer}
	\begin{pgfonlayer}{edgelayer}
		\draw (10) to (11.center);
		\draw [bend left=90, looseness=1.75] (12.center) to (11.center);
		\draw [bend right=90, looseness=1.75] (14.center) to (13.center);
		\draw (0.center) to (2.center);
		\draw (15) to (16.center);
		\draw (3.center) to (6.center);
		\draw (7.center) to (8.center);
		\draw (20.center) to (19.center);
		\draw (17.center) to (18.center);
		\draw (18.center) to (19.center);
		\draw (20.center) to (17.center);
		\draw (22.center) to (21.center);
		\draw (27.center) to (24.center);
		\draw (24.center) to (25.center);
		\draw (25.center) to (26.center);
		\draw (26.center) to (27.center);
		\draw (30.center) to (29.center);
		\draw (34) to (35.center);
		\draw (31.center) to (32.center);
		\draw (39.center) to (38.center);
		\draw (36.center) to (37.center);
		\draw (37.center) to (38.center);
		\draw (39.center) to (36.center);
		\draw (41.center) to (40.center);
		\draw (46) to (47.center);
		\draw (43.center) to (44.center);
		\draw (51.center) to (50.center);
		\draw (48.center) to (49.center);
		\draw (49.center) to (50.center);
		\draw (51.center) to (48.center);
		\draw (53.center) to (52.center);
	\end{pgfonlayer}
\end{tikzpicture} \quad,
\end{equation}
which consists of qubit systems, a CNOT, a local unitary channel $\cv$, and three coherent control channels, each corresponding to a unitary operator of the form $\ketbra{0} \otimes U^{(0)} + \ketbra{1} \otimes U^{(1)}$.
We assume nothing about $U^{(0)}$ and $ U^{(1)}$ except that they are neither equal to nor equivalent up to a phase. By \cref{thm:uni_control}, this is equivalent to the assumption that the control interactions are maximally decohering.

Notice that in \cref{eq:measurement}, not all systems are labelled, but only those in the subset
\begin{equation} \label{eq:bubble}
    \mathfrak{B} = \{\bF_1, \bF_2, \bG_1, \bG_2, \bS, \bM, \bN, \bH_1, \bH_2\}.
\end{equation}
These are our \textit{systems of interest}. This means that we will investigate only how preferred subspace decompositions emerge from the causal relations between members of $\mathfrak{B}$, rather than from more general causal relations. We will come to the thorny question of why the analysis requires a designated set of systems of interest and what this tells us about the emergence of classicality later on, in \cref{subsec:relationalism}. In the meantime, we simply take for granted that $\mathfrak{B}$ is the designated set.

A brief remark on how the causal relations in this circuit are interpreted is in order. Since the circuit is unitary, it can be read either from bottom to top, or from top to bottom (inverting each gate), and the influences between the systems in the circuit can equivalently be thought of as flowing upwards or downwards (by \cref{eq:reversibility_inf}). We will choose to think of influences as flowing upwards, although we will emphasise throughout that this is merely a convention, and not forced by the circuit itself.

We start by deriving preferred subspace decompositions for $\bS$. There will be two: the first arises from how $\bS$ influences systems above it (decoherence), the second from how it is influenced by systems below it (dual decoherence). Unlike in previous sections, here there are \textit{two} relevant systems that are influenced by $\bS$: $\bN$ and $\bH_2$. These two systems are themselves causally related, meaning that they cannot be combined via the usual tensor product to form a single system. Does this compel us to apply our analysis of decoherence separately, obtaining one decomposition from the influence of $\bS$ on $\bN$, and another one from the influence of $\bS$ on $\bH_2$?

It does not: if we work in the Heisenberg picture, then even causally related systems can be combined into a larger system. To shift to the Heisenberg picture, we pull back the algebra of each system through all of the unitaries in the circuit that come below it, until it reaches the overall five-qubit input space of the circuit. (Or we could equally well push forward onto the output space; what matters is just that the algebras all land on the \textit{same} Hilbert space). We denote the pulled-back algebras with tildes, e.g., $\tilde \bN$. Having done this, we take the algebra
\begin{equation}
    \tilde \bN \lor \tilde \bH_2
\end{equation}
generated by $\tilde \bN$ and $\tilde \bH_2$ to represent the joint system.

\Cref{sec:channel} defined decoherence as the conjunction of accessibility and potential accessibility. What are the relevant accessible and potentially accessible algebras here? We begin by determining the subalgebra of $\tilde \bS$ left potentially accessible by $\tilde \bN \lor \tilde \bH_2$ to other fragments of the environment that do not directly interact with $\tilde \bN \lor \tilde \bH_2$---i.e.\ fragments that only interact with the commutant $(\tilde \bN \lor \tilde \bH_2)'$.
(This parallels the discussion in \cref{sec:channel} around \cref{eq:potential_acc}, in which we asked what was left potentially accessible after the unitary interaction $\cu: \bS\bF \rightarrow \bT \bG$ by $\bG$ to a fragment $\bI$ that did not directly interact with $\bG$, but only with $\bT = \bG'$.) 

\Cref{thm:pacc_ninf} established that this potentially accessible algebra is equivalent to the subalgebra of $\tilde \bS$ that does not influence $\tilde \bN \lor \tilde \bH_2$.\footnote{\Cref{thm:pacc_ninf} only explicitly covers the case where $\tilde \bN \lor \tilde \bH_2$ is a factor algebra, but can be straightforwardly generalized to the nonfactor case.} We therefore write
\begin{equation}\label{eq:circuit-pacc-of-s}
    \tilde \bS^\uparrow_{{\rm pacc}\mathfrak{B}} = (\tilde \bN \lor \tilde \bH_2)' \cap \tilde \bS.
\end{equation}
The $\uparrow$ superscript reminds us that this algebra is interpreted as potentially accessible only because we have chosen to view causal influences as flowing upwards. The $\mathfrak{B}$ subscript then tells us that this is the algebra left potentially accessible by the combination of all systems in the future of $\bS$ within $\mathfrak{B}$. It helps to represent these operators as a subalgebra of $\bS$ itself (rather than $\tilde \bS$) by pushing them forward through each unitary in the circuit that comes below $\bS$:
\begin{equation} \label{eq:measurement_pacc}
    \begin{split}
        \bS^\uparrow_{{\rm pacc}\mathfrak{B}} 
        &= {\rm Push} (\tilde \bS^\uparrow_{{\rm pacc}\mathfrak{B}}) \\
        & = \gen{Z_\bS} \subseteq \bS.
    \end{split}
\end{equation}

In \cref{sec:channel}, it was shown that the accessible algebra is the commutant of the potentially accessible algebra, and the decoherent algebra was defined as their intersection. We thus obtain
\begin{subequations}
\begin{align}
    \bS^\uparrow_{{\rm acc}\mathfrak{B}}
        &= (\bS_{{\rm pacc}\mathfrak{B}}^\uparrow)' \cap \bS \\
        &= \gen{Z_\bS},
        \\
    \bS^\uparrow_{{\rm dec}\mathfrak{B}}
        &= \bS_{{\rm acc}\mathfrak{B}}^\uparrow \cap \bS_{{\rm pacc}\mathfrak{B}}^\uparrow \\
        &= \gen{Z_\bS},
        \\
    \ch_\bS
        &= \mathbb{C}  \ket 0_\bS \oplus \mathbb{C} \ket 1_\bS.
        \label{eq:s_up_decomp}
\end{align}
\end{subequations}
$\bS_{{\rm acc}\mathfrak{B}}^\uparrow$ is then to be interpreted as the algebra of operators accessible to the combination of all systems in the future of $\bS$ within $\mathfrak{B}$ (that is, to $\tilde \bN \lor \tilde \bH_2$), while $\bS^\uparrow_{{\rm dec}\mathfrak{B}}$ consists of the operators decohered by these systems. 
The $\uparrow$ superscript on $\bS^\uparrow_{{\rm dec}\mathfrak{B}}$ reminds us that, relative to our choice to view influence as travelling upwards, $\bS^\uparrow_{{\rm dec}\mathfrak{B}}$ is the decoherent algebra and not the dual decoherent one. It is not surprising that in this example we recover the same preferred subspace decomposition \cref{eq:s_up_decomp} as in \Cref{ex:balance_c}, since here $\bS$ is again the control input of a CNOT, while $\bN$ is the target output.

In a similar way, another subspace decomposition of $\bS$ is selected by \textit{dual} decoherence, i.e.\ by the way that $\bS$ is influenced by the systems of interest \textit{below} it. Since the only system of interest that influences $\bS$ is $\bF_1$, it is derived as follows, where $\left\{\ket {\psi^0}_\bS, \ket{\psi^1}_\bS\right\}$ is the basis to which $\cv$ maps the $\{\ket0,\ket1\}$-basis of its (unlabelled) input system, and $\psi_\bS^0 := \ketbra{\psi^0}_\bS$:
\begin{subequations}
\begin{align}
    \bS_{{\rm pacc}\mathfrak{B}}^\downarrow
        &= {\rm Push}(\tilde \bF_1' \cap \tilde \bS) \\
        &= \gen{ \psi_\bS^0}, 
        \\
    \bS_{{\rm acc}\mathfrak{B}}^\downarrow
        &= (\bS_{{\rm pacc}\mathfrak{B}}^\downarrow)' \cap \bS \\
        &= \gen{ \psi_\bS^0},
        \\
    \bS_{{\rm dec}\mathfrak{B}}^\downarrow
        &= \bS_{{\rm acc}\mathfrak{B}}^\downarrow \cap \bS_{{\rm pacc}\mathfrak{B}}^\downarrow \\
        &= \gen{ \psi_\bS^0},
        \\
    \ch_\bS
        &= \mathbb{C}  \ket {{\psi^0}}_\bS \oplus \mathbb{C} \ket {{\psi^1}}_\bS.
        \label{eq:s_down_decomp}
\end{align}
\end{subequations}
The $\downarrow$ superscripts remind us that if we had decided to view causal influences as flowing downwards, then these algebras would have been interpreted in terms of (potential) accessibility and decoherence. But since we have in fact chosen to view influences as flowing upwards, we instead interpret them in terms of (potential) implementability and dual decoherence. On this interpretation, $\bS_{{\rm dec}\mathfrak{B}}^\downarrow$ is the algebra of generators on $\bS$ that are both implementable by systems within $\mathfrak{B}$ below $\bS$ (in this case, just $\bF_1$), and potentially implementable by other systems that do not interact with those systems directly.

The above generalizes to define two preferred subspace decompositions for every system of interest $\bA \in \mathfrak{B}$. One originates from the algebra $ \bA^\uparrow_{{\rm pacc}\mathfrak{B}}$ (along with the corresponding
$\bA^\uparrow_{{\rm acc}\mathfrak{B}}$
and
$\bA^\uparrow_{{\rm dec}\mathfrak{B}}$), which is always defined as the algebra of all $\bA$-operators that commute with all systems that come above $\bA$ in the circuit---where \enquote{above} means that there is an upward-directed path from $\bA$ to that system (cf.\ \eqref{eq:circuit-pacc-of-s}).\footnote{An alternative, and perhaps better motivated, definition of $ \bA^\uparrow_{{\rm pacc}\mathfrak{B}}$ would consider only systems above $\bA$ that also \textit{do not commute} with $\bA$, i.e.\ that are in the \textit{causal future} of $\bA$ rather than simply above $\bA$ in the diagram. But these two definitions are equivalent because they only differ with respect to the consideration of algebras that the full $\bA$ commutes with anyway.}
The other subspace decomposition originates from the algebra $ \bA^\downarrow_{{\rm pacc}\mathfrak{B}}$ that commutes with the algebra generated by all systems in $\mathfrak{B}$ that come \textit{below} $\bA$. Given that influences are thought of as flowing upwards, the first of these subspace decompositions arises from decoherence and the second from dual decoherence.

Applying this procedure to every system of interest in our circuit, one obtains the following algebras, some of which the reader may like to verify. Here $\star$ denotes the trivial algebra $\mathbb{C}I$.

{\allowdisplaybreaks
\begin{subequations} \label{eq:single-mmt-all-the-algebras} \small 
\begin{align}
\tilde \bF^\downarrow_{1 {\rm pacc}\mathfrak{B}} &= \star' \cap \tilde \bF_1
&\qquad
\bF^\downarrow_{1{\rm dec}\mathfrak{B}} &= \star
\\ 
\tilde \bF^\uparrow_{1{\rm pacc}\mathfrak{B}} &= (\tilde \bF_2 \lor \tilde \bS \lor \tilde \bN \lor \tilde\bH_2)' \cap \tilde \bF_1
&
\bF^\uparrow_{1{\rm dec}\mathfrak{B}} &= \star
\\[4pt]  
\tilde \bF^\downarrow_{2 {\rm pacc}\mathfrak{B}} &= \tilde \bF_1' \cap \tilde \bF_2
&\qquad
\bF^\downarrow_{2{\rm dec}\mathfrak{B}} &= \star
\\
\tilde \bF^\uparrow_{2{\rm pacc}\mathfrak{B}} &= \star' \cap \tilde \bF_2
&
\bF^\uparrow_{2{\rm dec}\mathfrak{B}} &= \star
\\[4pt] 
\tilde \bG^\downarrow_{1{\rm pacc}\mathfrak{B}} &= \star' \cap \tilde \bG_1
&
\bG^\downarrow_{{1\rm dec}\mathfrak{B}} &= \star
\\
\tilde \bG^\uparrow_{1{\rm pacc}\mathfrak{B}} &= (\tilde \bG_2 \lor  \tilde \bM \lor \tilde \bN \lor \tilde\bH_2)' \cap \tilde \bG_1
&
\bG^\uparrow_{1{\rm dec}\mathfrak{B}} &= \star
\\[4pt] 
\tilde \bG^\downarrow_{2{\rm pacc}\mathfrak{B}} &= \tilde \bG_1' \cap \tilde \bG_2
&
\bG^\downarrow_{2{\rm dec}\mathfrak{B}} &= \star
\\
\tilde \bG^\uparrow_{2{\rm pacc}\mathfrak{B}} &= \star' \cap \tilde \bG_2
&
\bG^\uparrow_{2{\rm dec}\mathfrak{B}} &= \star
\\[4pt] 
\tilde \bS^\downarrow_{{\rm pacc}\mathfrak{B}} &= \tilde \bF_1' \cap \tilde \bS
&
\bS^\downarrow_{{\rm dec}\mathfrak{B}} &= \gen{\psi_\bS^0}
\\
\tilde \bS^\uparrow_{{\rm pacc}\mathfrak{B}} &= (\tilde \bN \lor \tilde \bH_2)' \cap \tilde \bS
&
\bS^\uparrow_{{\rm dec}\mathfrak{B}} &= \gen{Z_\bS}
\\[4pt]
\tilde \bM^\downarrow_{{\rm pacc}\mathfrak{B}} &= \tilde \bG_1' \cap \tilde \bM
&
\bM^\downarrow_{{\rm dec}\mathfrak{B}} &=  \langle Z_\bM \rangle
\\
\tilde \bM^\uparrow_{{\rm pacc}\mathfrak{B}} &= (\tilde \bN \lor \tilde \bH_2)' \cap \tilde \bM
&
\bM^\uparrow_{{\rm dec}\mathfrak{B}} &= \star
\\[4pt]
\tilde \bN^\downarrow_{{\rm pacc}\mathfrak{B}} &= (\tilde \bF_1 \lor \tilde \bG_1 \lor \tilde \bS \lor \tilde \bM)' \cap \tilde \bN
&
\bN^\downarrow_{{\rm dec}\mathfrak{B}} &= \star
\\
\tilde \bN^\uparrow_{{\rm pacc}\mathfrak{B}} &= \tilde \bH_2' \cap \tilde \bN
&
\bN^\uparrow_{{\rm dec}\mathfrak{B}} &= \langle Z_\bN \rangle
\\[4pt] 
\tilde \bH^\downarrow_{1{\rm pacc}\mathfrak{B}} &= \star' \cap \tilde \bH_1  
&
\bH^\downarrow_{1{\rm dec}\mathfrak{B}} &= \star
\\
\tilde \bH^\uparrow_{1{\rm pacc}\mathfrak{B}} &= \tilde \bH_2' \cap \tilde \bH_1
&
\bH^\uparrow_{1{\rm dec}\mathfrak{B}} &= \star
\\[4pt] 
\tilde \bH^\downarrow_{2{\rm pacc}\mathfrak{B}} &= (\tilde \bF_1 \lor \tilde \bG_1 \lor \tilde \bS \lor \tilde \bM \lor \tilde \bN)' \cap \tilde \bH_2
&
\bH^\downarrow_{2{\rm dec}\mathfrak{B}} &= \star
\\
\tilde \bH^\uparrow_{2{\rm pacc}\mathfrak{B}} &= \star' \cap \tilde \bH_2
&
\bH^\uparrow_{2 {\rm dec} \mathfrak{B} } &= \star
\end{align}
\end{subequations}%
}%
While it may look as though things have gotten a little out of hand, note that only four of the $2 \times 9 = 18$ preferred subspace decompositions end up being nontrivial. Explicitly:
\begin{subequations}\label{eq:decom4}
    \begin{align}
        \ch_\bS &=  \mathbb{C}\ket{0}_\bS \oplus \mathbb{C}\ket{1}_\bS;  \\  
        \ch_\bS &=  \mathbb{C}\ket{\psi^0}_\bS \oplus \mathbb{C}\ket{\psi^1 }_\bS; \\
        \ch_\bM &=  \mathbb{C}\ket{0}_\bM \oplus \mathbb{C}\ket{1}_\bM; \\
        \ch_\bN &=  \mathbb{C}\ket{0}_\bN \oplus \mathbb{C}\ket{1}_\bN.
    \end{align}
\end{subequations}

In standard quantum theory, a subspace decomposition of a Hilbert space is associated with a projective-valued measurement, and each individual subspace is associated with an outcome. More generally, we can think of a subspace decomposition as defining a space of possible \textit{events}, where an event is just something that happens, and is not necessarily interpreted as the outcome of a measurement or anything witnessed by an observer. A list of four events, where each one is associated with a subspace from a different decomposition in \cref{eq:decom4}, is a \textit{history}. Since each of the four decompositions in \cref{eq:decom4} has two subspaces, we have thus derived a set of $2^4=16$ possible histories.

Having derived a history set, it is natural to ask whether there exists a natural probability distribution over the histories. It turns out that there is---and not just in this simple example, but for any unitary circuit, and any $\mathfrak{B}$.

\subsection{A general derivation of a consistent history set from decoherence}

To demonstrate this, we step back from the present example and denote by $\mathfrak{B} = \{\bA_1, \ldots, \bA_n\}$ an arbitrary set of systems in an arbitrary unitary circuit. $\tilde \pi_{\bA_k}^{\uparrow e_k} \in \tilde \bA_{{\rm dec}\mathfrak{B}}^{\uparrow}$ denotes the (pulled-back) projector onto the $e_k$th preferred subspace of $\bA_k$ selected by decoherence with respect to the $\mathfrak{B}$-systems above $\bA_k$. Similarly, $\tilde \pi_{\bA_k}^{\downarrow f_k}\in \tilde \bA_{{\rm dec}\mathfrak{B}}^{\downarrow}$ denotes the (pulled-back) projector onto the $f_k$th preferred subspace of $\bA_k$ selected by \textit{dual} decoherence with respect to the $\mathfrak{B}$-systems below $\bA_k$. A $2n$-tuple $(f_1, e_1, \ldots, f_n, e_n)$ denotes a history.

\begin{proposition}
    All decohered projectors commute with each other, and likewise all dual-decohered projectors commute with each other:
    \begin{subequations} \label{eq:comm_rels}
    \begin{align}
        \forall k, l, e_k, e_l, \qquad [\tilde \pi_{\bA_k}^{\uparrow e_k}, \tilde \pi_{\bA_l}^{\uparrow e_l}] = 0;  
        \label{eq:comm_rels_a}\\
        \forall k, l, f_k, f_l, \qquad [\tilde \pi_{\bA_k}^{\downarrow f_k}, \tilde \pi_{\bA_l}^{\downarrow f_l}] = 0.
        \label{eq:comm_rels_b}
    \end{align}
\end{subequations}
As a result, the operators
\begin{subequations} \label{eq:fine_grained_proj}
    \begin{align}
            \tilde \pi_{\mathfrak{B}}^{\uparrow e_1\ldots e_n} &:= \tilde \pi_{\bA_1}^{\uparrow e_1} \cdots \tilde \pi_{\bA_n}^{\uparrow e_n}, \\
        \tilde \pi_{\mathfrak{B}}^{\downarrow f_1\ldots f_n} &:= \tilde \pi_{\bA_1}^{\downarrow f_1} \cdots \tilde \pi_{\bA_n}^{\downarrow f_n} 
    \end{align}
\end{subequations}
form two complete families of orthogonal projectors $\{\tilde \pi_{\mathfrak{B}}^{\uparrow e_1 \ldots e_n}\}_{e_1\ldots e_n}$ and $\{\tilde \pi_{\mathfrak{B}}^{\downarrow f_1 \ldots f_n}\}_{f_1\ldots f_n}$.
\end{proposition}

\begin{proof}
    Assume without loss of generality that the subscript labelling on the $\bA_k$ is compatible with the temporal ordering of the circuit in the sense that if there is an upward path of wires from $\bA_k$ to $\bA_l$ then $k < l$. By the definition of the decohered algebras, $\tilde \bA_{k {\rm dec}\mathfrak{B}}^\uparrow \subseteq \tilde \bA_{k {\rm pacc}\mathfrak{B}}^\uparrow$ for all $k$.  
    For any $k$ and $l > k$, there are two possibilities: there is no path of wires between $\bA_k$ and $\bA_l$, or there is an upward path  from $\bA_k$ to $\bA_l$. In the first case, $\tilde \bA_k \subseteq \tilde \bA_l'$, and in the second, the definition of $\tilde \bA_{k {\rm pacc}\mathfrak{B}}^\uparrow$ implies that $\tilde \bA_{k {\rm pacc}\mathfrak{B}}^\uparrow\subseteq  \tilde \bA_l'$. Hence in either case,  $\tilde \bA_{k {\rm dec}\mathfrak{B}}^\uparrow\subseteq  \tilde \bA_l'$. But $\tilde \bA_l'\subseteq (\tilde \bA_{l {\rm dec}\mathfrak{B}}^\uparrow)'$, entailing that $\tilde \bA_{k {\rm dec}\mathfrak{B}}^\uparrow \subseteq (\tilde \bA_{l {\rm dec}\mathfrak{B}}^\uparrow)'$. This implies \cref{eq:comm_rels_a}. \cref{eq:comm_rels_b} is proven by a dual argument.
\end{proof}

The complete orthogonal sets $\{\tilde \pi_{\mathfrak{B}}^{\downarrow f_1 \ldots f_n}\}_{f_1\ldots f_n}$ and $\{\tilde \pi_{\mathfrak{B}}^{\uparrow e_1 \ldots e_n}\}_{e_1\ldots e_n}$ each define a preferred subspace decomposition of the overall Hilbert space, respectively arising from decoherence and dual decoherence. We can therefore associate a complete history $(f_1, e_1, \ldots, f_n, e_n)$ with just two projectors, $(\tilde \pi_{\mathfrak{B}}^{\downarrow f_1 \ldots f_n}, \tilde \pi_{\mathfrak{B}}^{\uparrow e_1 \ldots e_n})$. It follows that the set $\{(f_1, e_1, \ldots, f_n, e_n)\}$ of all possible histories admits a natural probability rule:

\begin{theorem} \label{thm:prob}
Let $\mathfrak{B} = \{\bA_1, \ldots, \bA_n\}$ be a subset of the systems in an arbitrary unitary circuit. The projectors $\tilde \pi_{\mathfrak{B}}^{\downarrow f_1 \ldots f_n}$ and $
        \tilde \pi_{\mathfrak{B}}^{\uparrow e_1 \ldots e_n}$  from the $2n$ decoherent and dual decoherent algebras define the probability distribution
\begin{equation} \label{eq:prob}
      {\rm Prob}_\mathfrak{B}(f_1, e_1, \ldots, f_n, e_n)
    \coloneqq  \frac{1}{d} {\rm Tr}\!\bigl(\tilde \pi_{\bA_1}^{\downarrow f_1} \cdots \tilde \pi_{\bA_n}^{\downarrow f_n}   \tilde \pi_{\bA_1}^{\uparrow e_1} \cdots \tilde \pi_{\bA_n}^{\uparrow e_n}  \bigr),  
\end{equation}
where $d$ is the dimension of the overall Hilbert space. 
\end{theorem}

\begin{proof} Substituting \cref{eq:fine_grained_proj} into \cref{eq:prob} gives
\begin{equation}
    {\rm Prob}_\mathfrak{B}(f_1, e_1, \ldots, f_n, e_n) = \frac{1}{d} {\rm Tr}\!\bigl(
        \tilde \pi_{\mathfrak{B}}^{\downarrow f_1 \ldots f_n}
        \tilde \pi_{\mathfrak{B}}^{\uparrow e_1 \ldots e_n} \bigr).
\end{equation}
     To verify that this is indeed a probability distribution, note that $\{\tilde \pi_{\mathfrak{B}}^{\downarrow f_1 \ldots f_n}\}_{f_1\ldots f_n}$ and $\{\tilde \pi_{\mathfrak{B}}^{\uparrow e_1 \ldots e_n}\}_{e_1\ldots e_n}$ are each a set of positive semidefinite operators that sum to the identity. It follows that $\{\frac{1}{d} {\rm Tr}\!\bigl(
        \tilde \pi_{\mathfrak{B}}^{\downarrow f_1 \ldots f_n}
        \tilde \pi_{\mathfrak{B}}^{\uparrow e_1 \ldots e_n}
    \bigr)\}_{f_1, e_1, \ldots, f_n, e_n}$ is a set of nonnegative numbers that sum to 1.
\end{proof}

We will soon discuss how to interpret this probability distribution within our example of a single measurement. We first note that \cref{eq:prob} also entails that the history set $\{(f_1, e_1, \ldots, f_n, e_n)\}$ obeys the classical sum rule and is therefore \textit{consistent} in the sense of the  consistent histories formalism \cite{griffiths1984consistent, griffiths2003consistent, gell1990complexity}, as the following remark explains further.

\begin{remarkbox} \label{remark:prob} In general, a quantum history set can violate the classical sum rule, which requires that the probability for the disjunction of a pair of mutually exclusive events is the sum of the individual probabilities. For example, consider the probability distribution associated with two successive projective-valued measurements $\{P_1^{e_1}\}_{e_1}$ and $\{P_2^{e_2}\}_{e_2}$ on a system initialized in the state $\rho$:
\begin{equation} \label{eq:consistent_ex_prob}
    {\rm Prob}(e_1, e_2) = {\rm Tr}(P_2^{e_2}P_1^{e_1} \rho P_1^{e_1} P_2^{e_2}).
\end{equation}
For any pair of distinct events $e_1$ and $e_1'$, the disjunction $e_1 \lor e_1'$ is associated with the projector
\begin{equation}
    P_1^{e_1 \lor e_1'} =  P_1^{e_1} + P_1^{e_1'}.
\end{equation}
For a generic choice of the state and measurements, we are likely to find that the probabilities are non-additive in the sense that
\begin{equation}
    \begin{split}
        {\rm Prob}(e_1 \lor e_1', e_2) &= {\rm Tr}(P_2^{e_2}P_1^{e_1 \lor e_1'} \rho P_1^{e_1 \lor e_1'} P_2^{e_2}) \\ 
         &= {\rm Tr}(P_2^{e_2}(P_1^{e_1} + P_1^{e_1'}) \rho (P_1^{e_1} + P_1^{e_1'})P_2^{e_2}) \\ 
        & \neq {\rm Prob}(e_1, e_2) + {\rm Prob}(e_1', e_2).
    \end{split}
\end{equation}
A classic example is provided by the double slit experiment. $e_1$ and $e_1'$ can respectively be taken to represent the event that the particle went through the left slit and the event that it went through the right one, while $e_2$ represents the detection on the screen.

When a probability distribution ${\rm Tr}(P_n^{e_n} \ldots P_1^{e_1} \rho P_1^{e_1}\ldots P_n^{e_n})$ \textit{does} respect the classical sum rule, the associated history set $\{(e_1, \ldots, e_n)\}$ is called \textit{consistent} by \cite{griffiths1984consistent}. Because \cref{eq:prob} is manifestly linear in each projector $\tilde \pi^{\uparrow_{e_k}}_{\bA_k}$ and $\tilde \pi^{\downarrow_{f_k}}_{\bA_k}$
(unlike \cref{eq:consistent_ex_prob}, in which $P_1^{e_1}$ appears twice), it evidently defines a consistent history set, with the two projective-valued measurements $\{\tilde \pi_{\mathfrak{B}}^{\uparrow e_1 \ldots e_n}\}_{e_1\ldots e_n}$ and $\{\tilde \pi_{\mathfrak{B}}^{\downarrow f_1 \ldots f_n}\}_{f_1\ldots f_n}$ and the initial state $\rho = I/d$. (In fact, it is consistent not just in the sense of \cite{griffiths1984consistent}, but in the logically stronger sense of \cite{gell1990complexity}, which requires that ${\rm Tr}(P_n^{e_n} \ldots P_1^{e_1} \rho P_1^{e_1'}\ldots P_n^{e_n'})=0$ whenever there is some $k$ such that $e_k \neq e_k'$.)
\Cref{thm:prob} thus shows that decoherence and dual decoherence privilege a unique consistent history set for any subset $\mathfrak{B}$ of the systems in a unitary circuit.
\end{remarkbox}

\subsection{A single measurement: probabilities and interpretation of the emergent events} \label{subsec:single_prob}

We now return to our previous example of a single measurement.
We have reproduced the circuit in \cref{fig:circuit-single-mmt}, where we have also added, in blue, the decohered and dually decohered algebras for each system $\bA$ that we derived in \cref{eq:single-mmt-all-the-algebras}, along with the classical events $e_\bA$, $f_\bA$ that correspond to the subspaces in the preferred subspace decompositions picked out by these algebras.
Finally, we have decorated the figure with interpretations of each of the nontrivial emergent events.
The goal of this subsection is to justify those interpretations by analysing the probability distribution over the derived history set whose existence is guaranteed by \cref{thm:prob}.

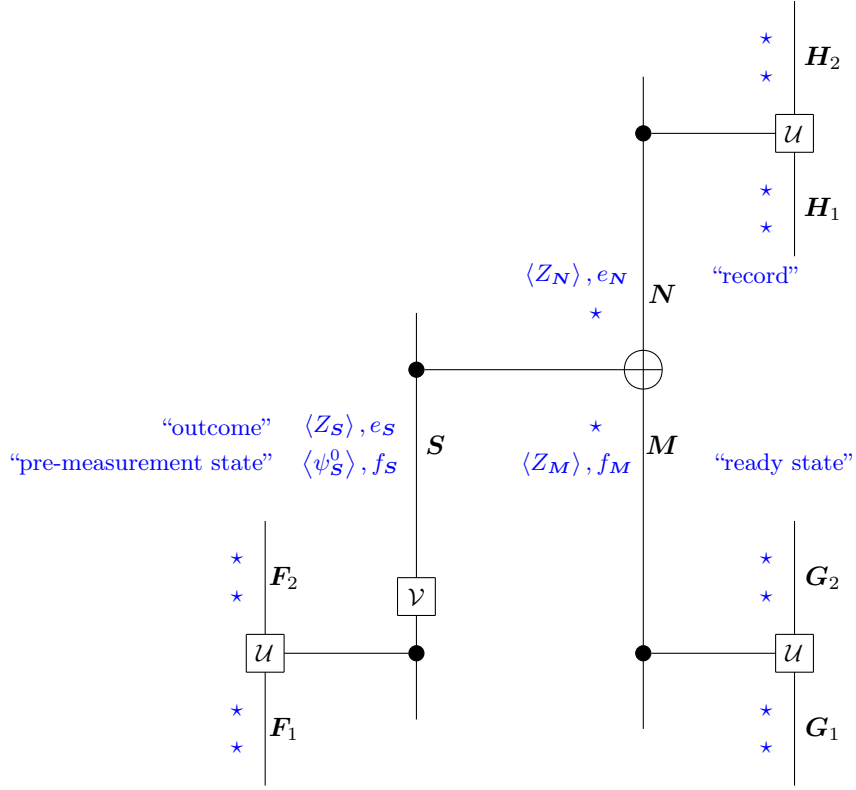
\begin{figure}[t]
    \centering
    \begin{tikzpicture}
	\begin{pgfonlayer}{nodelayer}
		\node [style=none] (61) at (3, -7.75) {};
		\node [style=none] (63) at (3, 9.5) {};
		\node [style=none] (64) at (-3, -3.75) {};
		\node [style=none] (65) at (-2.5, -0.25) {$\bS$};
		\node [style=none] (66) at (3.5, 3.75) {$\bN$};
		\node [style=none] (67) at (-3, 3.25) {};
		\node [style=none] (68) at (-7, -6.25) {};
		\node [style=none] (69) at (-7, -9.25) {};
		\node [style=none] (70) at (3.5, -0.25) {$\bM$};
		\node [style=black dot] (71) at (-3, 1.75) {};
		\node [style=none] (72) at (3.5, 1.75) {};
		\node [style=none] (73) at (2.5, 1.75) {};
		\node [style=none] (74) at (3.5, 1.75) {};
		\node [style=none] (75) at (2.5, 1.75) {};
		\node [style=black dot] (76) at (-3, -5.75) {};
		\node [style=none] (77) at (-6.5, -5.75) {};
		\node [style=none] (78) at (-6.5, -5.25) {};
		\node [style=none] (79) at (-6.5, -6.25) {};
		\node [style=none] (80) at (-7.5, -6.25) {};
		\node [style=none] (81) at (-7.5, -5.25) {};
		\node [style=none] (82) at (-7, -5.25) {};
		\node [style=none] (83) at (-7, -2.25) {};
		\node [style=none] (84) at (-7, -5.75) {\small $\cu$};
		\node [style=none] (85) at (-3.5, -3.75) {};
		\node [style=none] (86) at (-2.5, -3.75) {};
		\node [style=none] (87) at (-2.5, -4.75) {};
		\node [style=none] (88) at (-3.5, -4.75) {};
		\node [style=none] (89) at (-3, -4.25) {\small $\cv$};
		\node [style=none] (90) at (-3, -4.75) {};
		\node [style=none] (91) at (-3, -7.5) {};
		\node [style=none] (92) at (-4.75, 0.25) {\small $\color{blue} \gen{Z_\bS}, e_\bS$};
		\node [style=none] (93) at (-4.75, -0.75) {\small $\color{blue} \gen{\psi_\bS^0}, f_\bS$};
		\node [style=none] (94) at (1.75, 0.25) {\small $\color{blue} \star$};
		\node [style=none] (95) at (1.25, -0.75) {\small $\color{blue} \gen{Z_\bM}, f_\bM$};
		\node [style=none] (96) at (7, -6.25) {};
		\node [style=none] (97) at (7, -9.25) {};
		\node [style=black dot] (99) at (3, -5.75) {};
		\node [style=none] (100) at (6.5, -5.75) {};
		\node [style=none] (101) at (6.5, -5.25) {};
		\node [style=none] (102) at (6.5, -6.25) {};
		\node [style=none] (103) at (7.5, -6.25) {};
		\node [style=none] (104) at (7.5, -5.25) {};
		\node [style=none] (105) at (7, -5.25) {};
		\node [style=none] (106) at (7, -2.25) {};
		\node [style=none] (107) at (7, -5.75) {\small $\cu$};
		\node [style=none] (108) at (6.25, -7.25) {\small $\color{blue} \star$};
		\node [style=none] (109) at (6.25, -8.25) {\small $\color{blue} \star$};
		\node [style=none] (110) at (-7.75, -7.25) {\small $\color{blue} \star$};
		\node [style=none] (111) at (-7.75, -8.25) {\small $\color{blue} \star$};
		\node [style=none] (112) at (7, 7.5) {};
		\node [style=none] (113) at (7, 4.75) {};
		\node [style=black dot] (115) at (3, 8) {};
		\node [style=none] (116) at (6.5, 8) {};
		\node [style=none] (117) at (6.5, 8.5) {};
		\node [style=none] (118) at (6.5, 7.5) {};
		\node [style=none] (119) at (7.5, 7.5) {};
		\node [style=none] (120) at (7.5, 8.5) {};
		\node [style=none] (121) at (7, 8.5) {};
		\node [style=none] (122) at (7, 11.5) {};
		\node [style=none] (123) at (7, 8) {\small $\cu$};
		\node [style=none] (124) at (6.25, 10.5) {\small $\color{blue} \star$};
		\node [style=none] (125) at (6.25, 9.5) {\small $\color{blue} \star$};
		\node [style=none] (126) at (1.25, 4.25) {\small $\color{blue} \gen{Z_\bN}, e_\bN$};
		\node [style=none] (127) at (1.75, 3.25) {\small $\color{blue} \star$};
		\node [style=none] (129) at (6, 4.25) {\color{blue} \small ``record''};
		\node [style=none] (130) at (6.75, -0.75) {\color{blue} \small ``ready state''};
		\node [style=none] (131) at (-10.25, -0.75) {\color{blue} \small ``pre-measurement state''};
		\node [style=none] (132) at (-8.25, 0.25) {\color{blue} \small ``outcome''};
		\node [style=none] (133) at (3, 4.25) {};
		\node [style=none] (134) at (3, -0.75) {};
		\node [style=none] (135) at (-6.5, -7.75) {$\bF_1$};
		\node [style=none] (136) at (-6.5, -3.75) {$\bF_2$};
		\node [style=none] (137) at (-7.75, -3.25) {\small $\color{blue} \star$};
		\node [style=none] (138) at (-7.75, -4.25) {\small $\color{blue} \star$};
		\node [style=none] (139) at (7.75, -7.75) {$\bG_1$};
		\node [style=none] (140) at (7.75, -3.75) {$\bG_2$};
		\node [style=none] (141) at (7.75, 10) {$\bH_2$};
		\node [style=none] (144) at (7.75, 6) {$\bH_1$};
		\node [style=none] (145) at (6.25, -3.25) {\small $\color{blue} \star$};
		\node [style=none] (146) at (6.25, -4.25) {\small $\color{blue} \star$};
		\node [style=none] (147) at (6.25, 6.5) {\small $\color{blue} \star$};
		\node [style=none] (148) at (6.25, 5.5) {\small $\color{blue} \star$};
	\end{pgfonlayer}
	\begin{pgfonlayer}{edgelayer}
		\draw (71) to (72.center);
		\draw [bend left=90, looseness=1.75] (73.center) to (72.center);
		\draw [bend right=90, looseness=1.75] (75.center) to (74.center);
		\draw (61.center) to (63.center);
		\draw (76) to (77.center);
		\draw (64.center) to (67.center);
		\draw (68.center) to (69.center);
		\draw (81.center) to (80.center);
		\draw (78.center) to (79.center);
		\draw (79.center) to (80.center);
		\draw (81.center) to (78.center);
		\draw (83.center) to (82.center);
		\draw (88.center) to (85.center);
		\draw (85.center) to (86.center);
		\draw (86.center) to (87.center);
		\draw (87.center) to (88.center);
		\draw (91.center) to (90.center);
		\draw (99) to (100.center);
		\draw (96.center) to (97.center);
		\draw (104.center) to (103.center);
		\draw (101.center) to (102.center);
		\draw (102.center) to (103.center);
		\draw (104.center) to (101.center);
		\draw (106.center) to (105.center);
		\draw (115) to (116.center);
		\draw (112.center) to (113.center);
		\draw (120.center) to (119.center);
		\draw (117.center) to (118.center);
		\draw (118.center) to (119.center);
		\draw (120.center) to (117.center);
		\draw (122.center) to (121.center);
	\end{pgfonlayer}
\end{tikzpicture}
    \caption{Model of a single measurement. Systems of interest are labelled in black. Decohered and dual decohered algebras, associated emergent events, and their interpretations are labelled in blue. Each decohered algebra is depicted above the corresponding dual decohered algebra because the former arises from the interaction between the system and its future while the latter arises from its interaction with its past. As before, $\star$ denotes the trivial algebra $\mathbb{C}I$. Trivial decohered algebras are associated with trivial events (corresponding to the identity operator), which are not explicitly labelled.}
    \label{fig:circuit-single-mmt}
\end{figure}

The probability distribution is calculated from the projectors corresponding to each event.
In our example, each history
\begin{equation}
    (f_\bS, e_\bS, f_\bM, e_\bN)
\end{equation}
consists of only four nontrivial events, where the $e$-events emerge from decoherence and the $f$-events emerge from dual decoherence.
$e_\bS\in \{0,1\}$ denotes the event corresponding to the $e_\bS$th subspace picked out by the decohered algebra $\gen{Z_\bS}$, which is the subspace spanned by the computational basis state $\ket{e_\bS}$; equivalently, this event corresponds to the projector $\pi_\bS^{\uparrow e_\bS} = \ketbra{e_\bS}$.
Similarly for $f_\bM$ and $e_\bN$.
The event $f_\bS$ corresponds to the projector $\pi_\bS^{\downarrow f_\bS} = \ketbra{\psi^{f_\bS}}_\bS$.
\cref{eq:prob} thus becomes
\begin{equation} \label{eq:measurement_prob} 
    \begin{split}
        {\rm Prob}_\mathfrak{B}(f_\bS, e_\bS, f_\bM, e_\bN) &= \frac{1}{2^5} {\rm Tr} ( \tilde \pi_\bS^{\downarrow f_\bS} \tilde \pi_\bM^{\downarrow f_\bM} \tilde \pi_\bS^{\uparrow e_\bS} \tilde \pi_\bN^{\uparrow e_\bN}) \\
        &= \frac{1}{4} \left|\braket{e_\bS}{\psi^{f_\bS}}\right|^2 \delta( e_\bN, e_\bS \oplus f_\bM).
    \end{split}
\end{equation}
Here $\delta(x, y) =1$ if $x=y$ and $\delta(x, y) =0$ if $x \neq y$, and $\oplus$ denotes addition modulo 2. 

(To verify this equation, it is helpful to note that ${\rm Tr} ( \tilde \pi_\bS^{\downarrow f_\bS} \tilde \pi_\bM^{\downarrow f_\bM} \tilde \pi_\bS^{\uparrow e_\bS} \tilde \pi_\bN^{\uparrow e_\bN})$ is unchanged if one applies the same unitary channel to all four traced-over projectors. This means that to obtain the correct distribution, one is not forced to represent the projectors on the input space of the circuit; instead, one can choose whichever time slice is most convenient for the calculation. Here, the time slice immediately below the CNOT proves particularly convenient, since three of the four events are already associated with it.)

To see how these statistics reproduce the predictions of standard quantum theory, we must now discuss how \textit{prior beliefs} enter the formalism. In standard quantum theory, prior beliefs are often encoded in an initial state which is then fed into the dynamics. But that is not how prior beliefs are encoded on our approach. Indeed, the role of the circuit here is not to act on an initial state; at no point in our analysis will we feed a state into it. Instead, prior beliefs take the form of assumptions about events that emerge, and most typically about events that emerge early on in the circuit.

Let us suppose for instance that we have a prior belief that
    \begin{equation} \label{eq:prior}
        f_\bS = 0.
    \end{equation}
    In light of this prior, we can update our predictions for the event $e_\bS$ using standard probability theory:
    \begin{equation} \label{eq:state}
        \begin{split}
                {\rm Prob}_\mathfrak{B}(e_\bS \mid
            f_\bS = 0) &=  \frac{\sum_{f_\bM,e_\bN}{\rm Prob}_\mathfrak{B}(f_\bS=0, e_\bS, f_\bM, e_\bN)}{\sum_{e_\bS,f_\bM,e_\bN}{\rm Prob}_\mathfrak{B}(f_\bS=0, e_\bS, f_\bM, e_\bN)} \\
            &= \left|\braket{e_\bS}{\psi^0}\right|^2.
        \end{split}
    \end{equation}
    In standard quantum-theoretical terminology, this is the Born probability of obtaining the outcome associated with $\ket{e_{\bS}}$ when performing a measurement of the system prepared in the state $\ket{\psi^0}$.
    Similarly, conditioning on $f_\bS=1$ yields the Born probability $\left|\braket{e_\bS}{\psi^1}\right|^2$.
    The emergent event $f_\bS$ thus plays the same role that the quantum state does in the standard view: prior knowledge of $f_\bS$ allows one to make predictions about future events like $e_\bS$.
    This motivates the interpretation, stipulated in \cref{fig:circuit-single-mmt}, of $f_\bS$ as a \enquote{pre-measurement state} and of $e_\bS$ as an \enquote{outcome}.

    It remains to interpret the emergent events $f_\bM$ and $e_\bN$.
    To do this, note that if we make the prior assumption that
    \begin{equation}\label{eq:prior-fm}
        f_\bM = 0,
    \end{equation}
    then the variable $e_\bN$ is certain to match $e_\bS$:
    \begin{equation} \label{eq:record}
            {\rm Prob}_\mathfrak{B}(e_\bN \mid e_\bS, \;
        f_\bM = 0) = \delta(e_\bN, e_\bS).
    \end{equation}
    Thus $f_\bM=0$ can be interpreted as the \enquote{ready state} of the measurement device, which ensures that after the CNOT measurement interaction the value of $e_\bN$ \enquote{records} the measurement outcome $e_\bS$. (This is related to how on the more usual state-first picture from standard quantum theory, if the measurement device is prepared in the state corresponding to the vector $\ket0_\bM$ and the system is prepared in $\ket{k}_\bS$, then the measurement device will end up in the state $\ket{k}_\bN$.)

    We thus see that the events $f_\bS$ and $f_\bM$ are naturally interpreted as states prepared by those fragments.
    We therefore refer to $f$-events, i.e.\ events emerging from dual decoherence, as \emph{state-like events}, or sometimes simply as \emph{states}.
    In the present account, the traditional quantum state vector (e.g.\ $\ket{\psi^{0}}_\bS\otimes\ket{0}_\bM$) serves as a summary of precisely the information about past events (in this case $f_\bS=0$ and $f_\bM=0$) that is relevant to computing probabilities for future events.
    It is in this sense that the quantum state emerges from unitary dynamics.

    Before moving on to a second example, we finish this one off with some discussion. We begin by pointing out the role played in this example by the two core concepts of the causal approach to quantum Darwinism: the proliferation of information about observables and the redundant implementation of generators. The record $e_\bN$ emerges because information about the observables in $\gen{Z_\bN}$ is transmitted into a system within $\mathfrak{B}$, namely the environment fragment $\bH_2$, while also remaining available for further proliferation into systems outside $\mathfrak{B}$ that do not directly interact with the fragment $\bH_2$. Similarly, the outcome $e_\bS$ emerges because information about the observables in $\gen{Z_\bS}$ proliferates into some systems in $\mathfrak{B}$, namely $\bN$ and $\bH_2$, while remaining available for further proliferation outside $\mathfrak{B}$ into systems that do not directly interact with them.

On the other hand, the state $f_\bS$ emerges because there is a system within $\mathfrak{B}$, namely the environment fragment $\bF_1$, that can implement the generators within the associated algebras without blocking past systems outside $\mathfrak{B}$ that do not directly interact with $\bF_1$ from doing the same by intercepting the necessary causal influences. Similarly for the state $f_\bM$ and the environment fragment $\bG_2$.

This example illustrates the equal importance of proliferation and redundant implementation in explaining emergent classicality. It is vital that in this model, the two $e$-events emerging from the potential for proliferation were correlated (\cref{eq:record}), since that is what justifies regarding one as a record of the other.
But they were only correlated upon conditioning on the state-like event $f_\bM=0$ (or, alternatively, $f_\bM=1$); in the absence of such conditioning, $e_\bS$ and $e_\bN$ are uncorrelated.
It is thus the \textit{interplay} between outcome-like events emerging from the potential for proliferation and state-like events emerging from the potential for redundant implementation that allows the privileged history set to describe the measurement scenario.

    
Relatedly, we note that in this example, the emergence of states is highly analogous to the emergence of outcomes. As mentioned, the state vector $\ket{\psi^0}_\bS$ turns out to be a useful tool for calculating the probability of any $e_\bS$ given that $f_\bS = 0$. Similarly, the bra $\bra{0}_\bS$ is useful for calculating the conditional probability of any $f_\bS$ given that $e_\bS = 0$. The symmetry between outcomes and states is only broken by the fact that one more commonly wants to make predictions for unknown outcomes given prior beliefs about states rather than vice versa, and thus more often conditions on a state than on an outcome.\footnote{This builds on a theme from \cref{sec:symptoms}: the process of decoherence is thoroughly time symmetric, and any time asymmetry originates in prior beliefs that already pick out a preferred direction of time.}
One might even argue that all that it \textit{means} to assert that one event is an outcome and another is a state is to express a particular attitude regarding which events one is likely to condition on.

We are now able to see what makes our approach to emergent classicality not merely focused on the dynamics, but \textit{dynamics-first}. It is often assumed that (our knowledge of) the fundamental quantum reality is represented by a quantum state, and that classicality emerges when that state gets decohered. Here, by contrast, the quantum state is a part of what emerges from (dual) decoherence. States are on an equal footing with outcomes: each is a symptom of the process of decoherence. The most fundamental reality is represented not by any state at all, but by a unitary circuit. And the role of the circuit is not to act on some initial state, but to give birth to state-like and outcome-like events via dual decoherence and decoherence. \cref{sec:crqt} will discuss the implications of this dynamics-first perspective for the interpretation of quantum theory.

We acknowledge that the present model of a measurement is highly simplistic in the sense that it lacks records of the outcome stored in the environment $(\bH_1, \bH_2)$, and it lacks any records whatsoever of the state-like events $f_\bS$ and $f_\bM$. This simplicity is for pedagogical reasons only; as discussed further in \cref{sec:when}, such records appear in more detailed models featuring additional decohering interactions.

\subsection{Two incompatible measurements}

The example above included a single $Z$-measurement. Our next example will feature a $Z$-measurement followed by a measurement of a different, noncommuting observable $O$. As well as illustrating how our formalism can be used to model more complicated scenarios, this example will show how it recovers the projection rule for updating quantum states. 

As in the previous example, we set up our model by drawing a unitary circuit, namely the one in \cref{fig:circuit-two-mmts}, and declaring a subset $\mathfrak{B}$ to be our systems of interest. We denote by $O_\bT$ the observable mapped by the unitary channel $\cw = W(\cdot)W^\dagger$ to the $Z$ observable on $\cw$'s unlabelled output system, and we assume that $[ O_\bT, Z_\bT] \neq 0$.
We define $\mathfrak{B}$ to be precisely the set of systems labelled in black in the figure.

Decoherence and dual decoherence then select two commutative algebras $\bA_{{\rm dec}\mathfrak{B}}^\uparrow$ and $\bA_{{\rm dec}\mathfrak{B}}^\downarrow$ for each system of interest $\bA\in\mathfrak{B}$.
The resulting algebras, and (where they are nontrivial) the corresponding events $e_\bA$ and $f_\bA$, are labelled in blue in \cref{fig:circuit-two-mmts}; we skip their derivation, but the reader is again encouraged to verify them.
Finally, we have decorated the figure with interpretations of the (dual) decohered  algebras, which we justify in what follows.

\begin{figure}[t]
    \centering
    \scalebox{.9}{\input{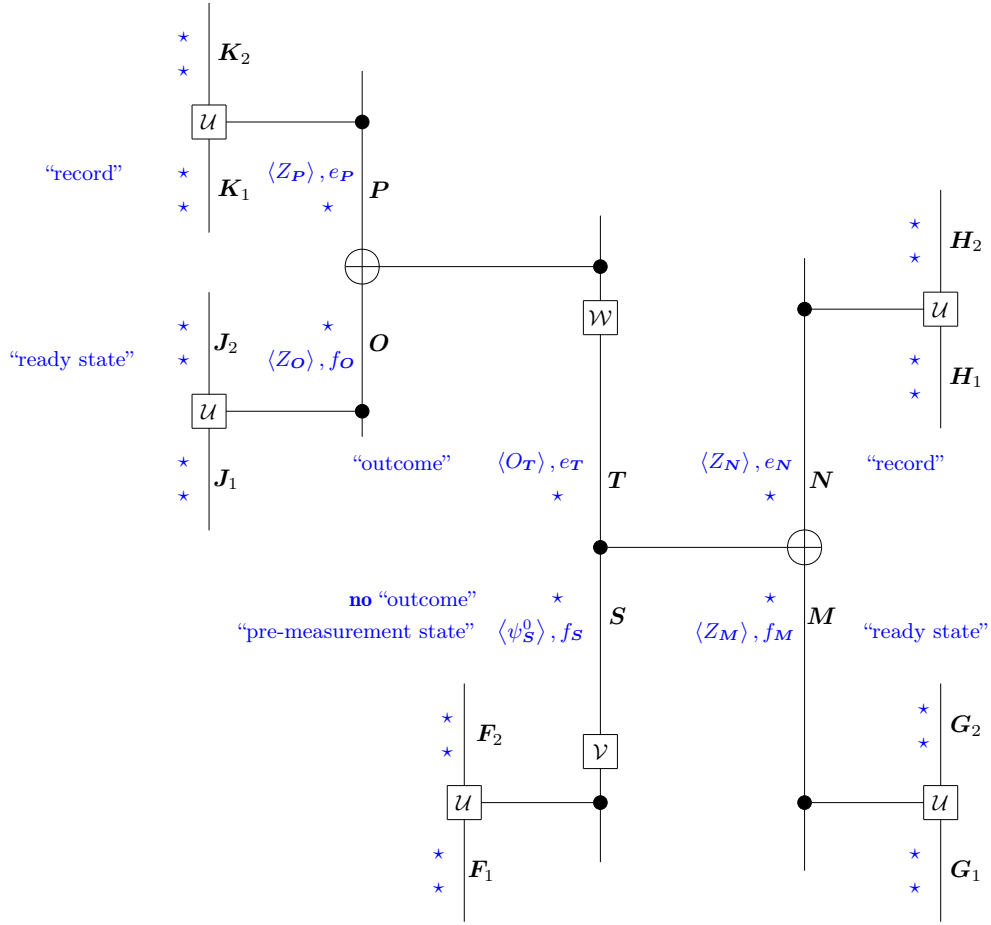}}
    \caption{Model of two incompatible measurements decorated as before with decohered and dual decohered algebras, associated events, and their interpretation. $O_\bT$ denotes the operator that is transformed to the $Z$-operator of the (unlabelled) output system of $\cw = W(\cdot) W^\dagger$. We assume that $[O_\bT, Z_\bT] \neq 0$.}
    \label{fig:circuit-two-mmts}
\end{figure}

The transformations in the lower right part of this circuit precisely match those from the previous example. As a consequence, almost all the events that emerged in the previous example also emerge here. There is one exception: no outcome event $e_\bS$ emerges at $\bS$. 

For an explanation of why not, recall that in the previous example, the observable $Z_\bS$ was decohered because it both proliferated into $\tilde \bN \lor \tilde \bH_2$ and remained available for further proliferation through future interactions with the measured system. But in the present example, $\tilde Z_\bS \not\in \tilde \bP'$. Hence $Z_\bS \not\in \bS_{{\rm pacc}\mathfrak{B}}^\uparrow$ and $Z_\bS \not\in \bS_{{\rm dec}\mathfrak{B}}^\uparrow$; the observable is not decohered. This is no coincidence: since the second measuring device $\bP$ accesses an observable $\tilde O_\bT$ that does not commute with $\tilde Z_\bS$, it inevitably (by \cref{thm:algebras}) fails to leave $Z_\bS$ potentially accessible. An event at $\bS$ corresponding to the outcome of the first measurement thus fails to emerge not only with respect to $\mathfrak{B}$, but any set of systems of interest that includes $\bP$.

    But even though no outcome event emerges at $\bS$, there is still an emergent event $e_\bN$ at $\bN$, corresponding to the same algebra $\gen{Z_\bN}$ as it did in the previous example.
    In the previous example, we showed that this event $e_\bN$ can be interpreted as a record of the outcome $e_\bS$ (given appropriate priors). 
    Here, the variable $e_\bN$ still behaves \emph{as if} it were a record of the outcome of a $Z$-measurement of $\bS$, even though strictly speaking there is no outcome event on $\bS$ itself that it is recording.
    Indeed, if we assume the prior
    \begin{equation}
        f_\bS = 0 \quadand f_\bM = 0,
    \end{equation}
    as in the previous subsection, then our probability rule from \cref{thm:prob} gives the conditional probability
    \begin{equation}
        {\rm Prob}_\mathfrak{B}(e_\bN \mid f_\bS=0, \; f_\bM = 0)  = \left|\braket{e_\bN}{\psi^{0}}\right|^2.
    \end{equation}
    Thus, as in the previous example, $f_\bS=0$ is interpreted as a \enquote{pre-measurement state}, $f_\bM=0$ as the \enquote{ready state} of the measuring device, and $e_\bN$ as the \enquote{record} of the outcome of a $Z$-measurement.
    There is in this case no event at $\bS$ itself that corresponds to this outcome, but assuming the outcome of a measurement is only ever perceived indirectly, via a record, there is no observable difference between the $Z$-measurement in this example and that in the previous one.

What about the remaining events $e_\bT$, $f_\bO$, and $e_\bP$?
Given a particular record $e_\bN$ of the first measurement, the probabilities for $e_\bT$ are given by
\begin{equation}
    \begin{split}
                {\rm Prob}_\mathfrak{B}(e_\bT \mid
        f_\bS = 0,\; f_\bM = 0,\; e_\bN)
        = \left|\bra{e_\bT} W \ket{\phi(e_\bN)}\right|^2, \\
        \quad \text{where} \ \ket{\phi(e_\bN)} := \frac{\ketbra{e_\bN}{e_\bN} \ket{\psi^0}}{\sqrt{\left|\braket{e_\bN}{\psi^0}\right|^2}} .
    \end{split}
\end{equation}
In other words, the correct probabilities for $e_\bT$ are calculated by applying the projection postulate to update the quantum state vector after the first measurement, applying the unitary $W$ to transform the state, and then finally using the Born rule.
From standard quantum theory we know this means that $e_\bT$ can be interpreted as the \enquote{outcome} of a second measurement performed on the same system.

Finally, similarly to \cref{eq:prior-fm,eq:record}, if we condition on the prior belief that
\begin{equation}
    f_\bO = 0,
\end{equation}
then the probabilities for $e_\bP$ are
\begin{equation}
     {\rm Prob}_\mathfrak{B}(e_\bP \mid
    e_\bT, f_\bO = 0) = \delta(e_\bP, e_\bT),
\end{equation}
confirming that $f_\bO=0$ serves as a \enquote{ready state} of the second measurement device and $e_\bP$ as a \enquote{record} of its outcome.

We finish this example with some discussion.
In the first example, the states arose solely from conditioning on events emerging from dual decoherence. Here, the state $\ket{\phi(e_\bN)}$ has arisen from conditioning on not only events emerging from dual decoherence (when assuming the prior) but also on an event $e_\bN$ emerging from decoherence. This is the caveat to the statement that \enquote{outcomes emerge from decoherence and states emerge from dual decoherence}. Although a \textit{state-like event} is one that emerges solely from dual decoherence, a state vector or density matrix is often a useful tool for calculating probabilities conditional on a collection of events that emerge from both decoherence and dual decoherence. 

Similarly, while an outcome-like event emerges from decoherence alone, an \textit{effect} is often a useful tool for calculating probabilities associated with a mixed collection of events. For example, in the present example the effect $\bra{e_\bN}_\bS$ on $\ch_\bS$ serves as a tool for calculating the probability of the record $e_\bN$ given the prior belief that the \enquote{ready} state-like event $f_\bM=0$ takes place. 

If we were to add a third noncommuting measurement, then the story would be similar: the outcome of the second measurement would vanish, but the first two records would remain. The prediction for the outcome of the third measurement would be correctly obtained by applying the projection postulate twice.

Before moving on to the next example, we comment briefly on how more general states and outcomes emerge via dual decoherence and decoherence. Any pure state of a system emerges via dual decoherence from a coherent control interaction where the pure state appears in the control basis. Since an arbitrary mixed state $\rho$ is the partial trace of some pure state, it is always possible to construct a circuit that \enquote{prepares} $\rho$---more precisely, such that conditioning on appropriate events that emerge from the circuit gives probabilities that coincide with those predicted by the density operator $\rho$. Similarly, the outcome of any basis measurement emerges via decoherence from an interaction coherently controlled on the basis to be measured. Since an arbitrary POVM (positive operator-valued measure) always arises from a basis measurement on a larger system, this provides a way of making an event emerge that, under the appropriate priors, corresponds to the outcome of the POVM.

\subsection{Wigner's friend}

Our final example is based on a version of the celebrated Wigner's Friend thought experiment \cite{wigner1995remarks}. It will give us a sense of the sort of relationalism that is arguably suggested by decoherence, about which we will have more to say in \cref{subsec:relationalism}.

In this thought experiment, Wigner stands outside a perfectly isolated laboratory, in which his friend performs a quantum measurement on a particle. Wigner, who presumably should be able to regard the perfectly isolated laboratory (including his friend) as evolving unitarily during the friend's measurement, subsequently undoes that unitary evolution by applying the inverse. He then performs his own measurement of an observable of the particle that does not commute with the observable that was measured by the friend. The scenario is an exotic one because not only the friend's outcome, but also their \textit{experience of seeing} the outcome corresponds to an operator that is incompatible with the operator measured by Wigner.

\begin{figure}[b!]
    \centering
    \scalebox{.9}{\input{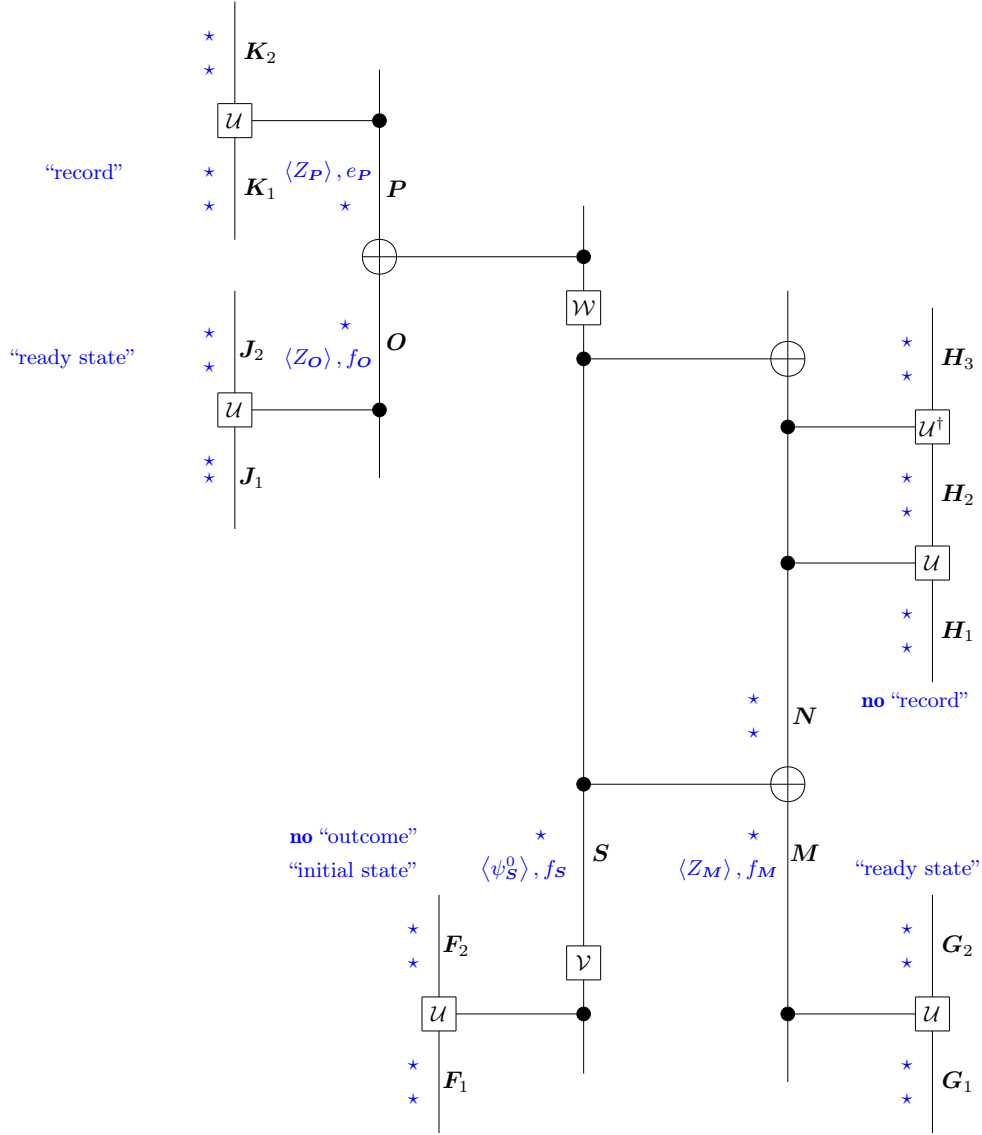}}
    \caption{Model of the Wigner's friend scenario decorated as before with decohered and dual decohered algebras, associated events, and their interpretation.}
    \label{fig:circuit-wigner}
\end{figure}

We model this scenario with the circuit in \cref{fig:circuit-wigner}, where $\bO$ and $\bP$ represent Wigner, $\bM$ and $\bN$ represent  his friend, $\bF_i$, $\bG_j$, and $\bH_k$ represent fragments of the environment within the lab, and $\bJ_m$ and $\bK_n$ represent fragments of the environment outside the lab. The circuit here is similar to the one in \cref{fig:circuit-two-mmts}, with one exception: this time, before the second measurement is implemented, two of the interactions are \enquote{undone} using their inverses. As in the previous examples, the set of systems of interest $\mathfrak{B}$ consists of precisely those labelled in black, and (dual) decohered algebras, emergent events, and their interpretations are labelled in blue.

Recall from the second example that although there was no outcome event on $\bS$ emerging from decoherence, there was nevertheless a record. But in this example, neither an outcome of the first measurement nor a record emerges from decoherence. Wigner's measurement prevents both degrees of freedom from achieving redundant proliferation: formally, $\tilde Z_\bS, \tilde Z_\bN \notin \tilde \bP'$, and thus $Z_\bS, Z_\bN \notin \bS_{{\rm dec}\mathfrak{B}}^\uparrow$.

Given our interpretation of $\bM$ and $\bN$ as the friend, $Z_\bN$ would be the natural candidate for an operator corresponding to the friend's experience of seeing a particular outcome. Hence this experience is not something that emerges from decoherence relative to $\mathfrak{B}$. However, it does emerge relative to other sets of systems of interest. For example, if we removed from $\mathfrak{B}$ every system that is higher up in the diagram than $\bH_2$ (regardless of whether any path of wires connects it to $\bH_2$), then we would find precisely the same decoherent algebras and statistics corresponding to the measurement outcome $e_\bS$ and the record $e_\bN$ that we did in the example of a single measurement.

This is a striking illustration of a relational aspect of decoherence: there is no \textit{absolute} fact about whether a degree of freedom decoheres, but only a fact about whether it decoheres \textit{relative} to a given set $\mathfrak{B}$ of systems of interest. In this example, there is no absolute fact of the matter about whether or not an event corresponding to the friend's experience emerges.

\subsection{Relationalism} \label{subsec:relationalism}

This raises an important question. Does the emergent classical world exist in an absolute sense, or do different classical worlds exist relative to different sets of systems? Does the need to begin the analysis of decoherence by choosing a subset of systems of interest reflect a genuinely relational aspect of emergent classicality, or merely a limitation of the current formalism?

To address these questions, we begin by arguing that if there is some way of eliminating the relational aspect of decoherence, it isn't \textit{straightforward}. Indeed, one obvious way of attempting to get a nonrelational account of emergent classicality is to simply let $\mathfrak{B}$ include every system. One might hope that by doing so one could obtain a single consistent history set through which all of our observations could be described, embedding all other relevant history sets. But unfortunately, this does not work: for any unitary circuit, the privileged history set for the set of all systems is trivial.
Indeed, by including every system, one ensures that for every system $\bA$ except those at the very top of the circuit, all $\bA$-observables are accessible to systems of interest above it: $\bA^\uparrow_{{\rm acc}\mathfrak{B}} = \bA$.
For a system $\bA$ in the top layer, $\bA^\uparrow_{{\rm acc}\mathfrak{B}} = \star$. Hence $\bA^\uparrow_{{\rm dec}\mathfrak{B}} = \star$ for all $\bA$. By a similar argument, $\bA^\downarrow_{{\rm dec}\mathfrak{B}}= \star $ for all $\bA$, and there are no nontrivial emergent events.

Another strategy for achieving a nonrelational account of emergent classicality is to look for some other way of combining all privileged history sets into a consistent whole, so that each history set can be regarded as a different partial description of the same underlying reality. However, while it is possible that there is some creative way of doing this, there does not appear to be any obvious way. One reason why not is that decohered algebras associated to a single system by two different sets of systems of interest $\mathfrak{B}_1$ and $\mathfrak{B}_2$ may fail to commute with each other, i.e.\ $\tilde \bA_{{\rm dec} \mathfrak{B}_1}^\uparrow \nsubseteq (\tilde \bA_{{\rm dec} \mathfrak{B}_2}^\uparrow)'$, and thus do not generally combine to form more fine-grained commutative algebras. 
(An example is provided by the circuit in \cref{fig:circuit-two-mmts} and the sets $ \mathfrak{B}_1 = \{\bS, \bN\}$ and $ \mathfrak{B}_2 = \{\bS, \bP\}$, where it turns out that $\bS_{{\rm dec} \mathfrak{B}_1}^\uparrow =\langle Z_\bS\rangle $ and $\bS_{{\rm dec} \mathfrak{B}_2}^\uparrow = \langle O_\bS \rangle $.)

None of the above precludes the possibility of some less straightforward path to a nonrelational account of emergent classicality. However, it does suggest that the relationalism in our formalism is not an entirely superficial feature of our formalism, but instead reflects something important about the phenomenon of decoherence itself. To reinforce this point, we recall from \cref{sec:channel} that decoherence requires a balancing of causal influence, and argue that the relational character of decoherence directly follows from this need for a causal balance. Indeed, we saw above that if $\mathfrak{B}$ includes every system, then the derived history set is trivial---the reason is that when every system is included, there is inevitably too much causal influence for decoherence to take place. Conversely, if we let $\mathfrak{B}$ include just \textit{one} system $\bA$, then $\bA^\uparrow_{{\rm acc}\mathfrak{B}} = \bA^\downarrow_{{\rm acc}\mathfrak{B}} =  \star$, and again the derived history set is trivial---this time because there is inevitably too little causal influence. In the context of circuits, the connection between decoherence and causal balance entails that nontrivial observables can only be decohered when some, but not all, systems in the circuit are included in $\mathfrak{B}$.

It is worth noting how this last observation resonates with a core claim of quantum Darwinism: that an observer perceives an object by interacting with more than one, but less than all, of the fragments of its environment, and this is the reason that the observer can only perceive the preferred classical degrees of freedom. The same idea is reflected in the plots one commonly encounters of mutual information between a system and a fraction of the environment against the size of the fraction (see e.g.\ Figure 2 of \cite{zurek2009quantum}), which reach the classical maximum for medium-sized fractions and exceed it for very large fractions.

In the absence of any definitive answer to the question of what sort of relationalism is involved in the emergence of classicality, we can always take the pragmatic approach of choosing whatever $\mathfrak{B}$ allows us to derive a history set that captures what we want to capture. Indeed, this is how we have chosen $\mathfrak{B}$ in our examples above. In each example, the derived history set can be seen as a post-hoc justification for the choice we have made. But questions linger about exactly what we are doing when we choose $\mathfrak{B}$, and about what the need to make this choice tells us about the emergence of classicality.

One possible (though for now, speculative) path towards answering the question of relationalism begins by acknowledging that the present approach takes events to emerge from decoherence and dual decoherence, i.e.\ whenever there is the \textit{potential} for proliferation or redundant implementation. A more stringent requirement would be that events emerge only when proliferation or redundant implementation actually occur. It is conceivable that more stringent conditions on the emergence of events could further limit the number of nontrivial privileged consistent history sets and/or improve the prospects for somehow unifying the different nontrivial sets that do emerge. But pursuing this possibility further is beyond our current scope.

\subsection{Summary: causal quantum Darwinism}\label{subsec:circuit-summary}

At the end of \cref{sec:dual}, we gave a brief heuristic account of the emergence of outcomes and states. Now, the derivation of history sets and the three examples studied in this section have made that account much more concrete. It is therefore worth restating it from our current vantage point.

In this picture, an outcome emerges due to the influence of a system on its surrounding environment, and more precisely due to the sort of influence that helps information about the associated observables to proliferate. A state emerges due to the influence of the surrounding environment on the system, and more precisely due to the sort of influence that helps the associated generators to become redundantly implementable. Importantly, neither the distinction between an outcome and a state, nor the underlying distinction between decoherence and its dual, is fundamental: each is related to the other by time reversal. But the symmetry is broken when we condition on a prior belief that picks out a preferred direction of time. Our account is dynamics-first in the sense that both outcomes and states emerge from the unitary dynamics alone, rather than from the way that the dynamics act on some initial state. Indeed there is no such initial state; the unitary dynamics are a \enquote{thing in themselves} rather than a rule for transforming states.

This point of view could be called \textit{causal quantum Darwinism}.  In this name, \enquote{Darwinism} points to the importance of the proliferation of information in explaining the emergence of classicality. The word \enquote{causal} has a double meaning: it signals the importance of the \emph{process} of decoherence, characterized independently of its symptoms, but it also emphasises that the experiences we are trying to explain are experiences of creatures who are not just observers, but agents who both influence and are influenced by the world around them; who both perceive outcomes and prepare states. By combining these terms, \enquote{causal quantum Darwinism} represents the view that the emergence of classicality consists not just of perceived outcomes emerging through the proliferation of information, but also of prepared states emerging through the redundant implementation of generators. 

The power of the causal quantum Darwinist approach is demonstrated by our derivation of privileged consistent history sets, in which the core causal quantum Darwinist idea of states emerging through dual decoherence plays a central role. While existing literature already contains some arguments that decoherence selects preferred consistent history sets in particular models (e.g.\ \cite{Zurek:1994zq}), what we have seen in this section is a rigorous derivation of a unique consistent history set in the highly general setting of an arbitrary unitary circuit and an arbitrary set of systems of interest. Our examples and the discussion around them show that the privileged history sets are capable of describing complex scenarios involving multiple incompatible measurements. Causal quantum Darwinism thus provides a unified formalism for decoherence in which many dynamically privileged degrees of freedom come together to form history sets plausibly rich enough to describe the emergent classical world.

\section{When does a measurement take place?} \label{sec:when}

One of the key motivations for the theory of decoherence was to explain how, in a unitary quantum world, there come to be measurements with apparently definite outcomes. Now that we have a precise account of how privileged consistent history sets emerge from decoherence, we can give a correspondingly precise account of the necessary conditions for measurement. By referring to the model of a single measurement from the previous section, displayed in \cref{fig:circuit-single-mmt}, this section will critically assess three commonly made claims:
\begin{enumerate}[label=(\arabic*)]
    \item a measurement can take place only if the measurement device interacts with its environment;
    \item a measurement can take place only if the measurement device begins in an appropriate initial state; and
    \item a measurement can take place only if the environment begins in an appropriate initial state.
\end{enumerate}

First of all, the causal approach suggests that (1) is true, though with a caveat. The production of an observable record of the outcome is an essential part of any measurement, and if one modifies the model in \cref{fig:circuit-single-mmt} by removing the interaction between the measurement device $\bN$ and the environment fragment $(\bH_1, \bH_2)$, then the record $e_\bN$ on the measurement device does not emerge. 
Note, however, that the emergence of the outcome $e_\bS$ itself does not depend on this interaction, since the interaction between the system and the measurement device suffices for its emergence.

The causal approach also suggests that (2) is true. The event $f_\bM$ is naturally interpreted as the preparation of the state $\ket{f_\bM}_\bM$, and only upon conditioning on this event does a correlation arise between the outcome $e_\bS$ and the record $e_\bN$.  Since such correlations are essential for measurement, it is indeed necessary that this state-like event occurs.

However, it is interesting to note that the mere \textit{existence} of nontrivial events $e_\bS$ and $e_\bN$ does not, on the causal approach, depend on $f_\bM$ (or $f_\bS$). Indeed, even if one prevented $f_\bM$ from emerging by removing the interaction with $\bG_1$, $\bG_2$, then $e_\bS$ and $e_\bN$ would still emerge, as they would still be privileged by the dynamics that would remain.

Finally, the causal approach suggests that (3) is false, though with a caveat. In our measurement model, there are no state-like events in the environment; nevertheless, an outcome occurs and a record is produced on the measurement device, and we take this to be sufficient for a measurement to have taken place. 

The caveat here is that our measurement model is very simplistic in that, although a record of the outcome is stored on the measurement device, there are no records of the outcome stored in the environment. As such, the model does not capture how information flows from the system to the environment and then ultimately to the observer. One can capture this process in our framework, but to do so one has to add in interactions between different fragments of the environment through which more records could emerge. To ensure that these records are correlated with the outcome, one also needs interactions out of which state-like events emerge so that one can then condition on those events. Hence a less simplistic model of a measurement \textit{would} require state-like events in the environment. 

Yet it is worth bearing in mind that no matter how long a chain of records in the environment one considers, the emergence of the final record in the chain will not depend on the state of the fragment that decoheres it, or indeed on whether that fragment has a state at all---just as in our model, the emergence of the record $e_\bN$ does not require $\bH_1$ to have any state. This sets the causal approach apart from some more traditional approaches to decoherence: on the view that events emerge due to the suppression of off-diagonal terms in a quantum state, whether or not an event emerges does depend on the initial state of the decohering system. But on the causal approach, both the state of the environment and the record are symptoms of the process of decoherence, and the latter symptom can manifest without the former.

\section{Decoherence and the interpretation of quantum theory} \label{sec:crqt}

The conceptual implications of decoherence depend on one's interpretation of quantum theory. Indeed, for an Everettian, decoherence is the branching of worlds, whereas for a Bohmian, it is the splitting of the quantum potential into one part that guides the corpuscles and other parts that do not. It also goes the other way: one's interpretation of quantum theory may be guided by one's understanding of decoherence. For many, the Everett interpretation did not appear plausible until it was shown that decoherence defines a preferred basis with respect to which the branching structure can be defined. 

Likewise, the conceptual implications of our own formalism for decoherence depend on one's interpretational perspective, but can also guide it. Indeed, there is an obvious way of taking the formalism \enquote{literally} so that it does define a precise, realist interpretation of quantum theory: \\

\textbf{Causal Decoherence Theory}
\begin{enumerate}
    \item Reality is represented by unitary dynamics together with the histories that emerge from them via decoherence and dual decoherence.
    \item Relative to every subset $\mathfrak{B}$ of systems, exactly one history takes place from the privileged consistent history set.
    \item The probability for a given history to take place is given by \cref{eq:prob}. 
\end{enumerate}

\Cref{app:bubble} shows that Causal Decoherence Theory (CDT) is equivalent to the interpretation of quantum theory originally introduced in~\cite{ormrod2024quantum}, in which events emerge from the \enquote{interference influences} between complete orthogonal families of projectors. Compared to the original, the present formulation of the theory has two main advantages. 
Technically, the formalism here is simpler: preferred subspace decompositions are derived directly from the circuit itself, rather than by first breaking the wires corresponding to systems of interest and then extracting them from the resulting channel (see \cref{app:bubble} for more details).
Conceptually, it sharpens the central claim of the interpretation: while \cite{ormrod2024quantum} posited that events emerge from causal structure, the present formulation makes clear that the relevant causal structures are those associated with decoherence and dual decoherence, i.e.\ those required for the proliferation of observables and the redundant implementation of generators.

Axiom (2) above entails that Causal Decoherence Theory is \textit{stochastic}: nothing predetermines which history from the preferred set will actually happen. This is why the theory requires a probability rule, i.e.\ axiom (3). 

Stochasticity might sound like a surprising feature of an interpretation based purely on unitary quantum theory and decoherence. It would indeed be surprising on a traditional state-first approach: if the world consisted only of a quantum state evolving unitarily, then the world would be deterministic, because the unitary dynamics act deterministically on the quantum state.
But in CDT, the role of the unitary circuit is not to act on a quantum state. Its role is to give rise to events through decoherence and dual decoherence, and it does so in a stochastic way. Therefore, from the dynamics-first perspective of CDT, the unitary dynamics are not deterministic.

The stochasticity of CDT distinguishes it from another interpretation of quantum theory based purely on unitary quantum theory and decoherence, namely the Everett interpretation as laid out in e.g.\ \cite{wallace2012emergent}. Because it holds that every possible outcome of a measurement actually happens on some branch, the Everett interpretation implies that the Born probabilities are not associated with ignorance about which outcome will arise. The plausibility of the Everett interpretation therefore rests on the plausibility of some nonstandard account of the role of probabilities, in which probabilities not only do not represent ignorance, but cannot be associated with ignorance. Such accounts have been provided in e.g.\ \cite{deutsch1999quantum, wallace2012emergent, greaves2010everett} and criticized in  e.g.\ \cite{albert2010probability, adlam2014problem}.

By contrast, in CDT the Born probabilities are associated with genuine ignorance of what is to come, and thus the plausibility of CDT is not tied to the plausibility of a nonstandard interpretation of probabilities. Hence those who would like a precise realist interpretation of quantum theory based only on unitary dynamics and decoherence but are sceptical that the Everett interpretation can provide a satisfactory account of the role of probabilities are likely to prefer CDT. 

On the other hand, for those who are open to a nonstandard account of probabilities, the causal formalism for decoherence also points to a possible alternative, dynamics-first version of the Everett interpretation. This Everett interpretation is closely related to CDT, but with a crucial difference.
Rather than positing that only one history from a preferred set happens as in axiom (2) above, it is held that \textit{all} happen. 
It thus requires a nonstandard account of probabilities to explain the meaning of the probability rule.\footnote{Since we have not given this account, this proposal is currently better described as an interpretational \textit{approach} than a worked-out interpretation. However, for ease of exposition we refer to it simply as an ``interpretation''.}

This dynamics-first Everett interpretation (DFEI) is Everettian in the sense of being a branching multiverse theory. However, the branches are not part of the wavefunction; instead, they are part of the \textit{unitary dynamics}. For example, consider a maximally decohering interaction $\cu: \bS \bF \rightarrow \bT \bG$ (whose form is shown in \cref{eq:uni_control}) and the set $\mathfrak{B} = \{\bS, \bG\}$ of systems of interest. Each unitary $V_{\bF \rightarrow \bG}^{(k)}$ corresponds to a branch of $\cu$ associated with the event $e_\bS = k$ emerging from decoherence. 

We note in passing that DFEI is closely related to Deutsch and Hayden's version of the Everett interpretation \cite{deutsch2000information}. Indeed, the decoherent algebras can be thought of as providing a generalization of the \enquote{Heisenberg picture relative states} described in \cite{kuypers2021}.

We will now sketch an argument for why an Everettian might prefer the dynamics-first version of the theory over a state-first version. Both versions of the Everett interpretation are committed to the unitary dynamics; the difference is simply that the state-first version is, in addition, committed to a quantum state on which the unitary dynamics act. But if DFEI is correct to argue that histories emerge from the unitary dynamics alone, then, the argument goes, they do not cease to emerge in the state-first interpretation simply because there is also a state.

That is, in both the state-first and the dynamics-first theory, histories emerge from the unitary dynamics alone. These histories may include planets, cats, trees, observers who experience what appears to be a classical world, and experimentalists who experience what appear to be definite outcomes of measurement.
Hence the challenge for state-first Everettians is: if the state is not necessary to explain the emergence of classicality, or the appearance of definite outcomes, then what is the justification for postulating it?

\begin{remarkbox}
We note that this argument closely resembles an argument from Brown and Wallace that the de Broglie--Bohm theory \enquote{loses out} to the (state-first) Everett interpretation \cite{brown2005solving}: both the de Broglie--Bohm theory and the (state-first) Everett interpretation are committed to a unitarily evolving wavefunction; the difference between them is simply that the former is also committed to corpuscles. But if the (state-first) Everett interpretation is correct to argue that apparently definite outcomes emerge when the wavefunction is decohered, then they do not cease to emerge in the de Broglie--Bohm theory simply because of the existence of corpuscles. The corpuscles mark out a particular branch of the wavefunction, but the fact that one branch bears this mark does not prevent the others from existing. The de Broglie--Bohm theory is thus every bit as much of a multiverse theory as is the Everett interpretation. (Here Brown and Wallace echo similar claims from Zeh \cite{deutsch1996comment} and Deutsch \cite{zeh1999bohm}.)

Our argument suggests that the wavefunction plays a similar role in the state-first Everett interpretation to the one played by the corpuscles in the de Broglie--Bohm theory according to Brown and Wallace. The universal wavefunction highlights certain branches of the dynamics by attributing them a low probability, and \enquote{lowlights} others by attributing them a low or even zero probability, but it does not thereby prevent any branches of the dynamics from existing. Whereas Brown and Wallace argued that a unitarily evolving state is not prevented from giving rise to a branching multiverse simply because it is made to guide corpuscles, our argument suggests that the unitary dynamics are not prevented from giving rise to a \enquote{dynamics-first} branched multiverse simply because they are made to guide a quantum state.

While assessing the validity of Brown and Wallace's argument is beyond our current scope, we do suggest that \textit{if} their argument did successfully establish that de Broglie--Bohm theory loses out to the state-first Everett interpretation, then our similar argument must also establish that, in turn, the state-first Everett interpretation loses out to DFEI.
\end{remarkbox}

Despite their differences, DFEI and CDT both embrace a similar metaphysical shift. The fundamental ontology includes nothing but dynamics and the histories that emerge from them. This means that the dynamics cannot be thought of in the usual way. Dynamics cease to be a rule governing the evolution of a state that is held to exist a priori. They instead become a thing in itself, out of which states emerge. The causal approach to decoherence suggests not only a shift in perspective on the emergence of classical physics, but on physics more generally: dynamics is not the \textit{governor}, but the \textit{mother}, of states.

\section{Limitations and future work} \label{sec:limitations}

The analysis of decoherence presented here is limited to the case where (1) decoherence is ideal, for example in the sense that no leakage into the environment whatsoever is permitted, and (2) all systems are finite-dimensional. Future work must generalize the analysis beyond this for the causal approach to provide a comprehensive foundation for the theory of decoherence. Indeed, both limitations must be overcome before the causal approach can be applied to the decoherence of coherent states and the derivation of classical Hamiltonian equations of motion. 

One promising strategy for extending the causal approach to approximate decoherence is to begin by quantifying quantum causal influence, perhaps using the operator norm of the commutator. One could then seek to leverage the quantitative notion of causal influence to define approximate accessibility, potential accessibility, and decoherence, and ultimately to establish that the set of approximately decohered operators forms something \enquote{close} to a commutative von Neumann algebra. 

The causal approach to decoherence is likely to help explain the emergence of classical causal models \cite{pearl2009causality} from quantum causal models. It would be interesting to see what insights this causal modelling perspective may bring to the question of how classical physics emerges from quantum physics. 

Although \cref{app:continuous} generalizes the causal approach to the continuous-time case, it does not provide a continuous-time version of the derivation of consistent histories. Although \cref{sec:circuit} explicitly used circuits for the derivation, the role of those circuits was simply to define and motivate the operator algebras associated with systems of interest. This suggests that the derivation of consistent histories might be most naturally generalized to the continuous-time case by taking an algebraic field-theoretic approach, on which the systems of interest are taken to be algebras associated with spacetime regions.

One of the most interesting, but also one of the most ambitious, future directions is to take up the challenge from \cref{subsec:relationalism} of determining whether or in what sense classicality emerges in a relational way, and, in particular, to investigate whether the privileged history sets associated with different sets of systems of interest $\mathfrak{B}$ can be viewed as different descriptions of the same underlying reality. Another future direction---vital for transforming DFEI into a worked-out interpretation rather than a mere interpretational approach---is to derive the Born rule from the unitary dynamics alone. Such a derivation must not presuppose the existence of any quantum state or assume (like CDT) that a unique history is selected from each privileged set, and must therefore start from a nonstandard account of probabilities. Finally, a particularly ambitious future direction is to explore the idea that not only events but also spacetime emerge from a quantum causal structure, possibly leveraging insights from the Hawking-Malament theorem \cite{hawking1976new, malament1977class} and causal set theory (reviewed in \cite{Surya_2019}).

\section{Conclusion} \label{sec:conc}

We started this work by providing a precise account of the process of decoherence in the context of a bipartite unitary interaction between a system and a fragment of its environment.
In particular, taking a quantum Darwinist approach, we identified the classical degrees of freedom picked out by the unitary interaction through the conditions of reproduction and survival---concepts that we formalised in terms of the causal conditions of \emph{accessibility} and \emph{potential accessibility}.
Our focus on the \emph{process} of decoherence, and the absence of quantum states from our fundamental analysis, allowed us to naturally introduce the concept of \emph{dual decoherence} by considering the time-reversed unitary.
Finally, by applying the causal approach to decoherence to the larger context of a circuit of unitary interactions, we saw that it acts as a unifying framework for environmentally induced decoherence and consistent histories, in which states emerge from the unitary dynamics just as outcomes do.

We suggested in \cref{subsec:circuit-summary} that the point of view developed in this work might be called \textit{causal quantum Darwinism}. This term for the approach signals its two most significant influences: quantum causal modelling and quantum Darwinism. The word \enquote{causal} points to the role played by causal influences in pinning down the dynamical structures associated with decoherence independently of their kinematical symptoms. 
But the same word also serves as a reminder of a fact about our position in the world that we have argued is essential to an explanation of the classicality of our experiences: that we are not passive observers, but agents with causal power.
Our experiences, whose classical character we are attempting to explain, are defined not only by what we are able to perceive, but by what we are able to do. 

We have shown that if one articulates in causal terms the quantum Darwinist idea that classical outcomes emerge through the proliferation of information, then, simply by considering the time-reversed perspective, one uncovers a dual version of the story, in which quantum states emerge through the redundant implementation of generators. Decoherence explains the limitations in the observables we are able to measure, while dual decoherence explains the limitations in the states we are able to prepare. But because they are related by time reversal, dual decoherence is not a concept that is independent of decoherence or requires a wholly separate motivation. Indeed, if at the fundamental level the direction of time is an arbitrary convention, then so too is the distinction between decoherence and its dual. In our formalism, the symmetry between decoherence and its dual is only broken once one specifies a prior belief that is most naturally interpreted as a prior belief about states.

The major result of this paper is the derivation of physically meaningful consistent history sets from decoherence, leading to a unified framework that combines the strengths of environmentally induced decoherence and consistent histories. This result crucially relies on the causal quantum Darwinist idea that outcomes and states emerge on an even footing from the unitary dynamics alone. Decoherence and dual decoherence each contribute only half of the events that make up a privileged history, and, as our examples have illustrated, correlations between events in one half only arise by conditioning on events in the other half: it is the interplay between the two types of event that makes the privileged history sets rich enough to describe the classical world.

We conclude this work by placing it in a broader historical context. Zurek has suggested that one reason why it took so long for the importance of decoherence to be recognized was the traditional emphasis on isolated systems in fundamental physics:
\begin{quote}
    The idea that the “openness” of quantum systems might have anything to do with the transition from quantum to classical was ignored for a very long time, probably because in classical physics problems of fundamental importance were always settled in isolated systems. \cite{zurek2003decoherence}
\end{quote}
We suggest that a second, complementary tendency has also been an obstacle to progress in the emergence of classicality: the \textit{state-first} view of physics. On this view, physics describes the world by assigning it a quantum state, while the role of dynamics is merely to transform that state. From this perspective, the existence of a state is not something that needs to be explained. It is unproblematic to assume a priori that a system begins in some particular initial state.

By contrast, on a \textit{dynamics-first} approach, a state only exists if it emerges from the causal structure of the dynamics. This means that \textit{the existence of a state calls for an explanation}. We believe that this demand for explanation facilitates progress by forcing one to ask questions that turn out to be important in the emergence of classicality. On the other hand, by not demanding such explanations, the state-first perspective inhibits progress. Indeed, the main result of this paper provides a concrete instance of this: requiring an explanation of the origin of quantum states leads the dynamics-first approach to the notion of dual decoherence which is essential for the derivation of privileged consistent history sets, whereas a state-first approach has no obvious need for the concept of dual decoherence.

The dynamics-first viewpoint that we are advocating is foreshadowed by the observation, made nearly half a century ago~\cite{zurek1981pointer}, that preferred degrees of freedom can be read directly off the Hamiltonian. Despite this suggestive mathematical insight, the emergent classicality literature since then has tended to view classicality as emerging from a decohering quantum state rather than directly from the underlying dynamical process. Perhaps one reason for this is that, due to the state-first view, physicists have lacked a conceptually motivated and formally well-defined language in which to speak of the dynamics of decoherence independently of its kinematical consequences. 

But this work has shown that with the account of causal influence developed in recent work on quantum causal modelling, we finally have a language for decoherence without the quantum state. The time is ripe to put that language to work in explaining the emergence of the classical world.

\section*{Acknowledgments}
\addcontentsline{toc}{section}{Acknowledgements}  

We are grateful to Emily Adlam and Robert W.\ Spekkens for helpful discussions. We also thank the organisers and participants of the Emergence of Classicality 2024 conference at Trinity College Dublin.

Perimeter Institute is situated on the traditional territory of the Anishinaabe, Haudenosaunee, and Neutral peoples.
It is located on the Haldimand Tract, granted by the British to the Six Nations of the Grand River and the Mississaugas of the Credit First Nation.
Of the 950,000 acres granted to the Haudenosaunee, less than 5 percent remains Six Nations land, and only 6,100 acres remain Mississaugas of the Credit land.

Research at Perimeter Institute is supported in part by the Government of Canada through the Department of Innovation, Science and Economic Development and by the Province of Ontario through the Ministry of Colleges, Universities, Research Excellence and Security.
This project was funded in part within the QuantERA II Programme that has received funding from the European Union’s Horizon 2020 research and innovation programme under Grant Agreement No 101017733. JKK acknowledges the support of the National Science Centre (NCN) through the QuantEra project Qucabose grant no 2023/05/Y/ST2/00139.
YY acknowledges the support of the Natural Sciences and Engineering Research Council of Canada (Grant No. RGPIN-2024-04419).

\phantomsection
\addcontentsline{toc}{section}{References}
\printbibliography

\appendix

\section{On the use of non-Hermitian operators in \cref{sec:channel}} \label{app:hermitian} 

Our discussion of the concepts of accessibility, potential accessibility,
and decoherence often assumes that we are dealing with an operator $M_\bS$ that represents an observable, and must therefore be Hermitian. However, \cref{def:influence,def:influence2,def:accessibility,def:potential_accessibility,def:decoherent} refer to both Hermitian and non-Hermitian operators, and the algebras $\bS_{{\rm acc}\bG}^\cu$, $\bS_{{\rm pacc}\bG}^\cu$, and $\bS_{{\rm dec}\bG}^\cu$ contain both Hermitian and non-Hermitian elements.
Here we explain how the consideration of non-Hermitian operators in our analysis turns out to be an inconsequential mathematical convenience. 

For any algebra $\bA$, each element $M_\bA \in \bA$ can be written as a complex combination of Hermitian operators in the algebra, $M_\bA = \frac{1}{2}(M_\bA + M_\bA^\dagger) - \frac{i}{2}(iM_\bA - iM_\bA^\dagger)$. Hence even though not all members of $\bA$ are Hermitian, $\bA$ is spanned by its Hermitian elements ${\rm Herm}(\bA)$; that is
\begin{equation} \label{eq:alg_span_herm}
    \bA =  {\rm span}_{\mathbb{C}} ({\rm Herm}(\bA)).
\end{equation}
It follows that
\begin{equation} \label{eq:comm_herm_comm}
    \bA' = {\rm Herm}(\bA)'.
\end{equation}

Applying this identity to our accessible algebra,
\begin{equation}
    \begin{split}
        \bS_{{\rm acc}\bG}^\cu & \stackrel{\eqref{eq:accdoublecom}}{=}  ( \cu^{-1}(\bG)' \cap \bS)' \cap \bS \\
        &\stackrel{\eqref{eq:comm_herm_comm}}{=} {\rm Herm}({\rm Herm}(\cu^{-1}(\bG))' \cap \bS)' \cap \bS, \\
    \end{split}
\end{equation}
we find that $M_\bS \in \bS_{{\rm acc}\bG}^\cu $ if and only if every \textit{Hermitian} operator $G_\bS \in \bS $ that does not commute with $M_\bS$ influences some \textit{Hermitian} element of $\bG$. Therefore, if \cref{def:influence,def:influence2} of influence and  \cref{def:accessibility} of accessibility were modified so that all mentioned operators were assumed to be Hermitian, then the resulting set of accessible operators would simply be ${\rm Herm}(\bS_{{\rm acc}\bG}^\cu)$, i.e.\ the set of all of the operators deemed accessible by \cref{def:accessibility} that happen to be Hermitian. By \cref{eq:alg_span_herm}, one can recover $\bS_{{\rm acc}\bG}^\cu$ from ${\rm Herm}(\bS_{{\rm acc}\bG}^\cu)$ by taking the complex linear span. This alternative definition of accessibility is thus effectively equivalent to \cref{def:accessibility}. 

Similarly, if one defined potential accessibility and decoherence purely in terms of Hermitian operators, then one would simply obtain the sets ${\rm Herm}(\bS_{{\rm pacc}\bG}^\cu)$ and ${\rm Herm}(\bS_{{\rm dec}\bG}^\cu)$, whose complex linear spans are the original algebras $\bS_{{\rm pacc}\bG}^\cu$ and $\bS_{{\rm dec}\bG}^\cu$. Hence this entirely Hermitian version of the analysis turns out to be effectively equivalent to the one from \cref{sec:channel} that refers to both Hermitian and non-Hermitian operators. 

Due to this equivalence, it is harmless to give the notions of influence, accessibility, potential accessibility, and decoherence a \textit{physical interpretation} that assumes Hermiticity and a \textit{mathematical analysis} that does not. We have chosen to take this approach for pedagogical reasons.

\section{The predictability sieve} \label{app:sieve}

A classical system can be measured without disturbing its state, and thus without any loss in an observer’s ability to predict its behaviour.
Zurek’s predictability sieve \cite{Zurek:1994zq} uses loss of predictability as a criterion to \enquote{filter out} states that are genuinely quantum, leaving only those that are effectively classical. In this appendix, we explore the connection between the predictability sieve and the causal approach to decoherence in the discrete-time setting, demonstrating that the two approaches agree on when a system basis is privileged due to an interaction with the environment. Since such a basis is not necessarily robust, this reinforces the argument from \cref{sec:symptoms} that robustness is not necessary for decoherence in this setting.

Although \cite{Zurek:1994zq} introduced the sieve in the context of a
Hamiltonian interaction between a system and its environment, we
describe it here in the context of a unitary interaction
$\cu : \bS \bF \rightarrow \bT \bG$, where we assume that
\begin{equation} \label{eq:dim}
    \dim(\bS) = \dim(\bT), \qquad \dim(\bF) = \dim(\bG).
\end{equation}

The \enquote{unpredictability} of a state $\rho$ is quantified by its von Neumann
entropy $S(\rho):= - \Tr (\rho {\rm log} \rho)$, which is nonnegative,
\begin{equation} \label{eq:vne_pos}
    S(\rho) \geq 0,
\end{equation}
and vanishes if and only if $\rho$ is pure:
\begin{equation} \label{eq:vne_pure}
    S(\rho) = 0
    \qquad \Longleftrightarrow \qquad
    \exists\, \ket{\phi} \ \text{such that}\ 
    \rho = \ketbra{\phi}{\phi}.
\end{equation}

Given an initial state $\rho_\bF$ of the environment, the state of the system
is mapped as
\begin{equation}
    \rho_\bS
    \;\longmapsto\;
    \rho_\bT
    = {\rm Tr}_\bG\!\bigl( \cu(\rho_\bS \otimes \rho_\bF) \bigr).
\end{equation}
The loss of predictability of the system is quantified by comparing the
initial and final von Neumann entropies:
\begin{equation}
    {\rm loss}_{\cu,\rho_\bF}(\rho_\bS)
    = S(\rho_\bT) - S(\rho_\bS).
\end{equation}
If
\begin{equation} \label{eq:preferloss}
    {\rm loss}_{\cu,\rho_\bF}(\rho_\bS)
    <
    {\rm loss}_{\cu,\rho_\bF}(\sigma_\bS),
\end{equation}
then $\rho_\bS$ is said to be \emph{preferred over}, i.e.\ more effectively
classical than, $\sigma_\bS$.

We now come to state a theorem about the relationship between potential accessibility and the predictability sieve. It refers to the notion of a \textit{potentially accessible state}, by which we simply mean a density operator $\rho$ contained in the algebra $\bS_{{\rm pacc}\bG}^{\cu}$ from \cref{def:potential_accessibility}. It shows that for pure states, immunity to any loss of predictability is equivalent not to robustness, but to potential accessibility.

\begin{theorem} \label{thm:pacc_sieve}
    The potentially accessible pure states of $\bS$ are precisely those that suffer no loss of predictability for any initial state of the environment: for all $\ketbra{\psi}{\psi}_\bS \in \bS$,
    \begin{equation} \label{eq:pacc_sieve}
        \ketbra{\psi}{\psi}_\bS \in \bS_{{\rm pacc}\bG}^{\cu}
        \qquad \Longleftrightarrow \qquad
        \forall\, \rho_\bF,\ 
        {\rm loss}_{\cu,\rho_\bF}\!\bigl( \ketbra{\psi}{\psi}_\bS \bigr) = 0.
    \end{equation}
\end{theorem}

\begin{proof}
    If $\ketbra{\psi}{\psi}_\bS \in \bS_{{\rm pacc}\bG}^{\cu}$, then by condition \ref{itm:autonomy} of \cref{thm:pacc_ninf}, there exists a state-vector $\ket{\phi}_\bT \in \ch_\bT$ such that 
    \begin{align} \label{eq:phi}
            \cu(\ketbra{\psi}{\psi}_\bS \otimes I_\bF) = \ketbra{\phi}{\phi}_\bT \otimes I_\bG.
    \end{align}
    Hence for any density operator $\rho_\bF \in \bF$,
    \begin{equation}
    \begin{split}
        \cu\!\bigl( \ketbra{\psi}{\psi}_\bS \otimes \rho_\bF \bigr)
        &=
        \cu\!\Bigl(
            (\ketbra{\psi}{\psi}_\bS \otimes I_\bF)
            (\ketbra{\psi}{\psi}_\bS \otimes \rho_\bF)
            (\ketbra{\psi}{\psi}_\bS \otimes I_\bF)
        \Bigr) \\
        &=
        \cu(\ketbra{\psi}{\psi}_\bS \otimes I_\bF)\,
        \cu(\ketbra{\psi}{\psi}_\bS \otimes \rho_\bF)\,
        \cu(\ketbra{\psi}{\psi}_\bS \otimes I_\bF) \\
        &\stackrel{\eqref{eq:phi}  }{=}
        (\ketbra{\phi}{\phi}_\bT \otimes I_\bG)\,
        \cu(\ketbra{\psi}{\psi}_\bS \otimes \rho_\bF)\,
        (\ketbra{\phi}{\phi}_\bT \otimes I_\bG) \\
        &=
        \ketbra{\phi}{\phi}_\bT
        \otimes
        \bra{\phi}_\bT
        \cu(\ketbra{\psi}{\psi}_\bS \otimes \rho_\bF)
        \ket{\phi}_\bT .
    \end{split}
    \end{equation}
    Thus $\rho_\bT = \ketbra{\phi}{\phi}_\bT$ and $S(\rho_\bT)=S(\rho_\bS)=0$. It follows that ${\rm loss}_{\cu,\rho_\bF}(\ketbra{\psi}{\psi}_\bS) = 0$.
    This proves the $\Rightarrow$ direction of
    \cref{eq:pacc_sieve}.
    
    For the $\Leftarrow$ direction, it suffices to consider the case where $\rho_\bF = I_\bF / d$:
    \begin{equation}
    \begin{split}
        {\rm loss}_{\cu, I_\bF\!/\!d} \bigl( \ketbra{\psi}{\psi}_\bS \bigr) = 0 \quad 
        \;\;&\Longleftrightarrow\; \quad 
        S(\rho_\bT) = 0 \\
        &\Longleftrightarrow\; \quad 
        \exists\, \ket{\phi}_\bT: \ 
        \rho_\bT = \ketbra{\phi}{\phi}_\bT \\
        &\stackrel{\eqref{eq:dim}}{\Longleftrightarrow}\; \quad    \exists\, \ket{\phi}_\bT: \  
        \cu\!\bigl( \ketbra{\psi}{\psi}_\bS \otimes I_\bF/d \bigr)
        =
        \ketbra{\phi}{\phi}_\bT \otimes I_\bG/d \\
        &\Longleftrightarrow\; \quad  
        \cu\!\bigl( \ketbra{\psi}{\psi}_\bS \otimes I_\bF/d \bigr) \in \bT \\
        &\stackrel{\ref{thm:pacc_ninf}}{\Longleftrightarrow}\; \quad 
        \ketbra{\psi}{\psi}_\bS \in \bS_{{\rm pacc}\bG}^{\cu}.
        \qedhere
    \end{split}
    \end{equation}
\end{proof}

According to the predictability sieve, preferred states arise not when ${\rm loss}_{\cu, I_\bF/d}(\ketbra{\psi}{\psi}_\bS)$ is minimized for
all pure states, but when it is minimized for some states and not all. There is a unique preferred basis when states on one basis suffer a lower loss in predictability than any states outside the basis. From the causal perspective developed here, this reflects the fact that nontrivial observables are decohered only when some, but not all, observables are left potentially accessible. 

An important consequence of \cref{thm:pacc_sieve} is that our causal approach to decoherence is in perfect agreement with the predictability sieve on when a unique system basis is preferred. Indeed, on the causal approach to decoherence, a basis $\{\ket{\psi_k}_\bS\}_k$ for the Hilbert space of the system is privileged by decoherence precisely when $\bS_{{\rm pacc}\bG}^{\cu}={\rm span}_{\mathbb{C}}\,\{\ketbra{\psi_k}{\psi_k}_\bS\}_k$ (a condition which is equivalent to $\bS_{{\rm dec}\bG}^{\cu}={\rm span}_{\mathbb{C}}\,\{\ketbra{\psi_k}{\psi_k}_\bS\}_k$ by \cref{thm:algebras}). \cref{thm:pacc_sieve} shows that this is the case precisely when $\{\ket{\psi_k}_\bS\}_k$ is the unique system basis that suffers no loss of predictability for any initial state of $\bF$.

We note that the equivalence established by \cref{thm:pacc_sieve} is between potential accessibility and
\emph{perfect} preservation of predictability, whereas the predictability
sieve more generally allows one to regard as preferred all states satisfying
${\rm loss}_{\cu, I_\bF/d}(\ketbra{\psi}{\psi}_\bS) < \epsilon$ for some
$\epsilon > 0$. This serves as a reminder that what this paper is studying is \textit{ideal} decoherence rather than approximate. As discussed in \cref{sec:limitations}, an important future research direction for the causal approach to decoherence is developing a quantitative version of the analysis that permits the treatment of cases where decoherence is only approximate. If such a direction is pursued, it would be interesting to investigate whether a
quantitative definition of potential accessibility permits a generalization of \cref{thm:pacc_sieve} to the case where the preference threshold $\epsilon$ is greater than zero.

\section{Decoherence in the context of a Hamiltonian interaction} \label{app:continuous}

Throughout the main body of this work, our analysis of decoherence takes place in the context of discrete-time evolutions, represented by unitary channels and circuits.
This contrasts with many works on decoherence that focus on continuous-time evolutions generated by a Hamiltonian.
In this appendix, we apply the causal approach to analyse decoherence in a time-independent Hamiltonian $H \in \bS\bF$ acting on a system $\bS$ and a fragment $\bF$ of the environment (which are still assumed finite-dimensional).
As an application of our analysis, we derive the form of maximally decohering Hamiltonians and further clarify the relationship between decoherence and robustness, building on the discussion from \cref{sec:symptoms}.
Finally, we discuss how the situation changes when one describes the evolution in a rotating frame of reference, and form a connection between our approach and the noiseless subsystems of quantum error correction \cite{knill1997theory, lidar2003decoherence, lidar2014review}.

\subsection{Basic definitions}

Consider any Hermitian operator $H \in \bS\bF$.
It generates unitary channels of the form
\begin{equation}\label{eq:u_t}
    \cu_t \coloneqq e^{-iHt}(\cdot)e^{iHt} : \bS\bF \to \bS\bF.
\end{equation}
\Cref{sec:channel} has already discussed decoherence in the context of a general unitary channel $\cu: \bS \bF \rightarrow \bT \bG$;
by identifying $\bS=\bT$ and $\bF = \bG$, we can apply that analysis here and obtain the accessible and potentially accessible algebras for a particular time $t$, denoted by $\bS_{{\rm acc}\bF}^{\cu_t}$ and $\bS_{{\rm pacc}\bF}^{\cu_t}$.

We consider a system observable to be left \emph{potentially accessible} by $H$ to other fragments of the environment precisely if it is left potentially accessible after $\cu_t$ by $\bF$ for all possible times $t$.
Hence we define
\begin{equation}
    \label{eq:cont_pacc}
    \bS_{{\rm pacc}\bF}^H \coloneqq \bigcap_{t} \bS_{{\rm pacc}\bF}^{\cu_t} = \bigcap_{t} \cu_t^{-1}(\bS).
\end{equation}
By \cref{thm:pacc_ninf}, this is precisely the subalgebra of system operators that do not influence $\bF$ for any time $t$.
(It can be shown, by an analyticity argument, that it makes no difference for this definition if we restrict ourselves to positive times $t>0$, or even to an arbitrarily small time interval.)

Moreover, we consider an operator to be made \emph{accessible} to the fragment $\bF$ by $H$ if any local generator that nontrivially transforms it influences $\bF$ through some (or many) $\cu_t$. 
Just as in the discrete-time case (\cref{thm:algebras}), this makes the accessible algebra into the commutant of the non-influencing algebra: 
\begin{equation} \label{eq:cont_acc}
        \bS_{{\rm acc}\bF}^H
        \coloneqq (\bS_{{\rm pacc}\bF}^H)' \cap \bS
        = \bigvee_{t} \bS_{{\rm acc}\bF}^{\cu_t}, 
\end{equation}
where $\bigvee_t \bX_t$ denotes the algebra generated by all $\bX_t$.

Analogously to \cref{def:decoherent} for discrete-time unitaries, the observables \emph{decohered} by $H$ are those in the commutative subalgebra
\begin{equation}
     \bS_{{\rm dec}\bF}^H \coloneqq \bS_{{\rm acc}\bF}^H \cap \bS_{{\rm pacc}\bF}^H.
\end{equation}
Furthermore, we will call $H$ \emph{decohering} if $\bS_{{\rm dec}\bF}^H$ contains a nontrivial observable and \emph{maximally decohering} if it contains a nondegenerate observable.
The decohered algebra picks out preferred subspaces $\ch_\bS^i \subseteq \ch_\bS$ as in \cref{eq:preferred-decomposition}, and the condition of maximal decoherence is equivalent to each of these subspaces being one-dimensional.

\subsection{The form of maximally decohering Hamiltonians and the role of robustness}

We now state a theorem that relates decoherence to robustness and derives the form of maximally decohering Hamiltonians.

\begin{theorem} \label{thm:invariance-and-hamiltonian-control}
    In the continuous case, an observable is decohered if and only if it is both accessible and robust, i.e.\ 
    \begin{equation} \label{eq:dec_rob}
        \bS_{{\rm dec}\bF}^H = \bS_{{\rm acc}\bF}^H \cap \bS_{{\rm rob}}^H,
    \end{equation}
    where $\bS_{{\rm rob}}^H := \{H\}' \cap \bS = \{M_\bS\in \bS \mid \forall t,\, \cu_t(M_\bS) = M_\bS\}$ is the algebra of robust observables.

    As a consequence, $H$ is maximally decohering if and only if it has the \enquote{coherent} control form 
    \begin{equation} \label{eq:hamiltonian_control} 
        H= \sum_k \ketbra{\psi_k}{\psi_k}_\bS \otimes H_\bF^{(k)}
    \end{equation}
    where $\{\ket{ \psi_k}_\bS \}_{k} \subseteq \ch_\bS$ is an orthonormal basis and $H_\bF^{(k)}$ are Hermitian operators that are neither equal nor equivalent up to the addition of a multiple of the identity, i.e.\ $H_\bF^{(k)} - H_\bF^{(l)} \not\propto I_\bF$ for all $k\neq l$. Moreover, when $H$ is maximally decohering the control basis $\{\ket{ \psi_k}_\bS \}_{k}$ is unique up to phase.
\end{theorem}

\begin{proof}[Proof]
    By \cref{eq:cont_pacc}, an operator is potentially accessible, i.e.\   $M_\bS \in \bS_{{\rm pacc}\bF}^H$, if and only if every operator to which it is transformed is potentially accessible, i.e.\  $\cu_t(M_\bS) \in \bS_{{\rm pacc}\bF}^H$ for all $t$. Hence a potentially accessible $M_\bS$ has the property that for any $t$, $\cu_t(M_\bS)$ commutes with all accessible operators, i.e.\ $\cu_t(M_\bS) \in (\bS_{{\rm acc}\bF}^H)'$. 
 
    Therefore, if $M_\bS$ is both potentially accessible and accessible itself, then it commutes with every operator to which it is transformed, i.e.\ $[M_\bS, \cu_t(M_\bS) ]=0$ for all $t$. 
 
    But for a finite-dimensional system $\bS$ and any observable $M_\bS$, if $[M_\bS, \cu_t(M_\bS)]=0$ for all $t$ then in fact $\cu_t(M_\bS) = M_\bS$ for all $t$, meaning that $M_\bS$ is robust. It follows that all decohered observables are robust, $\bS_{{\rm dec}\bF}^H \subseteq  \bS_{{\rm rob}\bF}^H$. 

    Since decohered observables are also accessible, it follows that $\bS_{{\rm dec}\bF}^H \subseteq  \bS_{{\rm acc}\bF}^H \cap \bS_{{\rm rob}\bF}^H$. But since $\bS_{{\rm rob}\bF}^H \subseteq \bS_{{\rm pacc}\bF}^H$, we also have  that $\bS_{{\rm acc}\bF}^H \cap \bS_{{\rm rob}\bF}^H \subseteq \bS_{{\rm dec}\bF}^H$. This proves \cref{eq:dec_rob}.

    Therefore, $H$ is maximally decohering if and only if there exists a basis $\{\ket{ \psi_k}_\bS \}_{k}$ such that $\bS_{{\rm rob}\bF}^H = {\rm span}_\mathbb C \{\ket{ \psi_k}\bra{\psi_k}_\bS \}$ and $\bS_{{\rm acc}\bF}^H = {\rm span}_\mathbb C \{\ket{ \psi_k}\bra{\psi_k}_\bS\}$. The first condition holds if and only if $H = \sum_k \ket{ \psi_k}\bra{\psi_k}_\bS \otimes H_\bF^{(k)}$ for some set of Hermitian operators $H_\bF^{(k)}$. For an $H$ of this form, the second condition, $\bS_{{\rm acc}\bF}^H = {\rm span}_\mathbb C \{\ket{ \psi_k}\bra{\psi_k}_\bS\}$, holds if and only if $H_\bF^{(k)} - H_\bF^{(l)} \not\propto I_\bF$ for all $k\neq l$. (The argument for this last claim closely resembles the one from \cref{thm:uni_control} that established the condition $V_{\bF \rightarrow \bG}^{(k)} \not\propto V_{\bF \rightarrow \bG}^{(l)} $ for $k \neq l$, so we do not give it here explicitly.)
\end{proof}

We note that the proof above relies on a fact about finite-dimensional observables that does not extend to the infinite-dimensional case: that if an observable is nontrivially transformed at all by a Hamiltonian then it is at some point transformed to an observable with which it does not commute. The problem of generalizing \cref{thm:invariance-and-hamiltonian-control} to infinite dimensions is left open for future work.

\begin{examplebox}
    A good sanity check on this causal analysis of decoherence through Hamiltonians is to confirm that a von Neumann measurement of a nondegenerate observable $M_\bS$, implemented using a Hamiltonian of the form
    \begin{equation} \label{eq:vnm}
        H = M_\bS \otimes N_\bF,
    \end{equation}
    is maximally decohering.
    Indeed, since $M_\bS$ is nondegenerate, the Hamiltonian can be rewritten as
    \begin{equation}
        H = \sum_k \ketbra{\psi_k}{\psi_k}_\bS \otimes m_k N_\bF,
    \end{equation}
    where all $m_k$ are distinct. 
    It then follows from \cref{thm:invariance-and-hamiltonian-control} above that $H$ is indeed maximally decohering as long as $N_\bF$ is not proportional to the identity.
\end{examplebox}

\cref{thm:invariance-and-hamiltonian-control}, and in particular \cref{eq:hamiltonian_control}, directly connects the causal approach to many concrete interactions that are commonly studied in the decoherence literature. These include spin-spin and spin-boson models \cite{zurek1981pointer, zurek1982environment, leggett1987dynamics}, quantum Brownian motion \cite{caldeira1983path, hu1992quantum}, collisional decoherence \cite{joos1985emergence, hornberger2003collisional}, optomechanical and superconducting systems \cite{aspelmeyer2014cavity, clerk2010introduction}, and models of nitrogen-vacancy centers in diamond \cite{doherty2013nitrogen}, among many others. (We note however that some of these references include interactions that only approximately satisfy \cref{eq:hamiltonian_control}, and others that (approximately) satisfy an obvious infinite-dimensional generalization of \cref{eq:hamiltonian_control}.)

 Just as \cref{thm:uni_control} showed in the discrete-time case that off-diagonal suppression is a symptom rather than a defining feature of decoherence, \cref{thm:invariance-and-hamiltonian-control} shows the same in the continuous-time setting. Indeed, in the case of a maximally decohering Hamiltonian acting on an initial product state $\rho_\bS\otimes \rho_\bF$, the off-diagonal elements of $\rho_\bS$ are suppressed by a factor of $|{\rm Tr}(e^{-iH^{(k)}_{\bF}t} \rho_\bF e^{iH^{(l)}_{\bF}t})| \leq 1$. (As remarked below \cref{thm:uni_control}, this implies strict suppression whenever $\rho_\bF$ is full-rank.)

We now discuss the implications of \cref{thm:invariance-and-hamiltonian-control} for the role of robustness in the emergence of classicality. We argued in \cref{sec:symptoms} that in the discrete-time case robustness was not necessary for decoherence. We said that an observable was decohered if and only if it was both made accessible and left potentially accessible, and we saw that not all such observables are robust.

But \cref{thm:invariance-and-hamiltonian-control} shows that in the finite-dimensional continuous-time case, robustness is more closely related to decoherence. In this setting, the conjunction of accessibility and potential accessibility turns out to be equivalent to the conjunction of accessibility and robustness. Hence as long as we are working only in this setting, decoherence can equivalently be \textit{defined} as the conjunction of accessibility and robustness.

This explains why, in the finite-dimensional continuous-time setting, robustness does in fact serve as a useful criterion for identifying preferred degrees of freedom, despite the arguments from \cref{sec:channel} that potential accessibility is what really matters a priori.
Indeed, in the case of maximal decoherence, \cref{thm:invariance-and-hamiltonian-control} implies that the preferred basis is the unique system basis (up to phase) that is perfectly robust. However, it should be remembered that this is not true in the discrete-time setting, wherein robustness is not necessary for decoherence.

In existing works on emergent classicality, it is often only required that preferred degrees of freedom are \textit{approximately} robust. But \cref{thm:invariance-and-hamiltonian-control} shows that on our analysis, preferred degrees of freedom are perfectly robust. This serves as a reminder that in this work we are only studying \textit{ideal} decoherence, on which we will comment more in the next section.

\subsection{Robustness in a rotating frame and noiseless subsystems} \label{app:noiseless}

If one describes the Hamiltonian evolution of operators using a frame of reference that is itself rotating (as one does in the interaction picture), this changes which operators appear to be robust, i.e.\ invariant under the evolution. While not all operators that are potentially accessible are robust relative to the usual fixed reference frame, this subsection will show that  an operator is left potentially accessible if and only if it is robust relative to a rotating frame that fully absorbs the system’s local Hamiltonian. We will then use the result to form a link between potential accessibility and \textit{noiseless subsystems}, a key concept from quantum error correction \cite{knill1997theory, lidar2003decoherence, lidar2014review}.

Let us begin by explaining how the evolution of operators is described relative to the rotating frame. As above, we consider a time-independent Hamiltonian $H \in \bS\bF$ acting on a system $\bS$ and a fragment $\bF$ of the environment. The Hamiltonian can always be
decomposed as
\begin{equation} \label{eq:ham_int}
    H = H_\bS + H_I ,
\end{equation}
such that 
\begin{equation} \label{eq:partial_traceless}
    {\rm Tr}_\bF H_I = 0 .
\end{equation}
This condition ensures that $H_\bS$ fully captures the local evolution on $\bS$ while $H_I$ represents both $\bF$'s local evolution and the interaction between the two systems. 
In terms of the standard Heisenberg-evolved operator $M_\bS^H(t):=\cu^{-1}_t(M_\bS) = e^{iHt} M_\bS e^{-iHt}$, the operator
$M_\bS^I(t)$ seen in the rotating frame is defined by 
\begin{equation} \label{eq:int_op}
    M_\bS^I(t) = e^{-iH_\bS t} M_\bS^H(t) e^{iH_\bS t},
\end{equation}
and satisfies the equation of motion
\begin{equation} \label{eq:int_eom}
    \dot M_\bS^I(t) = i[H_I(t), M_\bS^I(t)] ,
\end{equation}
where
\begin{equation} \label{eq:evolved_ham}
    H_I(t) = e^{-iH_\bS t} H_I e^{iH_\bS t} .
\end{equation}
It will prove useful to note that, by \cref{eq:partial_traceless,eq:evolved_ham},
\begin{equation} \label{eq:partial_traceless_2}
    {\rm Tr}_\bF H_I(t) = 0
\end{equation}
for all $t$.

From the point of view of this rotating frame, the robust operators are those in the algebra
\begin{equation}
\begin{split}
        \bS_{\rm rob}^{H_I} &\coloneqq \{M_\bS \in \bS \mid \forall  t \in \mathbb R, \ M^I_\bS(t) = M_\bS \},
        \\ 
        &\phantom{:}=   \{ H_I(t) \mid t \in \mathbb R\}' \cap \bS .
\end{split}
\end{equation}

We are now ready to state the result. 

\begin{theorem} \label{thm:pacc_noiseless}
    Potential accessibility is equivalent to robustness in the rotating frame, i.e.\ 
    \begin{equation}
   \bS_{{\rm pacc}\bF}^H = \bS_{{\rm rob}}^{H_I}.
\end{equation}
\end{theorem}

\begin{proof}[Proof] 
An operator $M_\bS \in \bS$ does not leak into the environment in the usual nonrotating frame if and only if it does not leak into the environment in the rotating frame. Indeed, it is immediate from \cref{eq:int_op} that
\begin{equation}
    \forall t, \ \  M_\bS^H(t) \in \bS \qquad \qquad \Longleftrightarrow \qquad \qquad  \forall t,  \ \
    M_\bS^I(t) \in \bS.
\end{equation}

We will now show that in the rotating frame, $M_\bS \in \bS$ does not leak into the environment if and only if it is robust, i.e.\ that
\begin{equation} \label{eq:rot_aut_inv}
  \forall t, \ \ M_\bS^I(t) \in \bS \qquad \qquad  \Longleftrightarrow \qquad \qquad \forall t, \ \ M_\bS^I(t) = M_\bS.
\end{equation}
The $\Longleftarrow$ direction holds because $M_\bS \in \bS$. For the other direction, note that the condition on the left entails that $\dot M_\bS^I(t)\in \bS$ for all $t$, and hence, by \cref{eq:int_eom}, that for all $t$ there exists an $N^I_\bS(t) \in \bS$ such that
\begin{equation}
       [M_\bS^I(t), H_I(t)] = N^I_\bS(t) \otimes I_\bF.
    \end{equation}
Partial tracing over $\bF$ on both sides and applying \cref{eq:partial_traceless_2} shows  that $[M_\bS^I(t), H_I(t)] =0$. By \cref{eq:int_eom}, it follows that $M_\bS$ is robust in the rotating frame.

Combining the two equivalences gives 
\begin{equation}
    \forall t, \ \  M_\bS^H(t) \in \bS \qquad \qquad \Longleftrightarrow \qquad \qquad \forall t, \ \  M_\bS^I(t) = M_\bS.
\end{equation}
The condition on the left is equivalent to $M_\bS \in \bS_{{\rm pacc}\bF}^H$ by \cref{thm:pacc_ninf}, while the condition on the right is equivalent to $M_\bS \in \bS_{{\rm rob}}^{H_I}$ by definition.
\end{proof}

We note that $\bS_{{\rm pacc}\bF}^H$ is not in general equal to the algebra $\bS_{\rm rob}^H \coloneqq  \{H\}' \cap \bS$ of operators that are robust relative to the usual fixed frame. 
Indeed, the algebras are always distinct when the Hamiltonian acts nontrivially on $\bS$ but not $\bF$: if $H = H_\bS$, then $\bS_{{\rm pacc}\bF}^H = \bS$ but $\bS_{\rm rob}^H = \langle H_\bS \rangle$. But when such a Hamiltonian is seen in the rotating frame, it becomes trivial (i.e.\ $H_I(t) = 0$ for all $t$) and thus $\bS_{{\rm pacc}\bF}^H = \bS_{\rm rob}^{H_I} =\bS$.

\cref{thm:pacc_noiseless} establishes a connection between decoherence as we define it and the \textit{noiseless subsystems}  (introduced in \cite{knill1997theory}, reviewed in \cite{lidar2003decoherence, lidar2014review}) of quantum error correction. A noiseless subsystem of $\bS$ is one that is protected from noise induced by the environment. It is formally represented by a Hilbert space factor associated with a robust algebra of operators, where robustness may be judged using a frame of reference appropriate to the problem at hand. If we use the rotating frame that fully absorbs the system's local evolution (as suggested in \cite{knill1997theory}), then the noiseless subsystems are identified with the Hilbert factors appearing in the Artin--Wedderburn representation of $\bS_{\rm rob}^{H_I}$---that is, the factors $\ch_{\bS_R^i}$ appearing in the decomposition $\ch = \bigoplus_{i=1}^n \ch_{\bS_L^i} \otimes \ch_{\bS_R^i}$ chosen such that $\bS_{{\rm rob}}^{H_I}= \bigoplus_{i=1}^n I_{\bS_L^i} \otimes \cl(\ch_{\bS_R^i})$. \cref{thm:pacc_noiseless} shows that these noiseless subsystems are equivalently the Hilbert space factors associated with the algebra $\bS_{{\rm pacc} \bF}^H$, thus demonstrating an equivalence between noiselessness and potential accessibility.

\section{Causal Decoherence Theory} \label{app:bubble}

This appendix briefly introduces the interpretation of quantum theory laid out in \cite{ormrod2024quantum}, minorly adapting the notation for the sake of consistency with the current article. We will also show how this interpretation can naturally be reformulated in terms of our causal approach to decoherence, leading to a simplification of its formalism and clarification of its message.

The interpretation in \cite{ormrod2024quantum} claims that reality consists of a unitary circuit and, for each subset $\mathfrak{B}$ of systems, a stochastically selected history from a preferred consistent history set associated with $\mathfrak{B}$. Given this description, the interpretation already sounds very similar to the interpretation that we called Causal Decoherence Theory in \cref{sec:crqt}. Indeed, all that has to be done to demonstrate that the interpretations are equivalent is to show that the preferred consistent history sets from \cite{ormrod2024quantum} are precisely the ones that we have derived from decoherence and dual decoherence. 

We start by introducing some concepts from \cite{ormrod2024quantum} that are foundational to its notion of a preferred consistent history set. Given a system $\bS$, a \textit{projective decomposition} $\{P^i_\bS\}$ is a complete orthogonal family of projectors, i.e.\ a set of projectors that each projects onto some subspace in a direct-sum decomposition $\ch_\bS = \bigoplus_i \ch^i_\bS$ of the Hilbert space of the system. Given a unitary interaction $\cu: \bS \bF \rightarrow \bT \bG$, \cite{ormrod2024quantum} says that a projective decomposition $\{P_\bS^i\}$ exerts an \textit{interference influence} on a projective decomposition $\{P_\bG^i\}$, written $\{P_\bS^i\} \xrightarrow \cu \{P_\bG^j\}$, if and only if at least one $P_\bS^i$ influences at least one $P_\bG^j$ in the sense of our \cref{def:influence}, i.e.\ $\exists i, j: \ P_\bS^i \xrightarrow \cu P_\bG^j$. \cite{ormrod2024quantum} shows that there is a unique projective decomposition $\{P_\bS^i\}$ that satisfies all of the following conditions:
\begin{enumerate}
        \item $\{P_\bS^i\}$ does not exert an interference influence on any projective decomposition on $\bG$: $\forall \{P_\bG^j\}: \{P_\bS^i\} \ninf \cu \{P_\bG^j\}$.
    \item If $\{P_\bS^i\}$  is incompatible with some other projective decomposition $\{Q_\bS^k\}$, then $\{Q_\bS^k\}$ exerts an interference influence on at least one projective decomposition on $\bG$: $\exists i, k: [P_\bS^i, Q_\bS^k]\neq0 \implies \{Q_\bS^k\} \xrightarrow \cu\{P_\bG^j\} $.
    \item Any other decomposition $\{R_\bS^l\}$ satisfying (1) and (2) is a coarse-graining of $\{P_\bS^i\}$, in the sense that $\{R_\bS^l\} \subseteq {\rm span}_\mathbb{C} (\{P_\bS^i\})$.
\end{enumerate}
The unique projective decomposition that satisfies all three conditions is called \textit{preferred}. We denote it by $\cs_{{\rm pref}\bG}^\cu$, using the calligraphic symbol $\cs$ to avoid confusing this projective decomposition with an algebra (which would be represented using a boldface symbol $\bS$).

Making use of this notion of a preferred projective decomposition, \cite{ormrod2024quantum} introduces a procedure for finding a preferred consistent history set given any subset $\mathfrak{B}$ of systems in a unitary circuit. Although we will ultimately prove that it yields the same consistent history set as the procedure we introduced in \cref{sec:circuit}, the procedure itself is different. In \cite{ormrod2024quantum} the first step is to define a single unitary channel by \enquote{breaking the wires} corresponding to each system in $\mathfrak{B}$. For example, given the unitary circuit on the left side of \cref{fig:circuit_three} and the set $\mathfrak{B} = \{\bA, \bB, \bC\}$, one obtains the broken circuit on the right side of the figure. This broken unitary circuit defines a unitary channel whose outputs include the wires that go into the cuts and whose inputs include wires that come out of the cuts. Explicitly, the unitary channel associated with the right of \cref{fig:circuit_three}  is of the type
\begin{equation}
    \cv: \quad  \bA^{\rm out} \bB^{\rm out} \bC^{\rm out} \bP \quad \rightarrow \quad \bA^{\rm in} \bB^{\rm in} \bC^{\rm in} \bF,
\end{equation}
where $\bP$ corresponds to the bottom-right wire in the broken circuit and $\bF$ corresponds to the top-right wire. 

\begin{figure}
    \centering
    \begin{tikzpicture}
	\begin{pgfonlayer}{nodelayer}
		\node [style=none] (0) at (0.5, -2.75) {};
		\node [style=none] (1) at (5.5, -2.75) {};
		\node [style=none] (2) at (0.5, -1.75) {};
		\node [style=none] (3) at (5.5, -1.75) {};
		\node [style=none] (4) at (3, -2.25) {$\cv_1$};
		\node [style=none] (5) at (1, -5.75) {};
		\node [style=none] (6) at (1, -2.75) {};
		\node [style=none] (7) at (1, -1.75) {};
		\node [style=none] (8) at (1, 1.75) {};
		\node [style=none] (9) at (5, -5.75) {};
		\node [style=none] (10) at (5, -2.75) {};
		\node [style=none] (11) at (5, -1.75) {};
		\node [style=none] (12) at (5, 1.75) {};
		\node [style=none] (13) at (1.5, -4.25) {$\bA$};
		\node [style=none] (14) at (1.5, 0) {$\bB$};
		\node [style=none] (15) at (0.5, 1.75) {};
		\node [style=none] (16) at (5.5, 1.75) {};
		\node [style=none] (17) at (0.5, 2.75) {};
		\node [style=none] (18) at (5.5, 2.75) {};
		\node [style=none] (19) at (3, 2.25) {$\cv_2$};
		\node [style=none] (20) at (1, 2.75) {};
		\node [style=none] (21) at (1, 6) {};
		\node [style=none] (22) at (5, 2.75) {};
		\node [style=none] (23) at (5, 6) {};
		\node [style=none] (24) at (1.5, 4.5) {$\bC$};
		\node [style=none] (25) at (13.5, -2.75) {};
		\node [style=none] (26) at (18.5, -2.75) {};
		\node [style=none] (27) at (13.5, -1.75) {};
		\node [style=none] (28) at (18.5, -1.75) {};
		\node [style=none] (29) at (16, -2.25) {$\cv_1$};
		\node [style=none] (30) at (14, -4) {};
		\node [style=none] (31) at (14, -2.75) {};
		\node [style=none] (32) at (14, -1.75) {};
		\node [style=none] (33) at (14, -0.5) {};
		\node [style=none] (34) at (18, -6) {};
		\node [style=none] (35) at (18, -2.75) {};
		\node [style=none] (36) at (18, -1.75) {};
		\node [style=none] (37) at (18, 1.75) {};
		\node [style=none] (38) at (14.83, -5.25) {$\bA^{\rm in}$};
		\node [style=none] (39) at (13.5, 1.75) {};
		\node [style=none] (40) at (18.5, 1.75) {};
		\node [style=none] (41) at (13.5, 2.75) {};
		\node [style=none] (42) at (18.5, 2.75) {};
		\node [style=none] (43) at (16, 2.25) {$\cv_2$};
		\node [style=none] (44) at (14, 2.75) {};
		\node [style=none] (45) at (14, 4) {};
		\node [style=none] (46) at (18, 2.75) {};
		\node [style=none] (47) at (18, 6) {};
		\node [style=none] (48) at (14, -5) {};
		\node [style=none] (49) at (14, -6) {};
		\node [style=none] (50) at (14, 0.5) {};
		\node [style=none] (51) at (14, 1.75) {};
		\node [style=none] (52) at (14, 5) {};
		\node [style=none] (53) at (14, 6) {};
		\node [style=none] (54) at (15, -3.75) {$\bA^{\rm out}$};
		\node [style=none] (55) at (14.83, -0.75) {$\bB^{\rm in}$};
		\node [style=none] (56) at (15, 0.75) {$\bB^{\rm out}$};
		\node [style=none] (57) at (14.83, 3.75) {$\bC^{\rm in}$};
		\node [style=none] (58) at (15, 5.25) {$\bC^{\rm out}$};
		\node [style=none] (59) at (9.5, 0) { \Large $ \mapsto$};
	\end{pgfonlayer}
	\begin{pgfonlayer}{edgelayer}
		\draw (3.center) to (1.center);
		\draw (1.center) to (0.center);
		\draw (0.center) to (2.center);
		\draw (2.center) to (3.center);
		\draw (6.center) to (5.center);
		\draw (8.center) to (7.center);
		\draw (10.center) to (9.center);
		\draw (12.center) to (11.center);
		\draw (18.center) to (16.center);
		\draw (16.center) to (15.center);
		\draw (15.center) to (17.center);
		\draw (17.center) to (18.center);
		\draw (21.center) to (20.center);
		\draw (23.center) to (22.center);
		\draw (28.center) to (26.center);
		\draw (26.center) to (25.center);
		\draw (25.center) to (27.center);
		\draw (27.center) to (28.center);
		\draw (31.center) to (30.center);
		\draw (33.center) to (32.center);
		\draw (35.center) to (34.center);
		\draw (37.center) to (36.center);
		\draw (42.center) to (40.center);
		\draw (40.center) to (39.center);
		\draw (39.center) to (41.center);
		\draw (41.center) to (42.center);
		\draw (45.center) to (44.center);
		\draw (47.center) to (46.center);
		\draw (49.center) to (48.center);
		\draw (51.center) to (50.center);
		\draw (53.center) to (52.center);
	\end{pgfonlayer}
\end{tikzpicture}
    \caption{In \cite{ormrod2024quantum}, the first step to obtaining the preferred decompositions is to break the wires corresponding to all systems of interest, thus obtaining a unitary channel $\cv:  \bA^{\rm out} \bB^{\rm out} \bC^{\rm out} \bP  \rightarrow  \bA^{\rm in} \bB^{\rm in} \bC^{\rm in} \bF$, where $\bP$ (for \enquote{past}) is the bottom-right input to the circuit and $\bF$ (for \enquote{future}) is the top-right output. The labels $\bP$ and $\bF$ are omitted from the diagram because they are not being considered as systems of interest.}
    \label{fig:circuit_three}
\end{figure}
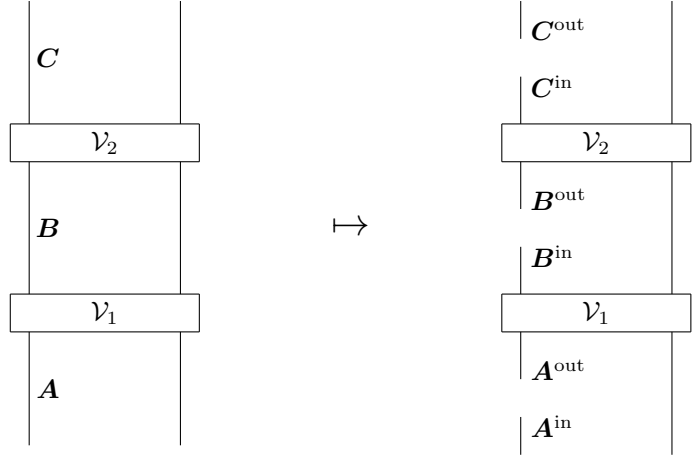

Having defined this unitary channel, one then applies the definition above to obtain the projective decomposition on each $\bX^{\rm out}$ that is preferred by the combination of all $\bY^{\rm in}$, which, in the present example, is $\bA^{\rm out} \bB^{\rm out} \bC^{\rm out}$. In order to avoid introducing a time asymmetry into the theory, one also does the analogous thing with the inverse $\cv^{-1}$, finding for each $\bX^{\rm in}$ the  projective decomposition preferred by the combination of all $\bY^{\rm out}$. In the example of \cref{fig:circuit_three}, one obtains $6 = 2 \times 3$ projective decompositions, denoted 
\begin{equation} \label{eq:bubble_decs}
\begin{array}{l@{\hspace{10em}}l}
(\ca^{\rm out})_{{\rm pref}\mathfrak{B}^{\rm in}}^\cv
&
(\ca^{\rm in})_{{\rm pref}\mathfrak{B}^{\rm out}}^{\cv^{-1}}
\\[0.6ex]
(\cb^{\rm out})_{{\rm pref}\mathfrak{B}^{\rm in}}^\cv
&
(\cb^{\rm in})_{{\rm pref}\mathfrak{B}^{\rm out}}^{\cv^{-1}}
\\[0.6ex]
(\cc^{\rm out})_{{\rm pref}\mathfrak{B}^{\rm in}}^\cv
&
(\cc^{\rm in})_{{\rm pref}\mathfrak{B}^{\rm out}}^{\cv^{-1}}.
\end{array}
\end{equation}
where $\mathfrak{B}^{\rm out}:= \bA^{\rm out}\bB^{\rm out}\bC^{\rm out}$ and $\mathfrak{B}^{\rm in} :=\bA^{\rm in}\bB^{\rm in}\bC^{\rm in}$. 

Since $\bX^{\rm out}$ and  $\bX^{\rm in}$ are each a copy of the original system $\bX$, both of the preferred projective decompositions $(\cx^{\rm out})_{{\rm pref}\mathfrak{B}^{\rm in}}^\cv$ and $(\cx^{\rm in})_{{\rm pref}\mathfrak{B}^{\rm out}}^{\cv^{-1}}$ naturally correspond to projective decompositions on $\bX$ itself. We denote these by $(\cx)_{{\rm pref}\mathfrak{B}}^\uparrow$ and $(\cx)_{{\rm pref}\mathfrak{B}}^\downarrow$ respectively. It is shown in \cite{ormrod2024quantum} that for any unitary circuit and any $\mathfrak{B}$, these projective decompositions define a consistent history set (with $\rho = I/d$), and that the probability rule takes the form of \cref{eq:prob}.

To show that this is in fact precisely the consistent history set that is selected by decoherence and dual decoherence, the first step is to show that in the simple case of a unitary channel $\cu: \bS  \bF \rightarrow \bT\bG$, the preferred projective decomposition $\cs_{{\rm pref}\bG}^\cu$ of  \cite{ormrod2024quantum} is precisely the set of projectors onto the subspaces that are selected by decoherence.
Appendix C of \cite{ormrod2024quantum} proves that $\cs_{{\rm pref}\bG}^\cu$ is the unique projective decomposition that satisfies 
\begin{equation}
    {\rm span}_\mathbb{C}(\cs_{{\rm pref}\bG}^\cu) = Z(\cu^{-1}(\bG)' \cap \bS).
\end{equation}
But this paper (\cref{thm:algebras}) has shown that
\begin{equation}
    \bS_{{\rm dec} \bG}^\cu = Z(\cu^{-1}(\bG)' \cap \bS).
\end{equation}
It follows that the preferred projective decomposition of \cite{ormrod2024quantum} is the unique one that spans the decoherent algebra, i.e.\ 
\begin{equation} \label{eq:pref_dec}
    {\rm span}_\mathbb{C}(\mathcal{S  }_{{\rm pref}\bG}^\cu) = \bS_{{\rm dec} \bG}^\cu,
\end{equation}
and thus that each element of $\cs_{{\rm pref}\bG}^\cu$ projects onto one of the subspaces from the Hilbert space decomposition induced by $\bS_{{\rm dec} \bG}^\cu$.

Having demonstrated an equivalence in the case of a single unitary between the preferred projective decomposition and our decohered algebra, we now show that in the setting of unitary circuits, the preferred consistent history sets of \cite{ormrod2024quantum} are precisely the ones privileged by decoherence and dual decoherence. We begin with the example of the circuit on the left of \cref{fig:circuit_three}. By \cref{thm:pacc_ninf},
\begin{equation}
     M_\bA \in \bA_{{\rm pacc}\mathfrak{B}}^\uparrow \qquad \Longleftrightarrow \qquad M_\bA \ninf {\cv_1} \bB \ \  {\rm and} \ \  \cv_1(M_\bA) \ninf {\cv_2} \bC.
\end{equation}
Note that $M_\bA \ninf {\cv_1} \bB$ implies that $\cv_1(M_\bA)$ acts locally on the (unlabelled) right-hand output of $\cv_1$, i.e.\ $\cv_1(M_\bA) \in \bB'$. Hence, by the circuit form of $\cv$, \cref{thm:pacc_ninf} also implies that
\begin{equation}
M_{\bA^{\rm out}}  \in (\bA^{\rm out})^\cv_{{\rm pacc} \mathfrak{B}^{\rm in}}\qquad \Longleftrightarrow \qquad M_{\bA^{\rm out}} \ninf{\cv_1} \bB \ \ {\rm and} \ \ \cv_1(M_{\bA^{\rm out}}) \ninf{\cv_2} \bC.
\end{equation}
It follows that  $\bA_{{\rm pacc}\mathfrak{B}}^\uparrow$  and $(\bA^{\rm out})^\cv_{{\rm pacc} \mathfrak{B}^{\rm in}}$ are equivalent in the sense that
\begin{equation}
    M_\bA \in \bA_{{\rm pacc}\mathfrak{B}}^\uparrow  \qquad \Longleftrightarrow \qquad M_{\bA^{\rm out}}    \in (\bA^{\rm out})^\cu_{{\rm pacc} \mathfrak{B}^{\rm in}}.
\end{equation}
Finally, \cref{eq:pref_dec} then implies that
\begin{equation}
    \bA_{{\rm dec}\mathfrak{B}}^\uparrow  = {\rm span}_\mathbb C (\ca_{{\rm pref} \mathfrak{B}}^\uparrow).
\end{equation}
Generalizing this argument, we find that for all three systems of interest $\bX \in \mathfrak{B}$ in this example,
\begin{equation} \label{eq:gen_match}
    \bX_{{\rm dec}\mathfrak{B}}^\uparrow  = {\rm span}_\mathbb C (\cx_{{\rm pref} \mathfrak{B}}^\uparrow), \qquad \qquad \bX_{{\rm dec}\mathfrak{B}}^\downarrow  = {\rm span}_\mathbb C (\cx_{{\rm pref} \mathfrak{B}}^\downarrow).
\end{equation}

To generalize the argument beyond this example, note that for any set $\mathfrak{B}$ of systems of interest in any unitary circuit, one can deform the circuit without changing its topology so that it resembles the left of \cref{fig:circuit_three}, in which each system of interest occupies a different time slice of the circuit, and each neighbouring pair of time slices is related by some unitary. Hence an obvious generalization of the arguments above establishes \cref{eq:gen_match} in this more general setting.

Therefore, for any unitary circuit and any $\mathfrak{B}$, the preferred consistent history set is precisely the one privileged by decoherence as we define it. It follows that the interpretation introduced in \cite{ormrod2024quantum} is equivalent to Causal Decoherence Theory.

\end{document}